\documentclass[10pt]{article}
\usepackage[applemac]{inputenc} 
\usepackage{latexsym,amsmath,amsthm,amscd,amssymb,framed,graphics,xcolor,mathrsfs}
\usepackage{mathtools}
\usepackage{enumerate}
\usepackage{geometry}
\usepackage[colorlinks]{hyperref}
\usepackage[T1]{fontenc}
\usepackage{graphicx}
\usepackage{empheq}
\usepackage{pb-diagram}
\usepackage{amscd}
\definecolor{sand}{rgb}{0.76, 0.7, 0.5}
\definecolor{taupegray}{rgb}{0.55, 0.52, 0.54}
\usepackage[colorlinks]{hyperref}
\usepackage{subcaption}
\hypersetup{
pdfstartview={FitH},
urlcolor=sand,
citecolor=blue,
linkcolor=taupegray,
}

\usepackage{soul}

\usepackage{array}
\usepackage{float}
\usepackage{multicol}
\usepackage{fancyhdr}
\usepackage[arrow, matrix, curve]{xy}
\usepackage{url}
\usepackage{colordvi}
\usepackage{framed}
\usepackage{color}
\usepackage{algorithmic}
\usepackage{algorithm}
\usepackage{tikz}
\usepackage{tikz-cd}
\usepackage{stackengine}
\usepackage{ulem}
\usepackage{cancel}
\usepackage{adjustbox}
\usepackage{enumitem}

\newcommand{\mathsym}[1]{{}}

\newtheorem{theorem}{Theorem}[section]
\newtheorem{definition}[theorem]{Definition}

\newtheorem{remark}[theorem]{Remark}
\newtheorem{proposition}[theorem]{Proposition}
\newtheorem{corollary}[theorem]{Corollary}

\begin{document}
\title{Dirac structures on tangent bundles: a geometric framework
for variational principles, constrained dynamics, and\\ symmetry reduction}


\author{\hspace{-1cm}
\begin{tabular}{c}
Hiroaki Yoshimura
\\[2mm] \normalsize   School of Fundamental Science and Engineering
\\ \normalsize  Waseda University 
\\  \normalsize  3-4-1 Okubo, Shinjuku, Tokyo, Japan
\\ \normalsize  yoshimura@waseda.jp\\
\end{tabular}\\\\
}

\maketitle

\date{}

\begin{abstract}
We introduce a Dirac structure on the tangent bundle of a configuration manifold, called a \textit{Lagrange--Dirac structure}, which is naturally induced by the Lagrangian two-form associated with a (possibly degenerate) Lagrangian and a constraint distribution on the configuration manifold. This structure provides a unified geometric framework for Lagrange--Dirac dynamical systems, including nonholonomic systems, degenerate Lagrangian systems, and systems with symmetry.
In the hyperregular case, the associated Lagrange--Dirac dynamical system recovers a first-order formulation of the Lagrange--d'Alembert equations. Although nonholonomic dynamics does not preserve the Lagrangian two-form along the flow, we show that the underlying Lagrange--Dirac structure is preserved up to gauge transformations, yielding a natural gauge covariance property.
We also formulate an intrinsic variational principle directly on the tangent bundle, which includes Hamilton's principle in the unconstrained case and the Lagrange--d'Alembert principle in the hyperregular constrained case. This variational principle naturally yields the Lagrange--Dirac dynamical system and is referred to as the \textit{Lagrange--d'Alembert--Dirac principle}.
Furthermore, we develop a reduction theory for systems with Lie group symmetry. In particular, we derive a reduced Lagrange--Dirac structure over the Lie algebra and show that the reduced dynamics yields the Euler--Poincar\'e--Dirac equations together with a corresponding reduced variational principle.
Finally, we illustrate the theory through several examples, including charged particle dynamics, nonholonomic systems, electric circuits, and systems with Lagrangians linear in velocity. We also discuss an infinite-dimensional extension to ideal fluid dynamics, where the framework naturally incorporates the incompressibility constraint and recovers the Euler equations, providing an intrinsic geometric interpretation of the relationship between constraints, pressure, and Dirac structures.
\end{abstract}

\tableofcontents

\section{Introduction}
\paragraph{Background.}
Hamilton's variational principle is a fundamental principle of classical mechanics, yielding the Euler--Lagrange equations for a wide class of systems, including particles, rigid bodies, fluids and plasma; see \cite{AbMa1978, Bl2003, CHHM1998, Lit1983, Lo1958, MaRa1999}. For hyperregular Lagrangians, the dynamics of such systems can be described geometrically in terms of the Lagrangian two-form, i.e., the symplectic structure on the tangent bundle induced by the Legendre transformation. In particular, for unconstrained or holonomic systems, the Lagrangian two-form is preserved along the flow of the associated Lagrangian vector field; see \cite{MaRa1999}. 

By contrast, nonholonomic mechanical systems cannot, in general, be formulated via Hamilton's variational principle. Instead, their equations of motion are derived from the Lagrange--d'Alembert principle. 
While numerous studies have addressed nonholonomic dynamics from the Lagrangian viewpoint; see, for example, \cite{NeFu1972, VF1981, BKMM1996, KM1997, Bl2003}, several fundamental questions remain. Specifically, since the Lagrangian two-form is no longer preserved along the flow in the presence of nonholonomic constraints, the precise intrinsic geometric structure underlying such dynamics---and the nature of any associated invariant properties---has not yet been fully elucidated.

Furthermore, in the case of degenerate Lagrangians, the second-order condition may fail, leading to additional constraints arising from the degeneracy, known as \textit{Dirac's primary constraints} \cite{Dirac1950}. In this setting, the Legendre transformation $\mathbb{F}L: TQ \to T^{\ast}Q$ is no longer a diffeomorphism, and hence a Hamiltonian cannot be defined on the cotangent bundle in the usual way. Instead, one needs to consider the primary constraint set $P=\mathbb{F}L(TQ) \subset T^{\ast}Q$ and define a constrained Hamiltonian and a total Hamiltonian using Lagrange multipliers \cite{YoGB2025}. While various constraint algorithms have been developed to determine the final constraint manifold \cite{GoNe1979, SkRu1983, Ha2011}, a unified geometric and variational understanding of systems simultaneously exhibiting both degeneracy and nonholonomic constraints remains an open challenge.

The theory of Dirac structures, which generalizes both presymplectic and (almost) Poisson structures, provides a natural framework for the geometric formulation of such constrained systems, initially in the context of Hamiltonian systems; see \cite{CoWe1988,Cour1990a,Dorfman1993, BS2001}. On the Lagrangian side, a fundamental advancement was made by Yoshimura and Marsden \cite{YoMa2006a}, who first introduced Dirac structures on the cotangent bundle $T^*Q$ induced by a distribution $\Delta_Q \subset TQ$. They demonstrated that \textit{implicit Lagrangian systems} can be formulated to accommodate both degenerate Lagrangians and nonholonomic constraints on the cotangent bundle by introducing the Dirac differential of the Lagrangian. This framework was subsequently extended to Dirac structures on the Pontryagin bundle $TQ \oplus T^*Q$ in conjunction with the Hamilton--Pontryagin principle \cite{YoMa2006b}, providing a unified variational perspective for these systems.

In the presence of symmetry, the reduction of Dirac structures on the cotangent bundle of a Lie group, referred to as \textit{Lie--Dirac reduction} \cite{YoMa2007a}, has systematically incorporated implicit Euler--Poincar\'e--Suslov and Lie--Poisson--Suslov equations. This theory was further extended to systems with advected parameters \cite{GBYo2015}, facilitating a formulation of Lie--Poisson reduction on semidirect products. Moreover, the Dirac framework has been shown to incorporate Dirac's original theory of constraints and to establish a natural connection with the Hamiltonian formulation \cite{YoMa2007b}. More recently, it has been extended to nonequilibrium thermodynamics through Dirac structures on the Pontryagin bundle $TQ \oplus T^\ast Q$, see \cite{GBYo2018,YoGB2025}.
%
%

\paragraph{Contribution of this paper.}
The geometric structure underlying Lagrangian systems on the tangent bundle, particularly in the presence of nonholonomic constraints and degeneracy, has not yet been fully elucidated within the Dirac geometric framework. While Dirac structures have been extensively studied on cotangent and Pontryagin bundles, a direct and intrinsic formulation on the tangent bundle remains relatively unexplored. In particular, for degenerate systems, how the Dirac theory on $TQ$ can systematically incorporate primary constraints remains an open question.
\medskip

In this paper, we  introduce a Dirac structure on $TQ$ referred to as a \textit{Lagrange--Dirac structure} to address these issues. This structure is intrinsically induced by the Lagrangian two-form $\Omega_L$ together with a constraint distribution $\Delta_Q \subset TQ$. We then develop the theory of \textit{Lagrange--Dirac dynamical systems}, which provides a unified first-order framework for mechanical systems with nonholonomic constraints, degenerate Lagrangians, and symmetry. 
\medskip

The key contributions are summarized as follows:
\begin{itemize}
    \item Our framework provides a unified geometric foundation for systems with nonholonomic constraints, degenerate Lagrangians and Lie group symmetry and it recovers the first-order formulation of the Lagrange--d'Alembert equations in the hyperregular case, and reduces to the classical Euler--Lagrange equations for unconstrained systems ($\Delta_Q = TQ$).
    
    \item By analyzing the dynamics on the admissible subbundle under a suitable normality condition, we show that the constrained Lagrange--Dirac structure is preserved up to a gauge transformation. This result provides a definitive geometric characterization of the covariance properties of nonholonomic dynamics.
    
    \item Associated with the Lagrange--Dirac dynamical systems, we establish a new variational principle on $TQ$. This principle elucidates that the phase-space Lagrangians commonly employed in plasma physics are not auxiliary constructions but arise intrinsically from the geometric data of the Lagrange--Dirac structure.
    
    \item We develop a reduction theory for Lagrange--Dirac structures under Lie group symmetry, recovering reduced dynamics such as the Euler--Poincar\'e--Suslov equations in a first-order form. 
\end{itemize}

Finally, we illustrate the versatility of the framework through several examples, including charged particle systems, nonholonomic systems and degenerate electric circuits. Furthermore, we extend the theory to infinite-dimensional systems, demonstrating that ideal fluid dynamics can be naturally formulated as a constrained Lagrange--Dirac system on the diffeomorphism group, where the incompressibility constraint and the resulting pressure gradient are treated as intrinsic components of the Lagrange--Dirac structure.

\paragraph{Organization of this paper.}
The remainder of this paper is organized as follows. 
In \S2, we review the variational formulations for  Lagrangian systems in both unconstrained and nonholonomic settings. We specifically analyze the failure of symplecticity in the presence of nonholonomic constraints and characterize the resulting geometric evolution in terms of an exact two-form.
In \S3, we introduce the \textit{Lagrange--Dirac structure} on the tangent bundle $TQ$ and develop the theory of its associated dynamical systems. We investigate the structural properties of these systems and demonstrate that, in the nonholonomic case, the Lagrange--Dirac structure is preserved along the dynamics up to a gauge transformation, providing a novel geometric characterization of such flows.
In \S4, we develop the reduction theory for Lagrange--Dirac structures under Lie group symmetry. We derive the reduced Lagrange--Dirac dynamical systems and show how they recover the first-order formulation of the Euler--Poincar\'e--Suslov equations.
In \S5, we propose a variational principle on $TQ$, termed the \textit{Lagrange--d'Alembert--Dirac principle}. We show that this principle naturally yields the Lagrange--Dirac dynamics on $TQ$ and derive its corresponding reduced variational formulation, clarifying the link between Dirac geometry and variational structures.
Finally, in \S6, we illustrate the framework through several physical examples, including charged particle dynamics, in which we clarify the geometric meaning of phase-space Lagrangians arising in plasma physics, degenerate systems such as electric circuits and Lagrangians linear in velocity, and nonholonomic Suslov systems. We conclude with an application to ideal fluids, showing that fluid motion naturally arises as an infinite-dimensional constrained Lagrange--Dirac system.

\section{Variational structures in Lagrangian mechanics}
In this section, we review Hamilton's principle and the Lagrange--d'Alembert principle, and derive the associated equations of motion as second-order differential equations. We include this material to formulate both unconstrained and nonholonomic systems within a unified variational framework.

In particular, this framework enables us to study the evolution of the Lagrangian two-form along constrained flows. Although this form is not preserved in general for nonholonomic systems (see, e.g., \cite{CoMa2001, PeYo2025}), a finite-time characterization has not been explicitly addressed in the literature. We derive such a formula and show that the deviation from symplecticity is governed by an exact two-form determined by constraint forces. This observation provides the key structural link with gauge transformations of Dirac structures developed later.

\subsection{Hamilton's principle in Lagrangian mechanics}
\paragraph{Hamilton's variational principle.}
Consider a mechanical system with a configuration manifold $Q$  and a Lagrangian $L: TQ \to \mathbb{R}$ as a smooth function defined on the tangent bundle. Let  $\tau_Q: TQ \to Q$ be the tangent bundle projection.  We introduce a space of curves on $Q$ with fixed endpoints by
\[
\mathcal{C}(Q) = \{ q:I=[0,T] \to Q \mid q \text{ is a } C^{2} \text{ curve with fixed endpoints} \},
\]
where $I=[0,T] \subset \mathbb{R}$ is the time interval. Then, we consider the following action functional, i.e., 
\begin{equation*}
\mathfrak{S}(q)= \int_{0}^{T} L(q(t),  \dot{q}(t))dt,
\end{equation*}
where $\dot{q}(t)=\frac{d}{dt}q(t)$ denotes the time derivative of $q(t) \in \mathcal{C}(Q)$.  Hamilton's variational principle states that 
the motion $q(t)$ is characterized as a critical curve of the action functional
\begin{equation}\label{HP_epsilon}
\begin{split}
\mathbf{d}&\mathfrak{S}(q)\cdot \delta{q}=\left.\frac{d}{d\epsilon}\right|_{\epsilon =0}\mathfrak{S}(q_{\epsilon}(t))\\[2mm]
&=\left.\frac{d}{d\epsilon}\right|_{\epsilon =0}  \int_{0}^{T} L\big(q_{\epsilon}(t), \dot{q}_{\epsilon}(t)\big){\rm d}t\\[2mm]
&=\int^T_0  \left(\mathbf{D}_1 L(q(t),\dot{q}(t))-\frac{d}{dt}\mathbf{D}_2 L(q(t),\dot{q}(t))\right) \cdot\delta{q}(t) \,dt+ \big(\mathbf{D}_2 L(q(t),\dot{q}(t) \big)\cdot\delta q(t) \biggr|_{0}^{T}=0,
\end{split}
\end{equation}
for all variations $\delta q$ with fixed endpoints. Here $q_{\epsilon}(t):=q(t, \epsilon)$ is a deformation of $q(t)$ in $\mathcal{C}(Q)$ with parameter $\epsilon \in (-a,a) \subset\mathbb{R}$, where $a$ is some positive small number and the infinitesimal variations associated with the variations $q_{\epsilon}(t)$ are given by
$
\delta q(t):= \left.\frac{d}{d\epsilon} q_{\epsilon}(t)\right|_{\epsilon =0},
$
which is referred to as the \textit{virtual displacement} in mechanics. By the fixed endpoint condition, i.e., $\delta{q}(0)=\delta{q}(T)=0$, it follows from equation \eqref{HP_epsilon} that the boundary term vanishes and the critical curve $q(t)$ satisfies the \textit{classical second-order Euler-Lagrange equation} on the configuration manifold $Q$:
\begin{equation}\label{EulerLagEqn}
\frac{d}{dt} \mathbf{D}_2 L(q,\dot{q}) = \mathbf{D}_1 L(q,\dot{q}).
%
\end{equation}
Here $\mathbf{D}_{i}  L,\; i=1,2$ denote the derivatives of $L(q,v)$ with respect to the $i$-th argument.
\medskip

In finite dimensions with $n=\dim Q$, using local coordinates $q^{1},\dots,q^{n}$ for $q\in Q$, we have
\begin{equation*}
\begin{split}
\frac{d}{dt}\frac{\partial L}{\partial \dot{q}^{i}}-\frac{\partial L}{\partial {q}^{i}}=0,\quad i=1,\dots,n.
\end{split}
\end{equation*}

\paragraph{The Legendre transform.} The Legendre transform associated with $L$ is given by the fiber derivative $\mathbb{F}L:TQ \to T^{\ast}Q$:
\[
\mathbb{F} L(q,v) \cdot w = \left. \frac{d}{d \epsilon}\right|_{\epsilon=0} L(q,v + \epsilon w), \quad \text{or} \quad \mathbb{F}L(q,v)\cdot w = \mathbf{D}_{2}L(q,v)\cdot w,
\]
where $v, w \in T_q Q$. 
The map $\mathbb{F}L$ is fiber-preserving over $Q$, mapping each fiber $T_qQ$ to $T_q^{\ast}Q$. Locally, the Legendre transform is given by
\[
\mathbb{F} L(q,v) = \bigl(q, \mathbf{D}_2 L(q,v) \bigr), 
\]
where $p := \mathbf{D}_2 L(q,v) \in T_{q}^{\ast}Q$ is a momentum. A Lagrangian $L$ is called \textit{regular} if the fiber Hessian $\mathbf{D}_{2}\mathbf{D}_{2}L(q,v)$ is locally nondegenerate. If $\mathbb{F}L$ is a global diffeomorphism, $L$ is called \textit{hyperregular} \cite{MaRa1999}.

\paragraph{Energy conservation.}
Now let us define the energy 
$E_{L}: TQ \to \mathbb{R}$ by, for $(q,v) \in TQ$,
\[
E_{L}(q,v):=\mathbb{F}L(q,v)\cdot v -L(q,v)= \mathbf{D}_{2}L(q,v)\cdot v-L(q,v).
\]
Along a solution curve $q(t) \in Q$ of the Euler--Lagrange equations in \eqref{EulerLagEqn}, $E_{L}(q(t),\dot{q}(t))$
is conserved, since
\[
\begin{split}
\frac{dE_L}{dt}(q(t),\dot{q}(t))=\left(\frac{d}{dt} \mathbf{D}_2 L(q,\dot{q}) - \mathbf{D}_1 L(q,\dot{q})\right) \cdot \dot{q} 
=0.
\end{split}
\]
\paragraph{Lagrangian two-forms on tangent bundles.}
Let $L:TQ \to \mathbb{R}$ be a \textit{hyperregular} Lagrangian. Let $\pi_{Q}:T^{\ast}Q \to Q; (q,p)\mapsto q$ be the  cotangent bundle projection. The cotangent bundle $T^{\ast}Q$ naturally has the {\it canonical one-form} $\Theta_{T^{\ast}Q}$ defined by, for each $p_{q}=(q,p) \in T^{\ast}_{q}Q$,
$$
\Theta_{T^{\ast}Q}(p_{q}) \cdot u_{p_{q}}= \left< p_{q}, T_{p_{q}}\pi_{Q}(u_{p_{q}}) \right>,
$$
for all  $u_{p_{q}} \in T_{p_{q}}T^{\ast}Q$, and the {\it canonical symplectic structure} is given by
\[
 \Omega_{T^{\ast}Q}=-\mathbf{d}\Theta_{T^{\ast}Q}.
\]

In finite-dimensional cases, using local coordinates $q^1, \dots, q^n,  p_1, \dots, p_n$ for $(q,p)\in T^{\ast}Q$, one has
$
\Theta_{T^{\ast}Q}=p_{i}\,dq^{i}
$
and the {\it canonical symplectic structure} is represented by $ \Omega_{T^{\ast}Q}=dq^{i} \wedge dp_{i}$, where we use Einstein's summation convention.
\medskip

By using the Legendre transform $\mathbb{F}L:TQ \to T^{\ast}Q$, we can define an induced one-form $\Theta_L$ on $TQ$, called the {\it Lagrangian  one-form}, by
$
\Theta_L = (\mathbb{F} L)^{\ast} \Theta_{T^{\ast}Q},
$
which is locally expressed by, for each $(q,v) \in TQ$,
\[
\Theta_L(q,v)\cdot (\dot{q},\dot{v})=\mathbf{D}_{2} L (q,v)\cdot \dot{q}.
\]
We also define the induced two-form $\Omega_L$ on $TQ$, called the {\it Lagrangian two-form},  by
$\Omega_L = (\mathbb{F} L)^{\ast} \Omega_{T^{\ast}Q}$. Since $\Omega_{T^{\ast}Q}=-\mathbf{d}\Theta_{T^{\ast}Q}$ holds and the exterior derivative $\mathbf{d}$ commutes with the pull-back, it reads
$
\Omega_L=-\mathbf{d}\Theta_L,
$
which is locally denoted by, for each $(q,v) \in TQ$, 
\begin{equation}\label{LocalLagTwoForm}
\begin{aligned}
\Omega_L(q,v)((\dot q,\dot v),(\delta q,\delta v))
&=
 \mathbf{D}_2(\mathbf{D}_2 L(q,v)\cdot \dot q) \cdot \delta v
-
\mathbf{D}_2(\mathbf{D}_2 L(q,v)\cdot \delta q) \cdot \dot v  \\
&\quad
+
 \mathbf{D}_1(\mathbf{D}_2 L(q,v)\cdot \dot q) \cdot  \delta q 
-
\mathbf{D}_1(\mathbf{D}_2 L(q,v)\cdot \delta q) \cdot  \dot q,
\end{aligned}
\end{equation}
for all $(\dot{q},\dot{v}), (\delta{q},\delta{v}) \in T_{(q,v)}TQ$.
\medskip

In finite-dimensional cases, using local coordinates $q^1,\dots, q^n, v^{1}, \dots, v^n$ for $v_q=(q,v)\in TQ$, the local coordinate expression of $\Theta_L$ is given by
\[
\Theta_L = \frac{\partial L}{\partial v^{i}}dq^i.
\]
Similarly, $\Omega_L=-\mathbf{d}\Theta_L$ is expressed in finite dimensions by
\begin{equation*}
\Omega_L = \frac{\partial^2 L}{\partial q^{j}\partial v^{i}}dq^i \wedge dq^j
+ \frac{\partial^2 L}{\partial v^{j}\partial v^{i}}dq^i \wedge dv^j,\quad i,j=1,\dots, n.
\end{equation*}

\paragraph{Lagrangian systems.}
Assume $L: TQ \to \mathbb{R}$ is a hyperregular Lagrangian. 
A vector field $X_L$ on $TQ$ is said to be a \textit{Lagrangian vector field} if it satisfies the \textit{Lagrangian condition}, for each point $v_{q} \in TQ$,
\begin{equation}\label{LagCond}
\Omega_{L}(v_{q})(X_{L}(v_{q}), w_{v_{q}})=\mathbf{d}E_{L}(v_{q}) \cdot w_{v_{q}},
\end{equation}
for all $w_{v_{q}} \in T_{v_{q}}TQ$. The triple $(X_{L}, E_{L}, \Omega_{L})$ is called a \textit{Lagrangian system}. From the condition \eqref{LagCond}, we immediately obtain the intrinsic expression of the Lagrangian system:
\begin{equation}\label{InstrinsicLagrangeSystem}
\mathbf{i}_{X_{L}}\Omega_{L}=\mathbf{d}E_{L}.
\end{equation}

We have the following theorem concerning Lagrangian systems; see \cite{MaRa1999},

\begin{theorem}\label{Thm:LagSys} 
Let $v_q(t)=(q(t),v(t)) \in TQ$ be a solution curve of $(X_{L}, E_{L}, \Omega_{L})$ that satisfies the Lagrangian condition \eqref{LagCond}.
Then it follows that $(q(t),v(t))$ satisfies the first-order system
of the local Euler--Lagrange equations on $TQ$:
\begin{equation}\label{EulLagEqn}
\displaystyle \frac{d}{dt}\mathbf{D}_{2} L  = \mathbf{D}_{1}L,\qquad \frac{dq}{dt}=v.
\end{equation}
\end{theorem}
\begin{proof}
Let $v_q(t)=(q(t),v(t))$ be a solution curve in $TQ$ that satisfies the Lagrangian condition \eqref{LagCond}.
Using local expressions $v_{q}=(q,v) \in TQ$ and $w_{v_{q}}=(\delta{q},\delta{v})\in T_{v_{q}}TQ$,  the right-hand side of \eqref{LagCond} is computed as
\[
\begin{split}
\mathbf{d}E_{L}(v_{q}) \cdot  w_{v_{q}}
&=\mathbf{D}_{1}E_{L}(q,v)\cdot \delta q+\mathbf{D}_{2}E_{L}(q,v)\cdot \delta v,\\
&=\left(\mathbf{D}_{1}\mathbf{D}_{2}L(q,v)\cdot v - \mathbf{D}_{1}L(q,v)\right)\cdot \delta q+
\left(\mathbf{D}_{2}\mathbf{D}_{2}L(q,v)\cdot v\right)\cdot \delta v.
\end{split}
\]

On the other hand, in view of $X_{L}(q,v)=(\dot{q}, \dot{v})$, it follows from \eqref{LocalLagTwoForm} that the left-hand side of \eqref{LagCond} is given by
\begin{equation*}
\begin{split}
\Omega_{L}(v_{q})(X_{L}(v_{q}), w_{v_{q}})
&=\mathbf{D}_1 (\mathbf{D}_2 L(q,v) \cdot \dot q) \cdot  \delta q  - \mathbf{D}_1 (\mathbf{D}_2 L(q,v) \cdot \delta q)\cdot \dot q \\
& \qquad+ \mathbf{D}_{2}(\mathbf{D}_{2} L(q,v) \cdot \dot{q}) \cdot \delta v  -  \mathbf{D}_{2}(\mathbf{D}_{2} L(q,v) \cdot \delta q) \cdot  \dot{v},
\end{split}
\end{equation*}
where $w_{v_{q}}=(\delta{q},\delta{v}) \in T_{(q(t),v(t))}TQ$ is arbitrary for each fixed time $t$.
Substituting these expressions into \eqref{LagCond}, we obtain
\begin{equation*}
\begin{split}
&\Big(
\mathbf{D}_1 (\mathbf{D}_2 L(q,v) \cdot \dot q)
-
\mathbf{D}_1 (\mathbf{D}_2 L(q,v) \cdot v)
+
\mathbf{D}_1 L(q,v)
\Big) \cdot \delta q \\
&\quad
-
\mathbf{D}_1 (\mathbf{D}_2 L(q,v) \cdot \delta q)\cdot \dot q 
-
\mathbf{D}_2 (\mathbf{D}_2 L(q,v) \cdot \delta q)\cdot \dot v
=
(\mathbf{D}_2\mathbf{D}_2L(q,v)\cdot (v-\dot q))\cdot \delta v
\end{split}
\end{equation*}
for all $\delta{q}$ and for all $\delta{v}$.

Since $L$ is hyperregular, the fiber Hessian $\mathbf{D}_{2}\mathbf{D}_{2} L$ is invertible. Then, the $\delta v$-term implies $\dot q=v$. Then, it follows from the $\delta q$-term that
\begin{equation}\label{LagVarEq-deltaq}
\mathbf{D}_1( \mathbf{D}_2 L(q,v) \cdot \delta{q}) \cdot  \dot q+\mathbf{D}_{2}(\mathbf{D}_{2} L(q,v) \cdot \delta q) \cdot  \dot{v} 
=
\mathbf{D}_{1}L(q,v) \cdot \delta q.
\end{equation}
Here, the left-hand side of \eqref{LagVarEq-deltaq} is 
\begin{equation*}
\begin{split}
\mathbf{D}_1( \mathbf{D}_2 L(q,v) \cdot \delta{q}) \cdot  \dot q
+
\mathbf{D}_{2}(\mathbf{D}_{2} L(q,v) \cdot \delta q) \cdot  \dot{v}
&=\frac{d}{dt} \left(\mathbf{D}_2 L(q,v)  \cdot  \delta q \right)\\
&=\left(\frac{d}{dt} \mathbf{D}_2 L(q,v)\right)  \cdot  \delta q ,
\end{split}
\end{equation*}
where $(\delta q, \delta v)$ is regarded as an arbitrary tangent vector at each fixed time $t$. 
Hence, in this pointwise computation, $\delta q$ is treated as time-independent, and thus $\frac{d}{dt}(\delta q)=0$; see \cite{MaRa1999}.

Thus, we obtain
\[
\left(\frac{d}{dt} \mathbf{D}_2 L(q,v)\right) \cdot  \delta q 
=
\mathbf{D}_{1}L(q,v) \cdot \delta q.
\]
Because $\delta q$ is arbitrary, we thus obtain the first-order Euler--Lagrange system on $TQ$:
\[
\frac{d}{dt} \mathbf{D}_2 L(q,v) = \mathbf{D}_1 L(q,v), \quad \frac{dq}{dt} = v,
\]
which are equivalent to the classical second-order Euler--Lagrange equation on $Q$, given in \eqref{EulerLagEqn}.
\end{proof}

\paragraph{Second-order submanifolds.}
Recall that $\tau_Q:TQ \to Q;\; v_q=(q,v) \mapsto q$ is the tangent bundle projection and $\tau_{TQ}:TTQ \to TQ;\; (q,v, \dot{q},\dot{v}) \mapsto (q,v)$ is also the canonical projection.
Let us define a submanifold of $TTQ$ by
\begin{equation}\label{SecOrdSubMan}
\ddot{Q} := \left\{w \in TTQ \,\big|\,T\tau_{Q}(w)=\tau_{TQ}(w) \right\}.
\end{equation}
Here $T\tau_{Q}: TTQ \to TQ; (q,v,\dot{q}, \dot{v}) \mapsto (q,\dot{q})$ is the tangent map of $\tau_{Q}$, and if $w=(q,v,\dot{q},\dot{v}) \in \ddot{Q}$, then it satisfies the {\it  second-order condition} 
\[
\dot{q}=v.
\]
Thus, an element $w \in \ddot{Q}$ can be identified with $(q,\dot{q},\ddot{q})$ using the second-order condition $\dot{q}=v$.

\paragraph{Lagrangian vector fields.}
Define a bundle map $\Omega_{L}^{\flat}: TTQ \to T^{\ast}TQ$ such that $\Omega_{L}(v_{q})(u, w)=\langle\Omega_{L}^{\flat}(v_{q})(u), w\rangle=\langle\mathbf{i}_{u}\Omega_{L}(v_{q}), w\rangle$ for each $v_{q} \in TQ$ and $u, w \in T_{v_{q}}TQ$.
From the intrinsic Lagrangian system in \eqref{InstrinsicLagrangeSystem}, for $v_{q} \in TQ$, we have
\[
\Omega^{\flat}_{L}(v_{q})\left(X_{L}(v_{q})\right)=\mathbf{d}E_{L}(v_{q}).
\]
The Lagrangian dynamics is uniquely determined by the Lagrangian vector field $X_L$: for each $v_{q}=(q,v) \in TQ$,
\begin{equation*}
X_{L}(q,v)=\left(\Omega^{\flat}_{L}(q,v)\right)^{-1}\mathbf{d}E_{L}(q,v).
\end{equation*}
Locally, 
$X_L=(X_{L_1},X_{L_2})$ is given by
\[
(q,v) \mapsto (\dot{q},\dot{v}) = \bigl(X_{L_1}(q,v), X_{L_2}(q,v)\bigr),
\]
where
\[
\begin{split}
X_{L_{1}}(v_{q})&=v,\\
X_{L_{2}}(q,v)&=\big( \mathbf{D}_{2}\mathbf{D}_{2}L(q,v)\big)^{-1}
\big(\mathbf{D}_{1}L(q,v)-\mathbf{D}_{1}\mathbf{D}_{2}L(q,v)\cdot v\big).
\end{split}
\] 
The Lagrangian vector field $X_{L}$ is a second-order vector field, meaning that $X_{L}: TQ \to \ddot{Q}$ satisfies $T\tau_Q \circ X_L = \text{id}_{TQ}$. Locally, $X_L(q,v) = (q, v, \dot{q}, \dot{v})$ where $\dot{q}=v$ and 
\[
\dot{v} = \big( \mathbf{D}_{2}\mathbf{D}_{2}L(q,v)\big)^{-1} \big(\mathbf{D}_{1}L(q,v) - \mathbf{D}_{1}\mathbf{D}_{2}L(q,v)\cdot v \big).
\]

\paragraph{Euler--Lagrange operator.} 
Following \cite{CeMaRa2001a}, we define the \textit{Euler--Lagrange operator}
$\mathcal{D}_{EL}(L):\ddot{Q} \to T^{\ast}Q$ by
\begin{equation*}
\mathcal{D}_{EL}(L)(q, \dot{q}, \ddot{q}) \cdot \delta {q}= \bigg(\mathbf{D}_{1} L(q,\dot{q})- \frac{d}{dt}\mathbf{D}_{2}L(q,\dot{q})\bigg) \cdot \delta{q},
\end{equation*}
Hamilton's principle \eqref{HP_epsilon} can be restated using the Euler--Lagrange operator as follows: if a curve $q=q(t) \in \mathcal{C}(Q)$ is a critical point of $\mathfrak{S}: \mathcal{C} (Q) \to \mathbb{R}$ that satisfies
\begin{equation*}
\mathbf{d} \mathfrak{S}(q)\cdot \delta{q}=\int^T_0 \mathcal{D}_{EL}(L) \left(q(t), \dot{q}(t), \ddot{q}(t)\right) \cdot \delta q (t)\,dt+  \Theta_L(q(t),\dot{q}(t)) \cdot \hat{\delta q}(t) \biggr|_{0}^{T}=0,
\end{equation*}
for all variations $\delta{q}$ with the fixed endpoint conditions $\delta{q}(0)=\delta{q}(T)=0$, then the critical curve $q=q(t)$ satisfies
\[
\mathcal{D}_{EL}(L)\left(q, \dot{q}, \ddot{q}\right)=0,
\]
which is nothing but the second-order Euler--Lagrange equation \eqref{EulerLagEqn}. In the above, 
\begin{equation*}
\begin{split}
\Theta_L(q,\dot{q}) \cdot (\delta q, \delta \dot{q})
= \mathbf{D}_2 L(q,\dot{q}) \cdot \delta q, \qquad
\hat{\delta q}=\frac{d}{d\epsilon}\biggr{\arrowvert}_{\epsilon =0}\frac{d}{dt}q_{\epsilon}
=\left( \left(q, \dot{q} \right),  \left(\delta q, \delta \dot{q} \right)\right),
\end{split}
\end{equation*}
where $\hat{\delta q}$ denotes the lift of the variation $\delta q$ to $TTQ$, given by the tangent lift of the curve $q_\epsilon(t)$.
\medskip

In finite dimensions, the Euler--Lagrange operator $\mathcal{D}_{EL}(L):\ddot{Q} \to T^{\ast}Q$ is denoted in local coordinates by
\begin{equation*}
\mathcal{D}_{EL}(L)_{i}\left(q, \dot{q}, \ddot{q}\right)dq^{i}= \left(\frac{\partial L}{\partial q^i}(q,\dot{q})- \frac{d}{dt}\frac{\partial L}{\partial \dot{q}^i}(q,\dot{q})\right)dq^{i}.
\end{equation*}
In local coordinates, the corresponding Lagrangian vector field is given by
\begin{equation*}
\ddot{q}^{j}=W^{ji}\left( \frac{\partial L}{\partial q^i} - \frac{\partial^{2} L}{\partial q^k \partial \dot{q}^i}\dot{q}^k\right),\quad i, j, k =1,...,n,
\end{equation*}
where $(W^{ji})$ denotes the inverse  of the Hessian matrix $\left(\frac{\partial^2 L}{\partial \dot{q}^j \partial \dot{q}^i}\right)$.

Thus, there exists a unique second-order Lagrangian vector field $X_{L} : TQ \to \ddot{Q}$ such that
\begin{equation*}
\mathcal{D}_{EL}(L) (X_{L}(q,\dot{q}))=0.
\end{equation*}

\paragraph{Symplecticity of the Lagrangian flow.}
Let $C_L(Q) \subset \mathcal{C}(Q)$ be the {\it space of solution curves of Euler--Lagrange equations}. 
Since the flow of $X_L$ on $TQ$ is uniquely determined by initial conditions, each solution curve is characterized by an initial value $v_q=(q(0),\dot{q}(0)) \in TQ$. We denote the flow of $X_L$ by $\varphi_t: TQ \to TQ$, and write
\[
\varphi_t(v_q) = (q(t), \dot{q}(t)),
\]
so that the base curve is given by $q(t) = \tau_Q(\varphi_t(v_q)) \in C(Q)$.
Now we define the {\it restricted action functional} by 
\[
\hat{\mathfrak{S}}(v_{q}):=\int_0^TL(\varphi_t(v_q))dt,
\]
and a direct computation shows that
\begin{equation*}
\begin{split}
\mathbf{d}\hat{\mathfrak{S}}(v_q)\cdot w_{v_q}&=\Theta_L(\varphi_T(v_q)) \cdot (T\varphi_T(w_{v_q}))-\Theta_L(v_q)\cdot w_{v_q}\\[3mm]
&=\varphi_T^{\ast}\Theta_L(v_q)\cdot w_{v_{q}}-\Theta_L(v_q)\cdot w_{v_q},
\end{split}
\end{equation*}
where $w_{v_q} \in T_{v_q}TQ$.
Taking the exterior derivative and using $\Omega_L = -\mathbf{d}\Theta_L$, we conclude that the Lagrangian flow $\varphi_t$ preserves the symplectic form $\Omega_L$:
\[
\varphi^{\ast}_T\Omega_L = \Omega_L.
\]

\subsection{The Lagrange--d'Alembert principle}\label{sect:LDAP}
\paragraph{Nonholonomic constraints.} Next we consider a mechanical system on an $n$-dimensional configuration manifold $Q$ with hyperregular Lagrangian $L: TQ \to \mathbb{R}$, and we also consider a constraint distribution $\Delta_Q \subset TQ$ on $Q$.  
In this paper, we assume that every distribution is {\it regular}, namely, it is smooth and has constant rank at each point $q \in Q$, unless otherwise stated.

Now we suppose that the distribution $\Delta_Q \subset TQ$ is given, for each $q \in Q$, by
\begin{equation}\label{ConstraintDistribution}
\Delta_Q(q)=\left\{ \dot{q} \in T_{q}Q\,\mid \left< \omega^{r}(q),\dot{q}\right>=0,\,r=1,\ldots,m <n\right\},
\end{equation}
where $\omega^{r}=\omega_{i}^{r}(q)dq^{i},\; r=1,\dots,m<n$ are some given $m$ independent constraint one-forms on $Q$. 

Let $X,Y \in \mathfrak{X}(Q)$ be vector fields taking values in $\Delta_Q$. 
If $\Delta_Q$ is closed under the Lie bracket, i.e., $[X,Y] \in \Delta_Q$ for all such $X,Y$, 
then the distribution $\Delta_Q \subset TQ$ is \textit{completely integrable} in the sense of Frobenius. 
In this case, the constraint is holonomic, since there exists a submanifold $N \subset Q$ such that
\[
T_q N = \Delta_Q(q), \quad \forall q \in N.
\]
Otherwise, the distribution $\Delta_Q$ is nonintegrable, namely, nonholonomic.

\paragraph{The Lagrange--d'Alembert principle.}
In the presence of nonholonomic constraints, Hamilton's principle in \eqref{HP_epsilon} is modified into the variational formulation called the \textit{Lagrange--d'Alembert principle}. 

A curve $q(t)$, $t \in [0,T]$, is a solution of this principle if it satisfies the kinematic constraint $\dot{q}(t) \in \Delta_Q(q(t))$ and if
\begin{equation}\label{LADP_op}
\begin{split}
\mathbf{d} \mathfrak{S}(q)\cdot \delta{q}
=\int^T_0  \left( \mathbf{D}_1 L(q,\dot{q})-\frac{d}{dt} \mathbf{D}_2 L(q,\dot{q})\right) \cdot \delta{q} \, dt+ \mathbf{D}_2 L(q,\dot{q}) \cdot \delta{q} \bigg|^T_0=0,
\end{split}
\end{equation}
for all variations $\delta{q}(t)$ chosen such that $\delta q(t) \in \Delta_Q(q(t))$ with $\delta q(0)=\delta q(T)=0$.

Then, the curve $q(t)$ obeys the \textit{local Lagrange--d'Alembert equations}:
\begin{equation}\label{LDAEq}
\frac{d}{dt} \mathbf{D}_2 L(q,\dot{q})-\mathbf{D}_1 L(q,\dot{q}) \in \Delta_Q^{\circ}(q), \qquad \dot{q} \in \Delta_Q(q),
\end{equation}
where $\Delta_Q^{\circ}\subset T^{\ast}Q$ is the annihilator  of the distribution $\Delta_Q \subset TQ$ and is defined by, for each $q \in Q$,
\[
\Delta_Q^{\circ}(q)=\left\{ \beta \in T_{q}^{\ast}Q \mid \left< \beta, v_q \right>=0, \;\textrm{for all}\;v_q \in \Delta_{Q}(q) \right\}.
\]

An element $\beta=\beta_i(q) dq^i \in \Delta_Q^{\circ}$ represents the constraint force, where the components are expressed as $\beta_i(q)=\mu_r \omega_i^r(q)$ using Lagrange multipliers $\mu_r$. Consequently, the local coordinate expressions for the Lagrange--d'Alembert equations in \eqref{LDAEq} are given by:
\begin{equation*}
 \frac{d}{dt} \frac{\partial L}{\partial \dot{q}^{i}} - \frac{\partial L}{\partial q^{i}} =\mu_{r}\omega^{r}_{i}(q), \;\; \omega_{i}^{r}(q)\dot{q}^{i}=0.
\end{equation*}
For the unconstrained case in which $\Delta_{Q}=TQ$, it follows that the equations \eqref{LDAEq} recover the Euler--Lagrange equations \eqref{EulerLagEqn}.

\paragraph{Energy conservation.}
Along the solution curve $q(t) \in Q$ of the Lagrange--d'Alembert equations in \eqref{LDAEq}, the energy $E_{L}$ is conserved, since
\[
\frac{dE_L}{dt}(q,\dot{q})=\left( \frac{d}{dt} \mathbf{D}_2 L(q,\dot{q})-\mathbf{D}_1 L(q,\dot{q})\right) \cdot \dot{q}
=0,
\]
where $\frac{d}{dt} \mathbf{D}_2 L(q,\dot{q})-\mathbf{D}_1 L(q,\dot{q})\in \Delta_{Q}^{\circ}(q)$ and $\dot{q} \in \Delta_{Q}(q)$.
\begin{definition}[Second-order nonholonomic constraint submanifold]
Recall that $\ddot{Q}$ is the second-order submanifold of $TTQ$ given in \eqref{SecOrdSubMan}. 
We restrict $\ddot{Q}$ to the distribution $\Delta_Q \subset TQ$ and define the \textit{second-order nonholonomic constraint submanifold} of $TTQ$ by
\[
\ddot{Q}_{\Delta_Q} := \{ w \in \ddot{Q} \mid \tau_{TQ}(w) \in \Delta_Q \}.
\]
In other words, this is the space of admissible accelerations. Namely, an element $w=(q,v,\dot{q},\dot{v}) \in TTQ$ satisfies the second-order condition $\dot q = v$ together with the constraint $v = \tau_{TQ}(w) \in \Delta_Q(q)$. 
Thus, an element of $\ddot{Q}_{\Delta_Q}$ is identified with $(q,\dot{q},\ddot{q})$, where $\dot{q} \in \Delta_Q(q)$.
\end{definition}
%
\begin{definition}[The Lagrange--d'Alembert operator]
Let $\iota_{\Delta_Q}: \Delta_Q \hookrightarrow TQ$ and 
$\iota^{(2)}_{\Delta_Q}: \ddot{Q}_{\Delta_Q} \hookrightarrow \ddot{Q}$ 
be the natural inclusions. 

Now, we define the \textit{Lagrange--d'Alembert operator} as the restriction of the Euler--Lagrange operator to $\ddot{Q}_{\Delta_Q}$:
\[
\mathcal{D}_{LD}(L) := \mathcal{D}_{EL}(L) \circ \iota^{(2)}_{\Delta_Q} 
: \ddot{Q}_{\Delta_Q} \to T^*Q.
\]
Its coordinate expression is given by
\[
\mathcal{D}_{LD}(L)_i(q,\dot q,\ddot q)\,dq^i
= \left(\frac{\partial L}{\partial q^i}(q,\dot q)
- \frac{d}{dt}\frac{\partial L}{\partial \dot q^i}(q,\dot q)\right)dq^i
\in T^*_q Q.
\]
\end{definition}
Now the Lagrange--d'Alembert principle in \eqref{LADP_op} can be restated by using the Lagrange--d'Alembert operator as follows:
\begin{equation}\label{LADP_int}
\mathbf{d} \mathfrak{S}(q)\cdot \delta{q}=\int^T_0 \mathcal{D}_{LD}(L) \left(q, \dot{q}, \ddot{q}\right) \cdot \delta q \,dt+\Theta_L \left(q,\dot{q}\right)\cdot \hat{\delta{q}}\biggr{\arrowvert}^T_0=0,
\end{equation}
for chosen variations $\delta q \in \Delta_Q(q)$ with the fixed  endpoint conditions $\delta q(0)=\delta q(T)=0$, together with the nonholonomic constraints $\dot{q} \in \Delta_{Q}(q)$, so that $\ddot{q}\in \ddot{Q}_{\Delta_Q}$. 
\medskip

Hence, it follows from equation \eqref{LADP_int} that if a curve $q=q(t) \in \mathcal{C}(Q)$ is a critical point of $\mathfrak{S}: \mathcal{C} (Q) \to \mathbb{R}$, i.e.,
$\mathbf{d} \mathfrak{S}(q)=0$ for all $\delta q(t) \in \Delta_Q(q(t))$ with the fixed endpoint conditions, then $q(t)$ satisfies the {\it Lagrange--d'Alembert equations}:
\begin{equation}\label{LDAEOp_eqn}
\mathcal{D}_{LD}(L)(q, \dot{q}, \ddot{q}) \in \Delta_{Q}^{\circ}(q), \qquad \dot{q} \in \Delta_{Q}(q),
\end{equation}
which is equivalent with the equations in \eqref{LDAEq}.

\begin{remark}\rm
In the present formulation, the Lagrange--d'Alembert operator is defined as
\[
\mathcal{D}_{LD}(L) := \mathcal{D}_{EL}(L) \circ \iota^{(2)}_{\Delta{Q}}: \ddot{Q}_{\Delta_Q} \to T^*Q,
\]
whose values are covectors in the ambient cotangent bundle \(T^*Q\).
The Lagrange--d'Alembert equations are then expressed by the condition in \eqref{LDAEOp_eqn}, namely, 
\[
\mathcal{D}_{LD}(L)(q, \dot{q}, \ddot{q}) \in \Delta_Q^\circ(q), 
\qquad 
\dot{q} \in \Delta_Q(q).
\]
This viewpoint keeps the dynamics in the ambient space and identifies the constraint force explicitly as covectors in $\Delta_Q^\circ$.

In conjunction with $\mathcal{D}_{LD}(L): \ddot{Q}_{\Delta_Q} \to T^*Q$, the projected Lagrange--d'Alembert  operator 
$
\mathcal{LD}(L): \ddot{Q}_{\Delta_Q} \to \Delta_Q^*,
$
whose values lie in the dual bundle $\Delta_Q^*$ rather than in the entire cotangent bundle, was employed in the paper \cite{CeMaRa2001b}.
The two formulations are related through the natural projection
\[
\mathcal{LD}(L) = (\iota_{\Delta_{Q}})^{\ast} \circ \mathcal{D}_{LD}(L),
\]
where $(\iota_{\Delta_{Q}})^{\ast}: T^*Q \to \Delta_Q^*$ is the dual of the natural injection $ \iota_{\Delta_{Q}}: \Delta_{Q}  \hookrightarrow  TQ$.
In this setting, we obtain the {\it constrained Lagrange--d'Alembert equations} as 
\begin{equation*}
\mathcal{LD}(L)(q, \dot{q}, \ddot{q}) =0.
\end{equation*}
In the above, since the kernel of the projection $(\iota_{\Delta_{Q}})^{\ast}: T^*Q \to \Delta_Q^*$ is precisely the annihilator  $\Delta_Q^\circ(q)$, the condition $\mathcal{LD}(L)(q, \dot{q}, \ddot{q}) =0$ is pointwise equivalent to $\mathcal{D}_{LD}(L)(q, \dot{q}, \ddot{q}) \in \Delta_Q^\circ(q)$.
Hence, the operator $\mathcal{LD}(L)$ may be viewed as the projection of the ambient-space operator $\mathcal{D}_{LD}(L)$ onto the dual of the constraint distribution.
This relation clarifies that both approaches are compatible; that is,
$\mathcal{D}_{LD}(L)$ emphasizes the embedding into $T^*Q$
while $\mathcal{LD}(L)$ focuses on the intrinsic geometry of the constraint subbundle.
\end{remark}

\subsection{Structural  properties in the Lagrange--d'Alembert principle}\label{Sec:structural properties}
In this subsection, we study structural properties underlying the Lagrange--d'Alembert principle. In particular, we derive a finite-time characterization of the evolution of the Lagrangian two-form along the constrained flow, which clarifies the failure of symplecticity in nonholonomic systems in terms of an exact two-form. 

\paragraph{Constraint distributions and annihilators.}
Associated with $\Delta_{Q}$ on $Q$, we define a distribution $\Delta_{TQ}$ on $TQ$ by
\[
\Delta_{TQ}=(T\tau_{Q})^{-1}(\Delta_{Q})  \subset TTQ, 
\]
which is locally denoted by, for each $v_{q} \in T_{q}Q$,
\[
\Delta_{TQ}(v_{q})=\left\{ (q,v,\dot{q},\dot{v}) \in T_{v_{q}}TQ \mid  \dot{q} \in \Delta_{Q}(q) \right\}.
\]
Therefore, the annihilator of $\Delta_{TQ}$ is given by the pull-back of $\Delta_Q^\circ$ as
\[
\Delta_{TQ}^{\circ}=\tau_Q^{\ast}\Delta_Q^{\circ} \subset T^{\ast}TQ.
\]
%
Let $\beta \in \Gamma(\Delta^{\circ}_{Q})$ be the covector field on $Q$, which represents the constraint force enforcing the constraints $\Delta_{Q} \subset TQ$.  An element of $\Delta^{\circ}_{TQ}$ is given by $\boldsymbol{\beta}:=\tau_Q^{\ast}\beta$, which is a horizontal one-form on $TQ$  that annihilates $w_{v_q} \in \Delta_{TQ}(v_{q})$ such that, for each $v_q \in TQ$, 
\[
\boldsymbol{\beta}(v_{q})\cdot w_{v_q}= \left<\beta(\tau_Q(v_{q})), T\tau_Q(w_{v_q})\right>=0.
\]
\begin{definition}[Symplectic orthogonal subbundles and normality]\label{def:SympOrthSub_Normality}
Associated with the Lagrangian two-form $\Omega_L$, define the
$\Omega_L$--orthogonal subbundle $\mathcal V\subset TTQ|_{\Delta_Q}$
(assumed to have constant rank) by, for each $v_{q}\in \Delta_{Q}$,
\[
\mathcal V_{v_q}
=
\{\,w\in T_{v_q}TQ \mid
\mathbf i_w\Omega_L(v_q)\in\Delta_{TQ}^\circ(v_q)\,\},
\]
for which we sometimes write $\mathcal{V}=(T\Delta_{Q})^{\Omega_{L}}$.
Assume that the constraints are normal~\cite{MadeLeDadeDi1996, KM1998, Marle1998}, i.e., for every $v_q\in\Delta_Q$,
\[
T_{v_q}\Delta_Q\cap\mathcal V_{v_q}=\{0\}.
\]
Under this assumption, the tangent space splits as
\[
T_{v_q}TQ=T_{v_q}\Delta_Q\oplus\mathcal V_{v_q},
\]
and hence 
\[
TTQ|_{\Delta_Q}=T\Delta_Q\oplus\mathcal V \quad \text{over } \Delta_Q.
\]
\end{definition}

\paragraph{The Lagrange--d'Alembert vector fields of constrained dynamics.}
Under the normality condition given in Definition \ref{def:SympOrthSub_Normality}, restricting the Lagrangian vector field $X_{L}$ to the constraint
submanifold $\Delta_{Q}\subset TQ$, we have the unique decomposition
\begin{equation*}
X_{L}=X_{\Delta_{Q}}+X_{\mathcal{V}}, \quad \textrm{over \;$\Delta_Q$}.
\end{equation*}
This decomposition separates the constrained dynamics and the constraint force.
Namely, $X_{\Delta_Q}: \Delta_Q \to T\Delta_{Q}$ represents the constrained dynamics, where the component of motion tangent to $\Delta_Q$ describes the actual evolution of the system under the constraints. On the other hand, $X_{\mathcal V}: \Delta_Q \to \mathcal V$ represents 
the constraint force field, taking values in the symplectic orthogonal subbundle $\mathcal V$. It is given by
\[
X_{\mathcal{V}}(q,\dot q)
= \big(\Omega^{\flat}_L(q,\dot{q})|_{\mathcal V_{(q,\dot{q})}}\big)^{-1} \cdot \beta\big(\tau_{Q}(q,\dot q)\big).
\]
In the above, $\Omega_L^{\flat}|_{\mathcal V}$ 
is the restriction of $\Omega_L^{\flat}: TTQ \to T^{\ast}TQ$ to $\mathcal V$, which is an isomorphism under
the assumptions of the hyperregularity of $L$ and the normality of the constraints. Note that, associated with the Lagrange--d'Alembert equations for the constrained system \eqref{LDAEOp_eqn}, we have the relation
\begin{equation}\label{eq: LDA}
\mathcal{D}_{LD}(L) \circ X_{\Delta_{Q}}(q,\dot{q})
= \beta\big(\tau_{Q}(q,\dot q)\big)
\in \Delta_Q^\circ\big(\tau_{Q}(q,\dot q)\big)
\end{equation}
for each $(q,\dot{q}) \in \Delta_{Q}$. Consequently, the constrained dynamics can be
written as
\begin{equation*}
\left\{
\begin{aligned}
& \ddot{q} = X_{\Delta_{Q}}(q,\dot{q}), \\[2mm]
& X_{\mathcal{V}}(q,\dot{q})
  =(\Omega^{\flat}_{L}(q,\dot{q})|_{\mathcal V})^{-1}\!\big(\boldsymbol{\beta}(q,\dot{q})\big).
\end{aligned}
\right.
\end{equation*}

\begin{remark}\rm
In equation \eqref{eq: LDA}, although the Lagrange--d'Alembert operator $\mathcal{D}_{LD}(L)$ is defined on the space of admissible accelerations $\ddot{Q}_{\Delta_Q}$, the nonholonomic dynamics is described by a vector field $X_{\Delta_{Q}}$ taking values in $T\Delta_Q \subset \ddot{Q}_{\Delta_Q}$. Consequently, the operator is evaluated only on accelerations compatible with the constraint manifold. This reflects the fact that the actual trajectory $(q, \dot{q})$ is confined to $\Delta_Q$, even though the operator itself is defined on the ambient acceleration space.
\end{remark}

\paragraph{Failure of symplectic preservation in nonholonomic systems.}
The Lagrange--d'Alembert principle \eqref{LADP_op} provides a fundamental identity characterizing the failure of symplectic preservation in the presence of nonholonomic constraints.

\begin{framed}
\begin{proposition}\rm\label{Prop:ConvOmegaL}
Let $X_{\Delta_Q}$ be the vector field on the constraint submanifold $\Delta_Q$ determined by the Lagrange--d'Alembert principle, and let 
\[
\varphi_t:\Delta_Q\to\Delta_Q
\]
be its (local) flow.  Let $\beta \in \Gamma(\Delta_Q^\circ) \subset \Omega^{1}(Q)$ and define $\boldsymbol{\beta} := \tau_Q^{\ast}\beta\in \Gamma(\Delta_{TQ}^\circ) \subset \Omega^{1}(TQ)$. Recall $\iota_{\Delta_Q}:\Delta_Q \hookrightarrow TQ$ is the inclusion map.  
Then, along solution curves of the constrained dynamics, one has the identity on $\Delta_Q$:
\begin{equation*}
\pounds_{X_{\Delta_{Q}}}\,(\iota_{\Delta_Q}^{\ast}\Omega_{L}) \;=\; \mathbf{d}\,(\iota_{\Delta_Q}^{\ast}\boldsymbol{\beta}).
\end{equation*}
Equivalently, the following relation gives a finite-time expression of this failure.
\begin{equation}\label{Non_Symp_ExtForceLagSys}
\varphi_t^{\ast}(\iota_{\Delta_Q}^{\ast}\Omega_L) = \iota_{\Delta_Q}^{\ast}\Omega_L + \int_0^t \varphi_s^{\ast}\mathbf{d}(\iota_{\Delta_Q}^{\ast}\boldsymbol{\beta})\, ds
\end{equation}
where all forms and pull-backs are restricted to $\Delta_Q$.
\end{proposition}
\end{framed}
\begin{proof}
Let $\mathcal{C}_{L}(Q) \subset \mathcal{C}(Q)$ be the {\it space of solution curves of the Lagrange--d'Alembert equations} in \eqref{LDAEq}. 
Consider the restriction of the action functional 
\begin{equation*}
\mathfrak{S}(q)= \int_0^t L \left( q(s), \dot{q}(s) \right) ds
\end{equation*}
to the submanifold $\mathcal{C}_{L}(Q) \subset \mathcal{C}(Q)$. We then identify 
a solution curve in $\mathcal{C}_{L}(Q)$, i.e., 
\[
s \mapsto \varphi_s(v_q), \quad s \in [0,t]
\]
with the initial condition $v_q=(q(0),\dot{q}(0)) \in \Delta_Q \subset TQ$ so that
\begin{equation*}
\hat{\mathfrak{S}}(v_{q})=\mathfrak{S}(q(\cdot)),\quad q \in \mathcal{C}_{L}(Q) \;\;\text{and}\;\; (q(0),\dot{q}(0))=v_{q}.
\end{equation*}
Define the restricted action functional $\hat{\mathfrak{S}}(v_{q}): \Delta_{Q} \to \mathbb{R}$ by
\begin{align*}
\hat{\mathfrak{S}}(v_{q})=\int_0^t L(\varphi_s(v_q))ds
\end{align*}
where we note $\hat{\mathfrak{S}}$ is a function on $\Delta_Q$. Then, it follows from \eqref{LDAEOp_eqn} that
\begin{equation}\label{dHatS}
\begin{split}
&\mathbf{d}\hat{\mathfrak{S}}(v_q)\cdot w_{v_q}
=\int^t_0 \mathcal{D}_{LD}(L)(q, \dot{q}, \ddot{q})\cdot \delta q ds+\Theta_L \left(\dot{q}\right)\cdot \hat{\delta{q}}\biggr{\arrowvert}^t_0\\[2mm]
&=\int_0^t \boldsymbol{\beta}(\iota_{\Delta_{Q}}(\varphi_s(v_q)))\cdot (T\iota_{\Delta_{Q}}({\varphi_{s}}_{\ast}(w_{v_{q}})))ds\\
&\hspace{3cm}+
\Theta_L(\iota_{\Delta_{Q}}(\varphi_t(v_q)))\cdot (T\iota_{\Delta_{Q}}{\varphi_{t}}_{\ast}(w_{v_{q}}))-\Theta_L(\iota_{\Delta_{Q}}(v_q))\cdot T\iota_{\Delta_{Q}}(w_{v_{q}})\\[2mm]
&=\int_0^t \varphi_s^{\ast}(\iota_{\Delta_{Q}}^{\ast}\boldsymbol{\beta})(v_q)\cdot w_{v_{q}}ds
+\varphi_t^{\ast}(\iota_{\Delta_{Q}}^{\ast}\Theta_L)(v_q)\cdot w_{v_{q}}-(\iota_{\Delta_{Q}}^{\ast}\Theta_L)(v_q)\cdot w_{v_q},
\end{split}
\end{equation}
for arbitrary $w_{v_{q}}=(q(0),\dot{q}(0),\delta{q}(0),\delta\dot{q}(0)) \in T_{v_{q}}\Delta_{Q}$. 
From \eqref{dHatS}, we obtain
\begin{align}\label{dHatS_2}
\mathbf{d}\hat{\mathfrak{S}}=\varphi_t^{\ast}(\iota_{\Delta_{Q}}^{\ast}\Theta_L)-(\iota_{\Delta_{Q}}^{\ast}\Theta_L) + \int_0^t \varphi_s^{\ast}(\iota_{\Delta_{Q}}^{\ast}\boldsymbol{\beta}) ds.
\end{align}
Taking the exterior derivative of \eqref{dHatS_2},
\begin{equation*}
\begin{split}
0&=\mathbf{d}\mathbf{d}\hat{\mathfrak{S}}\\
&=\mathbf{d}\left(\varphi^*_t(\iota_{\Delta_{Q}}^{\ast}\Theta_L)-(\iota_{\Delta_{Q}}^{\ast}\Theta_L) +\int_0^t \varphi_s^{\ast}(\iota_{\Delta_{Q}}^{\ast}\boldsymbol{\beta}) ds\right)\\
&=-\varphi^*_t(\iota_{\Delta_{Q}}^{\ast}\Omega_L)+(\iota_{\Delta_{Q}}^{\ast}\Omega_L) + \int_0^t \varphi_s^{\ast}\mathbf{d}(\iota_{\Delta_{Q}}^{\ast}\boldsymbol{\beta}) ds
\end{split}
\end{equation*}
and hence we get
\begin{equation*}
\varphi^{\ast}_t (\iota_{\Delta_{Q}}^{\ast}\Omega_L) = (\iota_{\Delta_{Q}}^{\ast}\Omega_L) + \int_0^t \varphi_s^{\ast}\mathbf{d}(\iota_{\Delta_{Q}}^{\ast}\boldsymbol{\beta}) ds.
\end{equation*}
By taking the derivative with respect to $t$ at $t=0$, it follows that 
\[
\frac{d}{dt}\biggr{\arrowvert}_{t=0}\left( \varphi_{t}^{\ast}(\iota_{\Delta_{Q}}^{\ast}\Omega_{L})\right)=\mathbf{d}(\iota_{\Delta_{Q}}^{\ast}\boldsymbol{\beta}).
\]

Since
\[
\pounds_{X_{\Delta_Q}}(\iota_{\Delta_{Q}}^{\ast}\Omega_{L})=\frac{d}{dt}\biggr{\arrowvert}_{t=0}\left( \varphi_{t}^{\ast}(\iota_{\Delta_{Q}}^{\ast}\Omega_{L})\right),
\]
we obtain the desired relation:
\[
\pounds_{X_{\Delta_Q}}(\iota_{\Delta_{Q}}^{\ast}\Omega_{L})=\mathbf{d}(\iota_{\Delta_{Q}}^{\ast}\boldsymbol{\beta}).
\] 
By abuse of notation, we sometimes write
\[
\pounds_{X_{\Delta_Q}}\Omega_{L}=\mathbf{d}\boldsymbol{\beta},
\]
meaning the pull-back of all quantities to $\Delta_Q$.
\end{proof}
Thus, although nonholonomic dynamics is not symplectic in the usual sense, the deviation from symplecticity is completely controlled by an exact two-form generated by the reaction forces.
\begin{remark}
In \cite{CoMa2001}, it was shown that the Lagrangian two-form is not preserved along the constrained flow, namely,
\[
\pounds_{X_{\Delta_Q}}\Omega_L = d\boldsymbol{\beta},
\]
and also that there exists the relation
\[
\Phi^*\Omega_{L_d} = \Omega_{L_d} + \mathbf{d}\boldsymbol{\beta}_d
\]
in the discrete level, where we denote the discrete
time evolution map on $Q \times Q$ by $\Phi$  and also denote the discrete Lagrangian two-form by $\Omega_{L_d}$. Here we develop the 
expression in \eqref{Non_Symp_ExtForceLagSys}, i.e., 
\begin{equation*}
\varphi^{\ast}_t (\iota_{\Delta_{Q}}^{\ast}\Omega_L) = (\iota_{\Delta_{Q}}^{\ast}\Omega_L) + \int_0^t \varphi_s^{\ast}\mathbf{d}(\iota_{\Delta_{Q}}^{\ast}\boldsymbol{\beta}) ds
\end{equation*}
can be regarded as a continuous-time analogue of the discrete relation $\Phi^*\Omega_{L_d} = \Omega_{L_d} + \mathbf{d}\boldsymbol{\beta}_d$. In particular, the time-integrated accumulation of the exact two-form in the continuous setting corresponds to the discrete two-form $\mathbf{d}\boldsymbol{\beta}_d$, suggesting a close relationship between the non-conservation of the Lagrangian two-form in continuous nonholonomic systems and its discrete counterpart, see \cite{GBYo2018b, PeYo2025}.
\end{remark}

\begin{remark}
In the holonomic case, the constraint distribution is integrable (i.e., $\Delta_Q = TN$, where $N\subset Q$), and there exist local functions $\phi^{r} \in C^{\infty}(Q)$ such that 
\[
N=\big\{ q \in Q \mid \phi^r(q)=0 \big\},
\]
and the constraint force is expressed as $\beta=\mu_{r}\mathbf{d}\phi^{r} \in \Gamma(\Delta_{Q}^{\circ})$ with $\mu_r \in C^\infty(Q)$. Since each $\phi^r$ is constant on $N$, the pullback of $\mathbf{d}\phi^r$ satisfies $\iota_{\Delta_Q}^*\tau_{Q}^{\ast}(\mathbf{d}\phi^r) = 0$. Consequently, the pullback of the lifted one-form $\boldsymbol{\beta} = \tau_Q^* \beta$ to $\Delta_Q$ vanishes identically, i.e.,
$\iota_{\Delta_Q}^{\ast} \boldsymbol{\beta} = (\iota_{\Delta_Q}^{\ast} \tau_{Q}^{\ast}\mu_r) \cdot \iota_{\Delta_Q}^{\ast} \tau_{Q}^{\ast}(\mathbf{d}\phi^r) = 0$, which immediately yields
\[
\mathbf{d}(\iota_{\Delta_Q}^* \boldsymbol{\beta}) = 0.
\]
This implies the preservation of the Lagrangian two-form.
\end{remark}

\section{Lagrange-Dirac dynamical systems on tangent bundles}
So far, we have developed the variational structures associated with Lagrangian mechanics for both unconstrained and nonholonomic systems. In the nonholonomic case, the dynamics is generally not symplectic, although its deviation is governed by a structural relation. This suggests the existence of some underlying geometric structure preserved beyond the symplectic framework. In this section, we formalize this idea by introducing a Dirac structure on $TQ$, referred to as a \textit{Lagrange--Dirac structure}, induced by the Lagrangian two-form $\Omega_{L}$ and a constraint distribution on $Q$. Here, the Lagrangian $L$ is allowed to be degenerate, so that $\Omega_{L}$ may be presymplectic.

We first formulate the corresponding Lagrange--Dirac dynamical system directly on $TQ$, yielding a first-order formulation of Lagrange--d'Alembert--Dirac equations. When $L$ is hyperregular, this formulation recovers the standard first-order form of the Lagrange--d'Alembert equations. 

In the degenerate case, a constraint algorithm is required to determine the final constraint submanifold. Therefore, in the following, we restrict our attention to the hyperregular case, so that the induced dynamics is well defined on the constraint manifold when analyzing dynamical properties such as structure preservation along the flow.

Now, we project the dynamics onto the admissible directions to obtain constrained Lagrange--Dirac dynamical systems. Moreover, we show that the apparent failure of preservation of the Lagrangian two-form in nonholonomic dynamics can be interpreted in terms of gauge equivalence of the underlying Lagrange--Dirac structures. This leads naturally to a gauge covariance property of Lagrange--Dirac dynamical systems along the constrained flow.
 
\subsection{Lagrange-Dirac structures on tangent bundles}\label{Sec:LagDiracStr}
\paragraph{Dirac structures.}
Following  \cite{CoWe1988}, we first recall the notion of a linear Dirac structure $D_{V}$ on a vector space $V$.
Denote by $\langle \cdot , \cdot \rangle$ the natural pairing between $V$ and its dual space $V^{\ast}$, and define a symmetric pairing on $V \oplus V^{\ast}$ by
\begin{equation*}
\langle \! \langle (v,\alpha),
(\bar{v},\bar{\alpha}) \rangle \!  \rangle
=\langle \alpha, \bar{v} \rangle
+\langle \bar{\alpha}, v \rangle,
\end{equation*}
for $(v,\alpha), (\bar{v},\bar{\alpha}) \in V \oplus  V^{\ast}$.
Then, a \textit{linear Dirac structure} on $V$ is a subspace $D_{V} \subset V \oplus V^{\ast}$ satisfying
\[
D_{V} = D_{V}^{\perp},
\]
where $D_{V}^{\perp}$ denotes the orthogonal complement of $D_{V}$ with respect to the symmetric pairing $\langle\!\langle \cdot , \cdot \rangle\!\rangle$.
\medskip

For a subbundle $D_M\subset TM\oplus T^*M$, if, for each $x\in M$, a subspace
\[
D_M(x)\subset T_xM\times T_x^*M
\]
is a linear Dirac structure on $T_{x}M$, then $D_M$ is called a Dirac structure on $M$.

As a typical example, a Dirac structure on $M$ is defined by a two-form $\Omega_M$ on $M$ and a distribution $\Delta_{M}$ on $M$ by
\begin{equation}\label{DiracManifold}
\begin{split}
D_{M}(x)=\big\{ (v_{x}, \alpha_{x}) \in T_{x}M \times T^{\ast}_{x}M
  \; \mid \; & v_{x} \in \Delta_{M}(x), \; \mbox{and} \\ 
  & \big<\alpha_{x}, w_{x} \big> = \Omega_M(x)(v_{x},w_{x}) \; \;
\mbox{for all} \; \; w_{x} \in \Delta_{M}(x) \big\}.
\end{split}
\end{equation}

If the condition
\[
\langle \pounds_{X_1} \alpha_2, X_3 \rangle
+\langle \pounds_{X_2} \alpha_3, X_1 \rangle
+\langle \pounds_{X_3} \alpha_1, X_2 \rangle = 0
\]
holds for all $(X_i, \alpha_i), \, i=1,2,3,$ taking values in $D_{M}$, then $D_{M}$ is said to be \textit{integrable}. In this paper, however, we are primarily concerned with Dirac structures that are not necessarily integrable, since nonholonomic constraint distributions are typically nonintegrable\footnote{
In some literature, a Dirac structure that does not satisfy the integrability condition is called an almost Dirac structure. Following the convention in nonholonomic mechanics, we simply refer to such structures as Dirac structures, as the nonintegrability of $\Delta_{M}$ is a central feature of the systems under study.
%
}.

\begin{remark}[\textbf{Courant bracket}]\rm  
Denote by $\Gamma(TM \oplus T^{\ast}M)$ the space of local sections of $TM \oplus T^{\ast}M$, 
which is endowed with the Courant bracket:
\[
\begin{split}
\left[(X_{1},\alpha_{1}),(X_{2},\alpha_{2})\right] &:= \left( \left[X_{1}, X_{2} \right],  \pounds_{X_1} \alpha_{2}- \pounds_{X_2} \alpha_{1} + \mathbf{d}\left< \alpha_{2}, X_{1} \right>\right)\\
&\;=\left( \left[X_{1}, X_{2} \right],  \mathbf{i}_{X_1} \mathbf{d}\alpha_{2}- \mathbf{i}_{X_2} \mathbf{d}\alpha_{1} +\mathbf{d} \left< \alpha_{2}, X_{1} \right>\right).
\end{split}
\]
This bracket is skew-symmetric but fails to satisfy the Jacobi identity. As shown in \cite{Dorfman1993}, a Dirac structure $D_{M} \subset TM \oplus T^{\ast}M$ is integrable if and only if it is closed under the Courant bracket,
\[
\left[\Gamma(D_{M}), \Gamma(D_{M}) \right] \subset \Gamma(D_{M}).
\]
However, the Courant bracket will not play an explicit role in the developments below.
\end{remark}

\paragraph{Induced Dirac structures on the cotangent bundles.}
In mechanics, one of the fundamental examples is the induced Dirac structure $D_{\Delta_Q}$ on the cotangent bundle $T^{\ast}Q$, constructed from a constraint distribution $\Delta_Q$ on $Q$ and the canonical symplectic form $\Omega_{T^{\ast}Q}$; see \cite{YoMa2006a}.

As a primary example of the construction \eqref{DiracManifold}, a Dirac structure $D_{\Delta_Q}$ on $T^{\ast}Q$, which is induced from a distribution $\Delta_{Q} \subset TQ$, is defined by the subbundle of $T  T^{\ast}Q \oplus T ^{\ast} T^{\ast}Q$ whose fiber is given, at each $p_{q} \in
T^{\ast}Q$, as
\begin{align}\label{IndDiracStrCot}
D_{\Delta_Q}(p_{q})
& =\{ (v_{p_{q}}, \beta_{p_{q}}) \in T_{p_{q}}T^{\ast}Q \times
T^{\ast}_{p_{q}}T^{\ast}Q  \mid v_{p_{q}} \in
\Delta_{T^{\ast}Q}(p_{q}),  \; \mbox{and} \;  \nonumber
\\ & \qquad \qquad
\left<\beta_{p_{q}}, w_{p_{q}} \right> = \Omega_{T^{\ast}Q}( p _q) (v_{p_{q}},w_{p_{q}}) \;\; \mbox{for
all} \;\; w_{p_{q}} \in \Delta_{T^{\ast}Q}(p_{q})\},
\end{align}
where $\Delta_{T^{\ast}Q}=( T\pi_{Q})^{-1} \, (\Delta_{Q}) \subset TT^{\ast}Q$ is the lifted distribution on $T^{\ast}Q$ via the tangent map of the canonical projection $\pi_{Q}:T^{\ast}Q \to Q$.

Such induced Dirac structures play an essential role in the formulation of the dynamics of nonholonomic mechanics, electric circuits, fluids as well as nonequilibrium thermodynamic systems in the context of {\it implicit Lagrangian systems} or \textit{Lagrange--Dirac systems on the cotangent bundle}, which may allow the cases of degenerate Lagrangians as shown in \cite{YoMa2006a, YoMa2006b, YoMa2007a, YoMa2007b, YoMa2009, YoGB2025, GBYo2015, GBYo2018}. 
\medskip

From here on, we focus on a \textit{Dirac structure on the tangent bundle $TQ$} induced by a Lagrangian $L: TQ \to \mathbb{R}$, together with the associated \textit{Lagrange--Dirac dynamical system on $TQ$}.

\paragraph{Lagrange-Dirac structures on the tangent bundles.} 
We now construct an induced Dirac structure on the tangent bundle $TQ$ from a Lagrangian two-form $\Omega_L$ and a distribution $\Delta_Q$ on $Q$.

Denote by $\langle \cdot , \cdot \rangle$ the paring between $T_{v_{q}}TQ$ and its dual space $T^{\ast}_{v_{q}}TQ$ at each $v_{q} \in TQ$, and define a symmetric pairing $\langle \! \langle \cdot, \cdot \rangle\!\rangle$ on $TTQ \oplus T^{\ast}TQ$ by, for each $v_{q}$, 
\begin{equation}\label{SymPairing}
\langle \! \langle (w_{v_{q}},\alpha_{v_{q}}),
(\bar{w}_{v_{q}},\bar{\alpha}_{v_{q}}) \rangle \!  \rangle
=\langle \alpha_{v_{q}}, \bar{w}_{v_{q}} \rangle
+\langle \bar{\alpha}_{v_{q}}, w_{v_{q}} \rangle,
\end{equation}
for $(w_{v_{q}},\alpha_{v_{q}}), (\bar{w}_{v_{q}},\bar{\alpha}_{v_{q}}) \in T_{v_{q}}TQ \times  T^{\ast}_{v_{q}}TQ$.

\begin{theorem}
Let  $L: TQ \to \mathbb{R}$ be a Lagrangian, possibly degenerate. Associated with $\Delta_{Q} \subset TQ$, we can define the distribution $\Delta_{TQ}$ on $TQ$ by
\[
\Delta_{TQ}=(T\tau_{Q})^{-1}(\Delta_{Q})\subset TTQ.
\]
Then, the subbundle $D_{L}\subset TTQ\oplus T^*TQ$ that is defined by, for each $v_q \in TQ$,
\begin{align}\label{LagDirac}
D_{L}(v_q)
& =\big\{ (w_{v_{q}}, \alpha_{v_{q}}) \in T_{v_q}TQ \times
T^{\ast}_{v_q}TQ\mid w_{v_{q}} \in
\Delta_{TQ}(v_q), \;   \mbox{and}\nonumber\\
&\hspace{1cm}
\langle \alpha_{v_{q}}, u_{v_{q}}\rangle = \Omega_L(v_q) (w_{v_{q}},u_{v_{q}}), \;\textrm{for all} \;\; u_{v_{q}} \in \Delta_{TQ}(v_q)\big\},
\end{align}
is a Dirac structure on $TQ$.
\end{theorem} 
\begin{proof}
The orthogonal complement of $D_{L} \subset TTQ \oplus
T^{\ast}TQ$ with respect to the symmetric pairing $\langle \! \langle \cdot, \cdot \rangle\!\rangle$ is given, at the point $v_{q} \in TQ$, by
\[
\begin{split}
D^{\perp}_{L}(v_{q})
=\{ (u_{v_{q}}, \beta_{v_{q}}) \in T_{v_{q}}TQ \times &
T^{\ast}_{v_{q}}TQ  \mid
\left<\alpha_{v_q},u_{v_q}\right> + \left<\beta_{v_q}, y_{v_q}\right> = 0
\;\; \text{for all } (y_{v_q}, \alpha_{v_q}) \in D_L(v_q)\}.
\end{split}
\]
We first show that
$D_{L}(v_{q}) \subset D^{\perp}_{L}(v_{q})$. Let
$(y_{v_{q}}, \alpha_{v_{q}}) \in D_{L}(v_{q})$ and $(y_{v_{q}}^{\prime},
\alpha_{v_{q}}^{\prime}) \in D_L(v_{q})$. Then,
\[
\langle \alpha_{v_{q}}, y_{v_{q}}^{\prime}\rangle +\langle \alpha_{v_{q}}^{\prime}, y_{v_{q}}\rangle 
=\Omega_{L}(v_{q})(y_{v_{q}},y_{v_{q}}^{\prime})
+\Omega_{L}(v_{q})(y_{v_{q}}^{\prime},y_{v_{q}})=0.
\]
Therefore,
$D_{L}(v_{q}) \subset D^{\perp}_{L}(v_{q})$.
\medskip

Next, let us check $D^{\perp}_{L}(v_{q})\subset D_{L}(v_{q})$.
Let $(u_{v_{q}}, \beta_{v_{q}}) \in D^{\perp}_{L}(v_{q})$. By the definition of the orthogonal complement,
\[
\langle \alpha_{v_q}, u_{v_q} \rangle + \langle \beta_{v_q}, y_{v_q} \rangle = 0 \quad \text{for all } (y_{v_q}, \alpha_{v_q}) \in D_L(v_q).
\]
First, consider the case where $y_{v_q} = 0$. For $(0, \alpha_{v_q})$ to be in $D_L(v_q)$, the condition $\langle \alpha_{v_q}, w_{v_q} \rangle = \Omega_L(0, w_{v_q}) = 0$ must hold for all $w_{v_q} \in \Delta_{TQ}(v_q)$. This implies $\alpha_{v_q} \in \Delta_{TQ}^\circ(v_q)$. Since $(u_{v_q}, \beta_{v_q})$ is orthogonal to all such $(0, \alpha_{v_q})$, we have $\langle \alpha_{v_q}, u_{v_q} \rangle = 0$ for all $\alpha_{v_q} \in \Delta_{TQ}^\circ(v_q)$, which proves $u_{v_q} \in \Delta_{TQ}(v_q)$.

It remains to verify that
$\langle\beta_{v_{q}}, y_{v_{q}} \rangle= \Omega_{L}(v_{q})(u_{v_{q}}, 
y_{v_{q}})$ for all $y _{v_{q}} \in \Delta_{TQ}(v_{q})$. Thus, let $y _{v_{q}} 
\in \Delta_{TQ}(v_{q}) $ be arbitrary and choose $\alpha _{v_{q}} $ 
satisfying  $\langle\alpha_{v_{q}}, w_{v_{q}}\rangle
=\Omega_{L}(v_{q})(y_{v_{q}},w_{v_{q}})$ for all $w_{v_{q}} \in 
\Delta_{TQ}(v_{q})$. From 
$\langle \alpha_{v_{q}}, u_{v_{q}}\rangle+\langle\beta_{v_{q}}, y_{v_{q}} \rangle=0$ 
and the fact that we have already proved that
$u _{v_{q}} \in
\Delta_{TQ}(v_{q})$, we get
\[
\Omega_{L}(v_{q})(y_{v_{q}},u_{v_{q}})+\langle\beta_{v_{q}}, y_{v_{q}} \rangle=0 \quad \mbox{for all}
\quad y_{v_{q}} \in \Delta_{TQ}(v_{q}),
\]
that is, $\langle\beta_{v_{q}}, y_{v_{q}} \rangle= \Omega_{L}(v_{q})(u_{v_{q}}, y_{v_{q}})$
for all $y _{v_{q}} \in \Delta_{TQ}(v_{q})$. Thus, $(u _{v_{q}}, \beta _{v_{q}}) \in 
D_{L}(v_{q})$ and hence 
$
D_{L}^{\perp} \subset 
D_{L}
$ 
as required. Therefore,
$D_L^\perp = D_L$.
\end{proof}

\begin{framed}
\begin{proposition}\label{Prop:CondLagDiracSys}
For the Lagrange--Dirac structure \eqref{LagDirac}, let us express the condition 
\begin{equation}\label{CondLagDiracStruct}
(w_{v_{q}}, \alpha_{v_{q}}) \in D_{L}(v_{q}),
\end{equation}
by using the local coordinates $v_{q}=(q,v) \in TQ$, $w_{v_{q}}=(q,v, \xi,\eta),u_{v_{q}}=(q,v,\zeta,\nu) \in T_{v_{q}}TQ$, and $\alpha_{v_{q}}=(q,v,\beta, \gamma)\in T^{\ast}_{v_{q}}TQ$. Namely, the condition,
\[
\left( (q,v,\xi,\eta), (q,v, \beta, \gamma)\right) \in D_{L}(q,v)
\]
is equivalent to 
\begin{equation}\label{LagDiracCond}
\begin{split}
& \gamma=\mathbf{D}_{2}\mathbf{D}_{2} L(q,v)\cdot \xi, 
\quad \xi \in \Delta_{Q}(q),\\[2mm]
&\displaystyle\vspace{0.2cm} \beta \cdot \zeta
-\mathbf{D}_1 (\mathbf{D}_2 L(q,v) \cdot \xi) \cdot  \zeta  + \mathbf{D}_1 (\mathbf{D}_2 L(q,v) \cdot \zeta)\cdot \xi  \\
&\hspace{4.5cm}+  \mathbf{D}_{2}(\mathbf{D}_{2} L(q,v) \cdot \zeta) \cdot  \eta= 0, \quad \forall \zeta \in \Delta_Q(q).
\end{split}
\end{equation}
\end{proposition}
\end{framed}
\begin{proof}
It follows from the condition \eqref{CondLagDiracStruct} that we have
\[
\left<\alpha_{v_{q}}, u_{v_{q}}\right> = \Omega_L(v_q) (w_{v_{q}},u_{v_{q}}).
\]
Using the local coordinates $(q,v)$ for $v_{q} \in TQ$, $(q,v, \xi,\eta)$ for $w_{v_{q}} \in T_{v_{q}}TQ$, $(q,v,\zeta,\nu)$ for $u_{v_{q}}\in T_{v_{q}}TQ$, $(q,v,\beta, \gamma)$ for $\alpha_{v_{q}}\in T^{\ast}_{v_{q}}TQ$, the right-hand side is expressed by
\begin{equation*}
\begin{split}
\Omega _L(q,v)((\xi,\eta), (\zeta,\nu)) 
&=\mathbf{D}_1 (\mathbf{D}_2 L(q,v) \cdot \xi) \cdot  \zeta  - \mathbf{D}_1 (\mathbf{D}_2 L(q,v) \cdot \zeta)\cdot \xi \\
& \qquad+ \mathbf{D}_{2}(\mathbf{D}_{2} L(q,v) \cdot \xi) \cdot \nu  -  \mathbf{D}_{2}(\mathbf{D}_{2} L(q,v) \cdot \zeta) \cdot  \eta,
\end{split}
\end{equation*}
while the left-hand side is 
\[
\left<\alpha_{v_{q}}, u_{v_{q}}\right> = \left<(q,v,\beta, \gamma), (q,v,\zeta,\nu)\right> = \beta \cdot \zeta + \gamma \cdot \nu.
\]
Comparing both sides, we obtain
\begin{equation}\label{VarCondDiracLagSys}
\begin{split}
&\mathbf{D}_1 (\mathbf{D}_2 L(q,v) \cdot \xi) \cdot  \zeta  - \mathbf{D}_1 (\mathbf{D}_2 L(q,v) \cdot \zeta)\cdot \xi \\
& \qquad+ \mathbf{D}_{2}(\mathbf{D}_{2} L(q,v) \cdot \xi) \cdot \nu  -  \mathbf{D}_{2}(\mathbf{D}_{2} L(q,v) \cdot \zeta) \cdot  \eta= \beta \cdot \zeta + \gamma \cdot \nu
\end{split}
\end{equation}
for all $u_{v_q}=(q,v,\zeta,\nu) \in \Delta_{TQ}(v_q)$, 
that is, for all $\zeta \in \Delta_Q(q)$ and all $\nu$, 
together with $\xi \in \Delta_Q(q)$ because $w_{v_q}=(q,v, \xi,\eta) \in \Delta_{TQ}(v_q)$.
\medskip

By arbitrariness of $\nu$, we obtain
\[
\gamma=\mathbf{D}_{2}\mathbf{D}_{2} L(q,v) \cdot \xi, \quad \xi \in \Delta_Q(q),
\]
where $\mathbf{D}_{2}\mathbf{D}_{2} L(q,v)$ is viewed as a linear map acting on $\xi$. We also obtain
\[
\beta \cdot \zeta
-\mathbf{D}_1 (\mathbf{D}_2 L(q,v) \cdot \xi) \cdot  \zeta  + \mathbf{D}_1 (\mathbf{D}_2 L(q,v) \cdot \zeta)\cdot \xi  +  \mathbf{D}_{2}(\mathbf{D}_{2} L(q,v) \cdot \zeta) \cdot  \eta= 0, \quad \forall \zeta \in \Delta_Q(q).
\]
Thus we conclude \eqref{LagDiracCond}.
\end{proof}

\paragraph{Finite-dimensional cases.}

In the finite dimensional case where $\Delta_{Q} \subset TQ$ is given in \eqref{ConstraintDistribution}, the condition for the Lagrange--Dirac system in \eqref{VarCondDiracLagSys} is denoted by
\begin{equation*}
\bigg(\gamma_{i} - \frac{\partial^2 L}{\partial v^{i} \partial v^{j} } \xi^{j} \bigg) \nu^{i}
+
\bigg\{\beta_{i} - \left( \frac{\partial^2 L}{ \partial q^i\partial v^{j}}- \frac{\partial^2 L}{\partial q^j \partial v^{i}}\right)\xi^j + \frac{\partial^2 L}{\partial v^{j} \partial v^{i}} \eta^{j} \bigg\}\zeta^i=0,
\end{equation*}
for all $\nu$ and all variations $\zeta^i$ satisfying the variational constraints
\[
\omega_{i}^{r}(q)\zeta^{i}=0,\quad r=1,\ldots,m <n,
\]
together with the kinematic constraints
\[
\omega_{i}^{r}(q)\xi^{i}=0,\quad r=1,\ldots,m <n.
\]
Then, using Lagrange multipliers $\mu_{r}$, we obtain 
\begin{equation*}
\left\{
\begin{aligned}
&\;\gamma_{i} - \frac{\partial^2 L}{\partial v^{i}\partial v^{j}} \xi^{j}=0, \\[3mm]
&\;\beta_{i}  - \left( \frac{\partial^2 L}{\partial q^i\partial v^{j}}  - \frac{\partial^2 L}{\partial q^j\partial v^{i}}\right)\xi^j+ \frac{\partial^2 L}{\partial v^{j}\partial v^{i}}\eta^{j}=\mu_{r}\omega^{r}_{i},\\[3mm]
&\omega_{i}^{r}(q)\xi^{i}=0,\quad r=1,\ldots,m <n.
\end{aligned}
\right.
\end{equation*}
%

Note that the induced Dirac structure $D_{L}$ is not a natural geometric object on the tangent bundle $TQ$, since it depends explicitly on the choice of Lagrangian $L:TQ \to \mathbb{R}$.  In this sense, we refer to $D_{L}$ as a \textbf{\textit{Lagrange-Dirac structure}}, or a \textbf{\textit{Lagrange-Dirac bundle}}, on the tangent bundle $TQ$.

\paragraph{Forward and backward Dirac maps.}
Let us introduce the notions of \textit{forward and backward Dirac maps} between two distinct Dirac structures, following \cite{BurRad2003, BuCr2005, YoMa2007a}.
\medskip

Let $V$ and $W$ be vector spaces, and denote by $\mathrm{Dir}(V)$ and $\mathrm{Dir}(W)$ the spaces of Dirac structures on $V$ and $W$ respectively. For a linear map $\varphi:V \to W$, the {\it forward Dirac map} $\mathcal{F}\varphi: \mathrm{Dir}(V) \to \mathrm{Dir}(W)$ is defined by
\[
D_{W}=\{ (\varphi(x),\eta) \mid x \in V,\; \eta \in W^{\ast}, \; (x,\varphi^{\ast}\eta) \in D_{V}\},
\]
where $\varphi^{\ast} : W ^{\ast} \rightarrow V ^{\ast}$ is the dual of $\varphi$, and we write 
\[
D_{W} = \mathcal{F}\varphi(D_{V}) = \varphi_{\ast}D_{V}.
\]
On the other hand, a {\it backward Dirac map} $\mathcal{B}\varphi:  \mathrm{Dir}(W) \to \mathrm{Dir}(V)$ is defined by
\[
D_{V} = \{ (x,\varphi^{\ast}\eta) \mid x \in V,\; \eta \in W^{\ast}, \; (\varphi(x),\eta) \in D_{W}\},
\]
which we denote by
\[
D_{V} = \mathcal{B}\varphi(D_{W}) = \varphi^{\ast}D_{W}.
\]

For the case of manifolds, recall that a Dirac structure $D_{M}$ on a manifold $M$ is a subbundle of $TM \oplus T^{\ast}M$ such that $D_{M}(x)$ is a 
linear Dirac structure on $T_{x}M$ at each point $x \in M$. Therefore, one can regard $D_{M}$ as a {\it Dirac structure on the (tangent) vector bundle $\tau_{M}:TM \rightarrow M$}, rather than a Dirac structure on the base manifold $M$. This viewpoint is essential in the context of Dirac bundle reduction, as will be shown. In this context, we denote the spaces of Dirac structures on  $M$ and $N$ by $\mathrm{Dir}(TM)$ and $\mathrm{Dir}(TN)$ respectively\footnote{In the literature, the space of Dirac structures on a manifold $M$ is often denoted by $\mathrm{Dir}(M)$ rather than $\mathrm{Dir}(TM)$.}, in order to emphasize the fiberwise linear structures over the tangent bundles.

Let $\varphi: M \to N$ be a smooth map. Then the tangent map defines a bundle morphism $T\varphi: TM \to TN$ satisfying the commutative relation:
$
\tau_{N}\circ T\varphi=\varphi\circ \tau_{M}.
$
\[
  \xymatrix{
    TM \ar[r]^{T\varphi} \ar[d]_{\tau_{M}} & TN \ar[d]_{\tau_{N}} \\
    M \ar[r]_{\varphi} & N
  }
\]
Then, a {\it forward Dirac map} $\mathcal{F}(T\varphi): \mathrm{Dir}(TM) \to \mathrm{Dir}(TN)$ is defined by, for all $x \in M$, 
\begin{equation*}
\begin{split}
D_{N}(\varphi(x))= \big\{ ( T_{x}\varphi(w_{x}),\alpha_{\varphi(x)} ) \mid w_{x} \in T_{x}M,\;  \alpha_{\varphi(x)} & \in T^{\ast}_{\varphi(x)}N,\; \\
& (w_{x}, (T_{x}\varphi_{x})^{\ast}(\alpha_{\varphi(x)}))  \in D_{M}(x) \big\},
\end{split}
\end{equation*}
where $T_{x}\varphi: T_{x}M \to T_{\varphi(x)}N$ is the tangent map of $\varphi: M \to N$ at $x \in M$ and $(T_{x}\varphi)^{\ast}: T^{\ast}_{\varphi(x)}N \to T^{\ast}_{x}M$ is the dual map of $T_{x}\varphi$. Then, we write the relationship between the two Dirac structures as
\[
D_{N}=\mathcal{F}(T\varphi)(D_{M})=\varphi_{\ast}D_{M},
\] 
where $\varphi_{\ast}D_{M}$ is the {\it pushforward} of the Dirac structure $D_{M}$ by $\varphi: M \to N$.
\medskip

Similarly, for a smooth map $\varphi: M \to N$,  a {\it backward Dirac map} $\mathcal{B}(T\varphi): \mathrm{Dir}(TN) \to \mathrm{Dir}(TM)$ is defined by, for all $x \in M$, 
\begin{equation*}
\begin{split}
D_{M}(x)=\big\{ ( w_{x},(T_{x}\varphi)^{\ast}\alpha_{\varphi(x)}) \mid w_{x} \in T_{x}M,\;  \alpha_{\varphi(x)} & \in T^{\ast}_{\varphi(x)}N,\; \\
& (T_{x}\varphi(w_{x}), \alpha_{\varphi(x)})  \in D_{N}(\varphi(x)) \big\},
\end{split}
\end{equation*}
and we write the relationship between the two Dirac structures as
\[
D_{M}=\mathcal{B}(T\varphi)(D_{N})=\varphi^{\ast}D_{N},
\]
where $D_{M}$ is the {\it pullback} of $D_{N}$ by  $\varphi: M \to N$. If $\varphi$ is a diffeomorphism, then 
the forward Dirac map is the inverse of the backward Dirac map.
\medskip

In the present context, we emphasize that the Lagrange--Dirac structure $D_{L}$ on $TQ$ can be defined from the induced Dirac structure $D_{\Delta_Q}$ on $T^{\ast}Q$ even for degenerate Lagrangians by using the backward Dirac map $\mathcal{B}(T\mathbb{F}L): \mathrm{Dir}(TT^{\ast}Q) \to \mathrm{Dir}(TTQ)$ associated with the Legendre transformation $\mathbb{F}L: TQ \to T^{\ast}Q$. Specifically, we define
\begin{equation*}
D_{L} = \mathcal{B}(T\mathbb{F}L)(D_{\Delta_{Q}}) = (\mathbb{F}L)^{\ast} D_{\Delta_Q}.
\end{equation*}
This allows us to treat degenerate Lagrangian systems in the context of Lagrange--Dirac dynamical systems. Locally, at each $v_q \in TQ$ with $p_q = \mathbb{F}L(v_q)$, this is expressed as
\begin{equation*}
\begin{split}
D_{L}(v_{q})= \big\{ (w_{v_{q}},(T_{v_{q}}\mathbb{F}L)^{\ast}(\beta_{p_{q}}) ) \, \mid \, 
(T_{v_{q}}\mathbb{F}L(w_{v_{q}}), \beta_{p_{q}})  \in D_{\Delta_{Q}}(p_{q}) \big\}.
\end{split}
\end{equation*}
Conversely, in the case where $L$ is hyperregular, $\mathbb{F}L$ is a diffeomorphism. In this setting, the forward and backward maps are equivalent, and the induced Dirac structure $D_{\Delta_Q}$ on $T^*Q$ can be recovered from $D_L$ via the forward Dirac map:
\begin{equation*}
D_{\Delta_{Q}} = \mathcal{F}(T\mathbb{F}L)(D_{L}) = (\mathbb{F}L)_{\ast} D_{L},
\end{equation*}
which is locally represented by, for all $v_{q} \in TQ$ and $p_{q}=\mathbb{F}L(v_{q})$, 
\begin{equation*}
\begin{split}
D_{\Delta_{Q}}(p_{q}) = \big\{ (T_{v_{q}}\mathbb{F}L(w_{v_{q}}),\beta_{p_{q}} ) \mid 
 (w_{v_{q}}, (T_{v_{q}}\mathbb{F}L)^{\ast}(\beta_{p_{q}}))  \in D_{L}(v_{q}) \big\}.
\end{split}
\end{equation*}
%


\subsection{Lagrange--Dirac dynamical systems on tangent bundles}\label{Sec:LagDiracSys_TangentBundle}
\paragraph{Lagrange--Dirac dynamical systems on $TQ$.}
The Lagrange--Dirac structure on $TQ$ is defined in terms of the presymplectic form $\Omega_L$ induced by the Lagrangian $L$ and a constraint distribution $\Delta_Q$, and it enables the treatment of the degenerate cases. Then, we obtain an implicit representation of the dynamics for Lagrange--Dirac dynamical systems on $TQ$, in which the Lagrange--Dirac structure effectively captures the interplay between the degeneracy of the Lagrangian and the nonholonomic constraints. In the following, we provide a concrete characterization of these systems through the corresponding equations of motion with their structural properties.
\begin{framed}
\begin{proposition}\label{Prop: IntLDA_Eqn}
Given a Lagrangian (possibly degenerate) $L: TQ \to \mathbb{R}$, define the energy $E_{L}$ on $TQ$ by, for $v_q=(q,v) \in TQ$,
\begin{equation*}
E_{L}(q,v)=\mathbb{F}L(v_q) \cdot v_q -L(q,v).
\end{equation*}
Let $D_L \subset TTQ \oplus T^*TQ$ be the Lagrange--Dirac structure given in \eqref{LagDirac}.
\medskip

Then, the following statements hold:
\begin{itemize}
\item[(i)]
If a curve $(q(t),v(t)), \; t \in [0,T]$ on  $TQ$  satisfies the condition:
\begin{equation}\label{LDDS}
\big((\dot q(t),\dot v(t)), \mathbf{d}E_{L}(q(t),v(t))\big) \in D_{L}(q(t),v(t)),
\end{equation}
then the curve is a solution curve of the {\it intrinsic Lagrange--d'Alembert--Dirac equations}:
\begin{equation}\label{intrinsic_LagrangeDiracdynamicalSystem}
\left\{
\begin{array}{l}
\displaystyle \mathbf{i}_{(\dot q(t),\dot v(t))}\Omega_{L}(q(t),v(t))-\mathbf{d}E_{L}(q(t),v(t)) \in \Delta_{TQ}^{\circ}(q(t),v(t)),\\[4mm]
\displaystyle (\dot q(t),\dot v(t)) \in \Delta_{TQ}(q(t),v(t)).
\end{array} 
\right.
\end{equation} 

\item[(ii)]
The implicit dynamics for the Lagrange--d'Alembert--Dirac equations are given by:
\begin{equation}\label{LagrangeDiracCond}
\begin{cases}
&\displaystyle\vspace{0.2cm} \frac{d}{dt} \mathbf{D}_2 L(q,v) - \mathbf{D}_{1}L(q,v) -\mathbf{D}_{1}\mathbf{D}_{2}L(q,v)\cdot \left(\dot{q}-v\right)
 \in \Delta_{Q}^\circ(q),\\[2mm]
& \mathbf{D}_{2}\mathbf{D}_{2} L(q,v)\cdot (\dot{q}-v)=0, \\[2mm]
& \dot{q} \in \Delta_{Q}(q).
\end{cases}
\end{equation}

\item[(iii)]For the hyperregular case, we recover the \textit{first-order system of the Lagrange--d'Alembert equations}:
\begin{equation}\label{LDA-DiracEqn} 
\frac{d}{dt} \mathbf{D}_2 L(q,v)-\mathbf{D}_1 L(q,v) \in \Delta_Q^{\circ}(q),\qquad  \dot q=v,\qquad \dot{q} \in \Delta_{Q}(q).
\end{equation}
\end{itemize}

\end{proposition}
\end{framed}
\begin{proof}
\begin{itemize}
\item[(i)]
Using \eqref{LagDirac} and \eqref{LDDS}, we have
\[
\Omega_L(v_q)  \left((\dot q,\dot v), w_{v_{q}} \right) =\mathbf{d} E_{L}(q,v) \cdot w_{v_{q}},
\]
for all $w_{v_{q}} \in \Delta_{TQ}(q,v)$. Then, it follows that we obtain the intrinsic Lagrange--d'Alembert--Dirac equations:
\begin{equation*}
\mathbf{i}_{(\dot q,\dot v)} \, \Omega _L(q,v)-\mathbf{d} E_{L}(q,v) \in  \Delta_{TQ}^{\circ}(q,v),
\end{equation*}
together with the kinematic constraints $(\dot q,\dot v) \in \Delta_{TQ}(q,v)$.

\item[(ii)]
Let us derive the local expressions for the Lagrange--d'Alembert--Dirac equations \eqref{LagrangeDiracCond} from \eqref{intrinsic_LagrangeDiracdynamicalSystem}. Recall the energy $E_{L}$ is defined by, for $(q,v) \in TQ$, 
\[
E_{L}(q,v) = \mathbb{F}L(q,v) \cdot v  - L(q,v) = \mathbf{D}_{2}L(q,v)\cdot v-L(q,v),
\]
and also that the differential of $E_{L}$ is locally denoted such that
\[
\mathbf{d}E_{L}(q,v) \cdot  (\delta{q}, \delta{v}) =\mathbf{D}_{1}E_{L}(q,v)\cdot \delta q+\mathbf{D}_{2}E_{L}(q,v)\cdot \delta v,\\
\]
where
$\mathbf{D}_{1}E_{L}(q,v)=\mathbf{D}_{1}\mathbf{D}_{2}L(q,v)\cdot v - \mathbf{D}_{1}L(q,v)$,  and $\mathbf{D}_{2}E_{L}(q,v)= \mathbf{D}_{2}\mathbf{D}_{2} L(q,v)\cdot v$.
\medskip

Using Proposition \ref{Prop:CondLagDiracSys}, the Lagrange--Dirac condition \eqref{LDDS} yields 
an \textit{implicit dynamics of the Lagrange--d'Alembert--Dirac systems on $TQ$} as in \eqref{LagrangeDiracCond}:
\begin{equation*}
\begin{cases}
&\displaystyle\vspace{0.2cm} \frac{d}{dt} \mathbf{D}_2 L(q,v)- \mathbf{D}_{1}L(q,v) 
-  \mathbf{D}_{1}\mathbf{D}_{2}L(q,v)\cdot \left(\dot{q}-v\right) \in \Delta_{Q}^\circ(q),\\[2mm]
& \mathbf{D}_{2}\mathbf{D}_{2} L(q,v)\cdot (\dot{q}-v)=0, \\[2mm]
& \dot{q} \in \Delta_{Q}(q).
\end{cases}
\end{equation*}
To derive the first equation of \eqref{LagrangeDiracCond}, we use the chain rule along $(q(t),v(t))$:
\[
\frac{d}{dt}\big(\mathbf{D}_2 L(q,v)\cdot \delta q\big)
=
\mathbf{D}_1 (\mathbf{D}_2 L(q,v)\cdot \delta q)\cdot \dot q
+
\mathbf{D}_2 (\mathbf{D}_2 L(q,v)\cdot \delta q)\cdot \dot v.
\]
Here, $\delta q$ is treated as a fixed variation along the curve, i.e., independent of time, so that $\frac{d}{dt}(\delta q)=0$. Hence, 
\[
\frac{d}{dt}\left( \mathbf{D}_2 L(q,v) \cdot \delta{q}\right)=\frac{d}{dt}\left( \mathbf{D}_2 L(q,v) \right)\cdot \delta{q}.
\]

The Lagrange--Dirac dynamics \eqref{LagrangeDiracCond} allows degenerate Lagrangian systems together with nonholonomic constraints. In fact, the second equation of \eqref{LagrangeDiracCond} describes the failure of the second-order condition that reflects the constraints due to the degeneracy of Lagrangian and the third equation the nonholonomic constraints imposed on the motion.

\item[(iii)]
In the case in which the Lagrangian is hyperregular, namely, $\mathbf{D}_{2}\mathbf{D}_{2} L $ is globally nondegenerate, it follows from the second equation of \eqref{LagrangeDiracCond} that 
\[
\dot{q}=v.
\]
Substituting this relation into the first equation of \eqref{LagrangeDiracCond} leads to the \textit{first-order system of the conventional Lagrange--d'Alembert equations} \eqref{LDA-DiracEqn}:
\begin{equation*}
\begin{split}
 \frac{d}{dt}\mathbf{D}_2 L(q,v) - \mathbf{D}_{1}L(q,v) \in \Delta_Q^\circ(q),\qquad \dot{q}=v, \qquad \dot{q}\in \Delta_{Q}(q).
\end{split}
\end{equation*}
Now we call a triple $(E_{L}, D_{L})$ that satisfies the condition \eqref{LDDS} a {\it Lagrange--Dirac dynamical system on the velocity phase space $TQ$}. Specifically, for the case where there is no constraint, i.e., $\Delta_Q=TQ$, the Lagrange-Dirac dynamical system recovers the {\it Euler-Lagrange equations} \eqref{EulLagEqn} in Theorem \ref{Thm:LagSys}.
\end{itemize}

This completes the proof. 
\end{proof}
This implicit representation of the dynamics, which we call the \textit{Lagrange--Dirac dynamics}, provides a unified framework where the regularity of the Lagrangian and hence the second-order condition are not assumed a priori. Instead, they emerge naturally as properties of the specific dynamics under consideration.

\begin{remark}
The term $\mathbf{D}_{1}\mathbf{D}_{2}L(q,v)\cdot \left(\dot{q}-v\right)$ appearing in the first equation in \eqref{LagrangeDiracCond} reflects the internal coupling between the velocity $\dot q$ and the auxiliary variable $v$ in the presence of degeneracy. It ensures the consistency of the implicit dynamics on $TQ$ when the second-order condition $\dot{q}=v$ is not enforced.
\end{remark}

\begin{remark}
The second equation in \eqref{LagrangeDiracCond} implies the constraint due to the degeneracy of $L$; the mismatch between $\dot{q}$ and $v$ must lie in the kernel of the Hessian $\mathbf{D}_{2}\mathbf{D}_{2}L(q,v)$:
\[
\dot q - v \in \operatorname{ker} \big(\mathbf{D}_{2}\mathbf{D}_{2}L(q,v)\big).
\]
This can be interpreted as a tangent-bundle counterpart of the Dirac primary constraint on the cotangent bundle $T^{\ast}Q$, 
which restricts the momentum variables to the image of the Legendre transform
\[
p \in \operatorname{Im}(\mathbb{F}L) \subset T^{\ast}Q,
\]
arising from the degeneracy of the Legendre transformation.

Indeed, both conditions originate from the degeneracy of the Legendre transformation: 
the cotangent-bundle constraint restricts admissible momenta, whereas the tangent-bundle condition restricts admissible second-order directions through the kernel of the fiber Hessian.
\medskip

Thus, these two conditions encode the degeneracy of Lagrangians from dual bundle perspectives.
\end{remark}

\paragraph{Finite-dimensional cases.}
From the Lagrange--Dirac condition
\[
( (\dot q,\dot v),\mathbf dE_L(q,v))\in D_L(q,v),
\]
it follows that
\[
\mathbf dE_L(q,v) \cdot (\delta q,\delta v) = \Omega_L(q,v)((\dot q,\dot v),(\delta q,\delta v)),
\]
where $(\delta q,\delta v)\in\Delta_{TQ}(q,v)$ is an admissible variation, so in particular $\delta q\in\Delta_Q(q)$.
\medskip

In finite dimensions, using local coordinates $(q^{i},v^{i})$ for $(q,v) \in TQ$, we have
\[
E_{L}(q,v) = \frac{\partial L}{\partial v^{i}}v^{i} -L(q,v),
\]
and hence the differential of $E_{L}$ is denoted by
\[
\mathbf{d}E_{L}(q,v)=\frac{\partial{E_{L}}}{\partial{q^{i}}}dq^{i}+\frac{\partial{E_{L}}}{\partial{v^{i}}}dv^{i},
\]
where
\begin{equation*}
\begin{aligned}
\frac{\partial{E_{L}}}{\partial{q^{i}}}&=\frac{\partial^2 L}{\partial q^{i} \partial v^j} \, v^j -  \frac{\partial L}{\partial q^{i}}, \qquad
\frac{\partial{E_{L}}}{\partial{v^{i}}}&= \frac{\partial^2 L}{\partial v^{i}\partial v^{j}} v^{j}.
\end{aligned}
\end{equation*}
\medskip

A direct computation gives
\[
\begin{split}
dE_L(q,v) \cdot (\delta q,\delta v) &= \frac{\partial{E_{L}}}{\partial{q^{i}}}\,\delta q^{i}+\frac{\partial{E_{L}}}{\partial{v^{i}}}\,\delta v^{i}\\
&=\left(\frac{\partial^2 L}{\partial q^{i} \partial v^j} \, v^j -  \frac{\partial L}{\partial q^{i}}\right) \delta q^{i}
+ \left( \frac{\partial^2 L}{\partial v^{i}\partial v^{j}} v^{j}\right)\delta v^{i},
\end{split}
\]
while the Lagrangian two-form is computed in coordinates as
\begin{equation*}
\begin{aligned}
\Omega _L(q,v)((\dot{q},\dot{v}), (\delta{q},\delta{v})) 
= \frac{\partial^2 L}{\partial v^{j}\partial v^{i}}  \left(\dot{q}^{i} \delta v^{j}-\delta q^{i} \dot{v}^{j}\right)
+ 
\left( \frac{\partial^2 L}{\partial q^{i}\partial v^{j}} -  \frac{\partial^2 L}{\partial q^{j} \partial v^{i}}\right)\dot{q}^j\delta{q}^i.
\end{aligned}
\end{equation*}

From the Lagrange--Dirac condition \eqref{LDDS}, the variational conditions are:
\[
\frac{\partial^2 L}{\partial v^{i} \partial v^{j}} (v^{j} \!- \!\dot{q}^{j} )\delta v^{i}
+\bigg\{ \frac{\partial^2 L}{\partial q^{i} \partial v^j} \, v^j -  \frac{\partial L}{\partial q^{i}} +  \frac{\partial^2 L}{\partial v^{j} \partial v^{i}}  \dot{v}^{j} - \left( \frac{\partial^2 L}{\partial q^i\partial v^{j}} \!-\!  \frac{\partial^2 L}{\partial q^j \partial v^{i}}\right)\dot{q}^j\bigg\}\delta{q}^i=0
\]
for all $\delta{v}$ and for all  variations $\delta{q}^{i}$ satisfying
\[
\omega_{i}^{r}(q)\delta{q}^{i}=0,
\]
together with the nonholonomic constraints 
\[
\omega_{i}^{r}(q)\dot{q}^{i}=0.
\]

Thus, by using undetermined Lagrange multipliers $\mu_{r},\; r=1,\dots,m<n$, we get the \textit{local Lagrange--d'Alembert--Dirac equations}:
\begin{equation}\label{LocalImpDynLagDirSys}
\begin{cases}
\displaystyle \frac{d}{dt} \frac{\partial L}{\partial v^i} -  \frac{\partial L}{\partial q^{i}}
- \frac{\partial^2 L}{\partial q^{i}\partial v^j } \left(\dot{q}^{j}- v^{j}\right) =\mu_{r}\omega^{r}_{i},\\[5mm]
\displaystyle \frac{\partial^2 L}{\partial v^{i} \partial v^{j}} \big(\dot{q}^{j}- v^{j} \big)=0,\\[5mm]
\displaystyle \omega_{i}^{r}(q)\dot{q}^{i}=0,
\end{cases}
\end{equation}
where we use 
\[
\frac{d}{dt} \frac{\partial L}{\partial v^i} = \frac{\partial^2 L}{\partial v^{j} \partial v^{i}}  \dot{v}^{j} 
+ \frac{\partial^2 L}{\partial q^j \partial v^{i}}  \dot{q}^j.
\]

In the \textit{hyperregular cases}, the matrix $\left(\frac{\partial^2 L}{\partial v^i\partial v^j}\right)$ has full rank and therefore we obtain the second-order condition
$
\dot q^j=v^j.
$
 We finally recover the first-order system of the Lagrange--d'Alembert equations:
\[
\frac{d}{dt}\frac{\partial L}{\partial v^i}-\frac{\partial L}{\partial q^i}=\mu_{r}\omega^{r}_{i},
\qquad \dot q^i=v^i, \qquad \omega_{i}^{r}(q)\dot{q}^{i}=0.
\]
\smallskip

The Lagrange--Dirac dynamical system $(E_L, D_L)$ on $TQ$ in \eqref{LDDS} yields the coorinate expressions of the implicit dynamics given in \eqref{LocalImpDynLagDirSys}, in which 
the first equations denote the equations of motion, the second equations are the constraints due to degeneracy of the Lagrangian, corresponding to Dirac's primary constraints, and the last equations are nonholonomic kinematic constraints.
\medskip

For the case in which $L$ is hyperregular, from the implicit dynamics in \eqref{LocalImpDynLagDirSys}, the Lagrange--d'Alembert equations in \eqref{LDAEq} or \eqref{LDAEOp_eqn} are recovered. This result confirms that the Lagrange--Dirac dynamical system $(E_L, D_L)$ is a coordinate-independent, geometric formulation that is pointwise equivalent to the classical Lagrange--d'Alembert equations. Specifically, the second-order condition $\dot{q}=v$ arises naturally from the Dirac inclusion under the assumption of hyperregularity. Of course, for the unconstrained case, namely, $\Delta_Q=TQ$, we can recover the Euler--Lagrange equation \eqref{EulerLagEqn}.

\paragraph{Energy conservation.} 
Here we check the energy conservation for the Lagrange--Dirac dynamical system in the general case where the Lagrangian $L$ is degenerate and 
there exist nonholonomic constraints.

\begin{proposition}[Energy conservation of  Lagrange--Dirac dynamical systems]\label{EneCon_UnreducedLDDS}
Along a solution curve $(q(t), v(t))$, $t\in [0,T]$ of the Lagrange--Dirac system, the energy $E_{L}(q,v)=\mathbf{D}_{2}L(q,v) \cdot v-L(q,v)$ is conserved.
\end{proposition}
\begin{proof}
\begin{itemize}
\item[(i)]
We can easily prove the energy conservation by using the maximally isotropic property of $D_{L}$: 
In fact, since we have
\[
(X(q(t),v(t)),\mathbf{d}E_{L}(g(t),v(t))) \in D_L(q(t),v(t)),
\]
and $D_L$ is isotropic, it follows that
\[
\mathbf{d}E_{L}(q(t),v(t))\cdot X(q(t),v(t))=0,\quad \textrm{for all}\;\;t \in [0,T].
\]
Hence, we get
\[
\frac{d}{dt}E_{L}(q(t),v(t))=0.
\]

\item[(ii)] Here, we also check to directly compute the time derivative of the energy $E_{L}$ along the solution curve $(q(t),v(t))$ as follows. 

By taking the time derivative of $E_L$, it follows that
\begin{equation}\label{dELdt_1stStep}
\begin{split}
\frac{d}{dt}E_L &= \dot{\mathbf{D}}_2 L \cdot v + \mathbf{D}_2 L \cdot \dot{v} - (\mathbf{D}_1 L \cdot \dot{q} + \mathbf{D}_2 L \cdot \dot{v}) \\
&= \dot{\mathbf{D}}_2 L \cdot v - \mathbf{D}_1 L \cdot \dot{q},
\end{split}
\end{equation}
where we set $\dot{\mathbf{D}}_2 L:=\frac{d}{dt}\left(\mathbf{D}_2 L\right)$.
Here, from the first equation of \eqref{LagrangeDiracCond}, there exists some $\beta \in \Delta_Q^\circ(q)$:
\[
\beta=\dot{\mathbf{D}}_2 L - \mathbf{D}_1 L - \mathbf{D}_1 \mathbf{D}_2 L \cdot (\dot{q} - v) 
\]
and hence
\begin{equation}\label{EqMo_D2L}
\dot{\mathbf{D}}_2 L= \mathbf{D}_1 L + \mathbf{D}_1 \mathbf{D}_2 L \cdot (\dot{q} - v) +\beta.
\end{equation}
Therefore, substituting $\dot{\mathbf{D}}_2 L$ into \eqref{dELdt_1stStep} yields
\begin{equation*}
\frac{d}{dt}E_L = \left( \mathbf{D}_1 L + \mathbf{D}_1 \mathbf{D}_2 L \cdot (\dot{q} - v) + \beta \right) \cdot v - \mathbf{D}_1 L \cdot \dot{q}.
\end{equation*}

Rearranging the terms to group them by the velocity mismatch $(\dot{q} - v)$:
\begin{equation}\label{EnergyRate_D1L+beta}
\begin{split}
\frac{d}{dt}E_L 
&=\left( \mathbf{D}_1 L + \mathbf{D}_1 \mathbf{D}_2 L \cdot (\dot{q} - v) + \beta \right) \cdot \big((v - \dot{q}) +\dot{q} \big) - \mathbf{D}_1 L \cdot \dot{q}\\
&=\mathbf{D}_1 L \cdot (v - \dot{q}) + \left(\mathbf{D}_1 \mathbf{D}_2 L \cdot (\dot{q} - v) \right) \cdot (v - \dot{q}) + \beta \cdot (v - \dot{q}) \\
&\qquad + \mathbf{D}_1 L \cdot \dot{q} + \left(\mathbf{D}_1 \mathbf{D}_2 L \cdot (\dot{q} - v)\right) \cdot \dot{q} + \beta \cdot \dot{q} - \mathbf{D}_1 L \cdot \dot{q}\\
&=\big( \mathbf{D}_1 L + \beta \big) \cdot (v-\dot{q}) +  \big(\mathbf{D}_1 \mathbf{D}_2 L \cdot (\dot{q} - v) \big) \cdot (v - \dot{q})+ \left(\mathbf{D}_1 \mathbf{D}_2 L \cdot (\dot{q} - v)\right) \cdot \dot{q}.
\end{split}
\end{equation}
In the above, note that $\beta \in \Delta_Q^\circ$ and $\dot{q} \in \Delta_Q$, which implies $\beta \cdot \dot{q} = 0$.
%
From the chain rule, 
\[
\dot{\mathbf{D}}_2 L = \mathbf{D}_1 \mathbf{D}_2 L \cdot \dot{q} + \mathbf{D}_2 \mathbf{D}_2 L \cdot \dot{v}. 
\]
As shown in \eqref{EqMo_D2L}, there is another expression for $\dot{\mathbf{D}}_2 L$  and hence equating them yields

\begin{equation*}
 \mathbf{D}_1 L + \mathbf{D}_1 \mathbf{D}_2 L \cdot (\dot{q} - v) +\beta= \mathbf{D}_1 \mathbf{D}_2 L \cdot \dot{q} + \mathbf{D}_2 \mathbf{D}_2 L \cdot \dot{v}
 \end{equation*}
 and hence
\[
\mathbf{D}_1 L  + \beta=\mathbf{D}_2 \mathbf{D}_2 L \cdot \dot{v} + \mathbf{D}_1 \mathbf{D}_2 L \cdot \dot{q} - \mathbf{D}_1 \mathbf{D}_2 L \cdot (\dot{q} - v).
\]

By substituting this into \eqref{EnergyRate_D1L+beta}, we get
\begin{equation*}
\begin{split}
\frac{d}{dt}E_L 
&=\big( \mathbf{D}_2 \mathbf{D}_2 L \cdot \dot{v} + \mathbf{D}_1 \mathbf{D}_2 L \cdot \dot{q} - \mathbf{D}_1 \mathbf{D}_2 L \cdot (\dot{q} - v) \big) \cdot (v-\dot{q}) \\
&\hspace{3cm}+  \big(\mathbf{D}_1 \mathbf{D}_2 L \cdot (\dot{q} - v) \big) \cdot (v - \dot{q})+ \left(\mathbf{D}_1 \mathbf{D}_2 L \cdot (\dot{q} - v)\right) \cdot \dot{q}  \\[2mm]
&=\big( \mathbf{D}_2 \mathbf{D}_2 L \cdot \dot{v}   \big) \cdot (v-\dot{q}) \\[2mm]
&=-\big( \mathbf{D}_2 \mathbf{D}_2 L \cdot (\dot{q} - v)  \big) \cdot \dot{v},
\end{split}
\end{equation*}
where the terms involving $\mathbf{D}_1\mathbf{D}_2L$ cancel each other out.

\quad By the second equation of the Lagrange--Dirac dynamical system, i.e., $\mathbf{D}_2 \mathbf{D}_2 L \cdot (\dot{q} - v) = 0$, we conclude:
\begin{equation*}
\frac{d}{dt}E_L = 0.
\end{equation*}

\end{itemize}

\end{proof}


\subsection{Constrained Lagrange--Dirac dynamics on admissible subbundles}
The Lagrange--Dirac dynamical system provides the Lagrange--d'Alembert--Dirac equations on $TQ$, where the dynamics cannot be uniquely determined by the Lagrangian vector field, since the constraint force is included by using undetermined Lagrange multipliers. Here we decompose the motion on $TQ$ into the tangential and complementary directions under the normality condition of the constraint distribution, and extract the constrained dynamics through the constrained Lagrange--Dirac structure on the admissible subbundle. 

In this subsection, we restrict ourselves to the \textit{hyperregular case} so that the Legendre transformation is a diffeomorphism. Under the additional assumption that the restricted two-form is fiberwise nondegenerate, the constrained dynamics on the admissible subbundle is uniquely characterized.

The extension to degenerate Lagrangians involves additional consistency conditions and will not be pursued here.

\paragraph{Setting for constrained dynamics.}
Let $L:TQ\to\mathbb{R}$ be a hyperregular Lagrangian and let $\Delta_Q$ be a given constraint distribution on $Q$. 
Denote by $\mathbb{F}L:TQ\to T^*Q$ the Legendre transformation, $E_L$ the energy, and $D_L=(\mathbb{F}L)^*D_{\Delta_Q}$ the Lagrange--Dirac structure over $TQ$. Define the constraint distribution $\Delta_{TQ}=(T\tau_{Q})^{-1}(\Delta_{Q})\subset TTQ$ and define also its annihilator $\Delta_{TQ}^\circ\subset T^*TQ$ in a usual way.
\medskip

Let $v_{q}=(q(t),v(t)), \, t \in [0,T]$ be a solution curve on $TQ$ of the Lagrange--Dirac dynamical system $(E_L, D_L)$, which satisfies
\[
\big(X(v_{q}),\mathbf{d}E_L(v_{q}) \big)\in D_L(v_{q}),
\]
where $X(q(t),v(t))=(\dot{q}(t),\dot{v}(t))$ denotes the tangent vector at each point $(q(t),v(t))$. 

Then, it follows from Proposition \ref{Prop: IntLDA_Eqn} that each $(q(t),v(t))$ satisfies the intrinsic Lagrange--d'Alembert--Dirac equations:
\[
\mathbf{i}_{X(v_{q})}\Omega_L(v_{q})  - \mathbf{d}E_L(v_{q}) \in \Delta_{TQ}^{\circ}(v_{q}),
\]
where $X(v_{q}) \in \Delta_{TQ}(v_{q})$.

\paragraph{Decomposition of the tangent bundle.}

Assume that the constraint satisfies the normality condition given in Definition \ref{def:SympOrthSub_Normality}, i.e., for $v_{q} \in \Delta_{Q}$,
\[
T_{v_q}\Delta_Q\cap\mathcal V_{v_q}=\{0\},
\]
where the $\Omega_L$--orthogonal complement of $T\Delta_Q$ is given by 
\begin{equation}\label{Omega_L_OrthogonalCompSpace}
\mathcal V_{v_q}
=\{\,w\in T_{v_q}TQ \mid \Omega_L(v_q)(w,Y)=0,\ \forall Y\in T_{v_q}\Delta_Q\}.
\end{equation}
Then, the tangent space $T_{v_q}TQ$ at each $v_{q} \in \Delta_{Q}$ can be split into the tangential and the $\Omega_L$--orthogonal subspaces as
\[
T_{v_q}TQ = T_{v_q}\Delta_Q \oplus \mathcal V_{v_q}.
\]
Thus, any vector in $T_{v_q}TQ$, in particular, the tangent vector $X(v_q)\in T_{v_q}TQ$, admits the decomposition
\[
X(v_q)=X_{\Delta_Q}(v_q)+X_{\mathcal V}(v_q).
\]

\paragraph{Constrained dynamics on admissible subbundles.}

Restricting the admissible directions to vectors tangent to the constraint manifold, define the subbundle:
\[
\mathcal K := T\Delta_Q \cap \Delta_{TQ}.
\]

Projecting the tangent vector $X(v_q)\in\Delta_{TQ}(v_q)$ at each $v_{q} \in \Delta_{Q}$ onto the decomposition
$T_{v_q}TQ=T_{v_q}\Delta_Q\oplus\mathcal V_{v_q}$ yields
\[
X_{\Delta_Q}(v_q)\in T_{v_q}\Delta_Q\cap\Delta_{TQ}(v_q)=\mathcal K_{v_q}, \qquad X_{\mathcal V}(v_q)\in \mathcal{V}_{v_q}.
\]
Hence, it follows that
\[
X_{\Delta_Q}\in\Gamma(\mathcal K),\qquad
X_{\mathcal V}\in\Gamma(\mathcal V|_{\Delta_Q}),
\]
where $X_{\Delta_Q}$ represents the \textit{constrained dynamics} and
$X_{\mathcal V}$ the vector field associated with the \textit{constraint force}.
\medskip

Define the restricted two-form $\Omega_{L_{\mathcal K}} \in\Gamma(\Lambda^2(\mathcal K))$
and the restricted differential $\mathbf dE_{L_{\mathcal K}} \in\Gamma(\Lambda^{1}(\mathcal K))$
fiberwise by, for each $v_q\in\Delta_Q$,
\[
\Omega_{L_{\mathcal K}}(v_q)
:=
\big(
\iota_{\Delta_Q}^\ast\Omega_L(v_q)
\big)
\big|_{\mathcal K_{v_q}\times\mathcal K_{v_q}},
\qquad
\mathbf dE_{L_{\mathcal K}}(v_q)
:=
\big(
\mathbf d(E_L\circ\iota_{\Delta_Q})(v_q)
\big)
\big|_{\mathcal K_{v_q}}.
\]

\begin{proposition}[Constrained dynamics]\label{prop:LDD-K}
Under the assumption that the restricted two-form $\Omega_{L_{\mathcal K}}$ is fiberwise nondegenerate on $\mathcal K$, the tangential component $X_{\Delta_Q}(v_q) \in \mathcal K_{v_q}$ of the tangent vector $X(v_q)$ satisfying the Lagrange--Dirac condition
\[
(X(v_q), \mathbf{d}E_L(v_q)) \in D_L(v_q), \quad \textrm{at each}\;\;v_q \in \Delta_Q,
\]
is uniquely determined by
\begin{equation}\label{ConstLagDiracSys}
\mathbf{i}_{X_{\Delta_Q}(v_q)}\Omega_{L_{\mathcal K}}(v_q)
=\mathbf{d}E_{L_{\mathcal K}}(v_q).
\end{equation}
\end{proposition}

\begin{proof}
From the Lagrange--Dirac condition,
\[
(X(v_q), \mathbf{d}E_L(v_q)) \in D_L(v_q),
\]
for each $v_q \in \Delta_Q$, the tangent vector $X(v_q)$ admits the decomposition $X(v_q) = X_{\Delta_Q}(v_q) + X_{\mathcal V}(v_q)$.  
Since $\mathcal V$ is $\Omega_L$--orthogonal to $T\Delta_Q$, i.e.,
\[
\Omega_L(v_q)(X_\mathcal{V}(v_q), Y) = 0, \qquad \forall\,Y\in T_{v_q}\Delta_Q,
\]
restricting to vectors tangent to $\Delta_Q$ removes the contribution of $X_\mathcal{V}(v_q)$,
leaving only the admissible dynamics along $\Delta_Q$. 
\medskip

On the Lagrange--Dirac dynamics
\[
\mathbf{i}_{X_{\Delta_Q}(v_q)}\Omega_L(v_q)
-\mathbf{d}E_{L}(v_{q}) \in \Delta_{TQ}^{\circ}(v_{q}),
\]
acting $Y\in\mathcal K_{v_q}\subset\Delta_{TQ}(v_q)$ yields
\[
\mathbf{d}E_L(v_q) \cdot  Y  = \Omega_L(v_q)(X_{\Delta_Q}(v_q), Y),
\]
where every element of $\Delta_{TQ}^\circ(v_q)$ vanishes on $Y \in \mathcal K_{v_q}$.
\medskip

Recall that the restricted two-form $\Omega_{L_{\mathcal K}}$ and the energy differential $\mathbf dE_{L_{\mathcal K}}$ are defined as
\[
\Omega_{L_{\mathcal K}}(v_q)
:=
\big(
\iota_{\Delta_Q}^\ast\Omega_L(v_q)
\big)
\big|_{\mathcal K_{v_q}\times\mathcal K_{v_q}},
\qquad
\mathbf dE_{L_{\mathcal K}}(v_q)
:=
\big(
\mathbf d(E_L\circ\iota_{\Delta_Q})(v_q)
\big)
\big|_{\mathcal K_{v_q}},
\]
and since we assume that \(\Omega_{L_{\mathcal K}}\) is fiberwise nondegenerate, we can uniquely determine the section $X_{\Delta_Q}\in\Gamma(\mathcal K)$ representing the constrained dynamics on $\Delta_{Q}$ by
\[
\mathbf{i}_{X_{\Delta_Q}(v_q)} \Omega_{L_{\mathcal K}}(v_q) = \mathbf{d}E_{L_{\mathcal K}}(v_q), \qquad v_q \in \Delta_Q.
\]
Thus, the constrained motion is characterized intrinsically on
$\mathcal K$, while the complementary bundle $\mathcal V$ 
corresponds to the directions associated with the constraint force. 
This completes the proof.
\end{proof}

\begin{remark}
Proposition~\ref{prop:LDD-K} shows that the constrained dynamics $X_{\Delta_Q}$ is completely determined 
on the admissible subbundle 
\[
\mathcal K = T\Delta_Q \cap \Delta_{TQ} \subset T\Delta_Q.
\]
In contrast, an extension to a vector field $X$ on $TQ$ is not unique: one may add an arbitrary component 
$X_{\mathcal V} \in \Gamma(\mathcal V|_{\Delta_Q})$ without violating the Lagrange--Dirac condition.

Thus, the indeterminacy of the dynamics corresponds precisely to the freedom of adding 
constraint forces, represented by the $\Omega_L$--orthogonal complement $\mathcal V$.
The physically relevant motion is uniquely characterized by the tangential component 
$X_{\Delta_Q} \in \Gamma(\mathcal K)$, which is determined by
\[
\mathbf{i}_{X_{\Delta_Q}(v_q)}\Omega_{L_{\mathcal K}}(v_q)=\mathbf dE_{L_{\mathcal K}}(v_q), \qquad v_q \in \Delta_Q.
\]

The uniqueness of $X_{\Delta_Q}(v_q)$ is equivalent to the injectivity of the map
\[
(\Omega_{L_{\mathcal K}})^\flat_{v_q}:\mathcal K_{v_q}\to\mathcal K_{v_q}^*, 
\qquad Z \mapsto \mathbf{i}_Z \Omega_{L_{\mathcal K}}(v_q),
\]
that is, to the nondegeneracy of $\Omega_{L_{\mathcal K}}(v_q)$. If this form has a nontrivial kernel, then $X_{\Delta_Q}$ is determined only up to addition of vectors in $\ker(\Omega_{L_{\mathcal K}}(v_q))$, and hence is not unique.
\end{remark}

Furthermore, the constrained dynamics~\eqref{ConstLagDiracSys} on $\mathcal K$
can be equivalently formulated as a Lagrange--Dirac dynamical system induced by the restriction
of $D_L$ to $\mathcal K$, as shown in the following theorem.

\begin{framed}
\begin{theorem}[Constrained Lagrange--Dirac dynamical systems on $\mathcal{K}$]\label{thm:ConstLagDirDySys}
\label{prop:induced_K_Dirac}
Let $D_L=(\mathbb{F}L)^*D_{\Delta_Q}\subset TTQ\oplus T^*TQ$ be the Lagrange--Dirac structure on $TQ$, and let
$\iota_{\Delta_Q}:\Delta_Q\hookrightarrow TQ$ be the natural inclusion. Define the backward image Dirac structure
\[
D_{L_{\Delta_Q}}
:=\mathcal{B}(T\iota_{\Delta_Q})(D_L)
\subset T\Delta_Q\oplus T^*\Delta_Q,
\]
where we assume that $D_{L_{\Delta_Q}}$ has constant rank (i.e., is a smooth Dirac structure).

Let $\mathcal K:=T\Delta_Q\cap\Delta_{TQ}\subset T\Delta_Q$ be the admissible subbundle and assume that $\mathcal K$ has constant rank. Now, we define the induced subbundle
\[
D_{L_{\mathcal K}}
\subset \mathcal K\oplus \mathcal K^*,
\]
whose fiber is given by, for each $v_q\in\Delta_Q$,
\begin{equation*}
\begin{split}
D_{L_{\mathcal K}}(v_q)
&:=
\left\{
(w_{v_q},\alpha_{v_q})\in \mathcal K_{v_q}\times \mathcal K_{v_q}^*
\ \middle|\
\textrm{There exists some}\;\widetilde{\alpha}_{v_q}\in T_{v_q}^*\Delta_Q
\right.\\
&\left.\hspace{3cm}
\textrm{ such that}\;\; (w_{v_q},\widetilde{\alpha}_{v_q})\in D_{L_{\Delta_Q}}(v_q),\quad \alpha_{v_q}
=\widetilde{\alpha}_{v_q}|_{\mathcal K_{v_q}}
\right\}.
\end{split}
\end{equation*}
Assume furthermore that the restricted two-form
\[
\Omega_{L_{\mathcal K}}
=\iota_{\Delta_Q}^*\Omega_L|_{\mathcal K\times\mathcal K}
\]
is fiberwise nondegenerate.
\medskip

Then the following statements hold:

\begin{enumerate}

\item[(i)]
For each $v_q\in\Delta_Q$,
\[
D_{L_{\mathcal K}}(v_q)
=
\left\{
(w_{v_q},\alpha_{v_q})\in\mathcal K_{v_q}\times\mathcal K_{v_q}^*
\mid
\alpha_{v_q}
=\mathbf{i}_{w_{v_q}}\Omega_{L_{\mathcal K}}(v_q)
\right\}.
\]
Hence, we have
\begin{equation}\label{Graph_DiracLK}
D_{L_{\mathcal K}}=\operatorname{graph}((\Omega_{L_{\mathcal K}})^\flat) \subset \mathcal K\oplus \mathcal K^*.
\end{equation}

\item[(ii)]
$D_{L_{\mathcal K}}$ is a smooth Dirac structure on the vector bundle 
$
\mathcal K\to \Delta_Q.
$

\item[(iii)]
The \textbf{constrained Lagrange--Dirac dynamical system} $(X_{\Delta_{Q}}, E_{L_{\mathcal K}}, D_{L_{\mathcal K}})$ satisfying the condition
\[
\big(X_{\Delta_Q}(v_q),\mathbf dE_{L_{\mathcal K}}(v_q)\big)
\in D_{L_{\mathcal K}}(v_q),
\qquad \forall v_q\in\Delta_Q.
\]
is equivalent to the \textbf{constrained Lagrange--d'Alembert--Dirac equation}
\[
\mathbf{i}_{X_{\Delta_Q}(v_q)}\Omega_{L_{\mathcal K}}(v_q)
=
\mathbf dE_{L_{\mathcal K}}(v_q).
\]
\end{enumerate}
\end{theorem}
\end{framed}
\begin{proof}
We proceed in several steps.

\medskip
\noindent

\begin{itemize}
\item[(i)]
By definition of the backward Dirac map under the inclusion
$
\iota_{\Delta_Q}:\Delta_Q\hookrightarrow TQ,
$
we have
\[
(w_{v_q},\widetilde{\alpha}_{v_q})
\in D_{L_{\Delta_Q}}(v_q)
\subset T_{v_q}\Delta_Q\oplus T^*_{v_q}\Delta_Q,
\]
if and only if there exists a covector
\[
\beta_{v_q}\in T^*_{v_q}TQ
\]
such that
\[
\widetilde{\alpha}_{v_q}
=(T\iota_{\Delta_Q})^*(\beta_{v_q}),
\qquad
\big((T\iota_{\Delta_Q})(w_{v_q}),\beta_{v_q}\big)\in D_L(v_q)\subset T_{v_{q}}TQ \times T_{v_{q}}^{\ast}TQ.
\]

Since $\iota_{\Delta_Q}$ is the natural inclusion, we identify $(T\iota_{\Delta_Q})(w_{v_q})$ with $w_{v_q}$ itself by regarding $T\Delta_Q$ as a subbundle of $TTQ$.
Then, we have
\[
(w_{v_q},\beta_{v_q})\in D_L(v_q),
\]
where $D_L$ is given fiberwise by
\[
D_L(v_q)
=
\left\{
(w_{v_q},\beta_{v_q})\in T_{v_q}TQ\times T^*_{v_q}TQ
\mid
w_{v_q}\in\Delta_{TQ}(v_q),\;
\beta_{v_q}-\mathbf{i}_{w_{v_q}}\Omega_L(v_q)\in\Delta_{TQ}^\circ(v_q)
\right\}.
\]
Hence, it follows that
\[
\beta_{v_q}-\mathbf{i}_{w_{v_q}}\Omega_L(v_q)
\in
\Delta_{TQ}^\circ(v_q).
\]
\quad Now restrict the vector $w_{v_q}\in T_{v_q}\Delta_Q$ to admissible subbundles
\[
w_{v_q}\in\mathcal K_{v_q}
=
T_{v_q}\Delta_Q\cap\Delta_{TQ}(v_q).
\]
By definition of $D_{L_{\mathcal K}}$, we set
\[
\alpha_{v_q}
=
\widetilde{\alpha}_{v_q}|_{\mathcal K_{v_q}}.
\]

Let $u_{v_q}\in\mathcal K_{v_q}\subset\Delta_{TQ}(v_q)$. Then every element of
$\Delta_{TQ}^\circ(v_q)$ vanishes on $u_{v_q}$, hence
\[
\langle\beta_{v_q},u_{v_q}\rangle
=
\Omega_L(v_q)(w_{v_q},u_{v_q}).
\]

Since $\alpha_{v_q}$ is the restriction of $\beta_{v_q}$ to $\mathcal K_{v_q}$, we obtain
\[
\langle\alpha_{v_q},u_{v_q}\rangle
=
\Omega_L(v_q)(w_{v_q},u_{v_q}),
\qquad
\forall u_{v_q}\in\mathcal K_{v_q}.
\]
By definition of the restricted two-form,
\[
\Omega_{L_{\mathcal K}}(v_q)
=
\Omega_L(v_q)|_{\mathcal K_{v_q}\times\mathcal K_{v_q}},
\]
therefore
\[
\langle\alpha_{v_q},u_{v_q}\rangle
=
\Omega_{L_{\mathcal K}}(v_q)(w_{v_q},u_{v_q}),
\qquad
\forall u_{v_q}\in\mathcal K_{v_q}.
\]
Hence
\[
\alpha_{v_q}
=
\mathbf{i}_{w_{v_q}}\Omega_{L_{\mathcal K}}(v_q).
\]

\item[(ii)]
We next show that $D_{L_{\mathcal K}} \subset \mathcal K\oplus\mathcal K^*$ defines a smooth Dirac structure on the vector bundle
\[
\mathcal K\to \Delta_Q.
\]

By part (i), for each $v_q\in\Delta_Q$, we have already proved
\[
D_{L_{\mathcal K}}(v_q)
=
\left\{
(w_{v_q},\alpha_{v_q})\in \mathcal K_{v_q}\times \mathcal K_{v_q}^*
\mid
\alpha_{v_q}
=
\mathbf{i}_{w_{v_q}}\Omega_{L_{\mathcal K}}(v_q)
\right\}.
\]
Hence fiberwise,
\[
D_{L_{\mathcal K}}(v_q)
=\operatorname{graph}\!\big((\Omega_{L_{\mathcal K}})^\flat_{v_q}\big),
\]
where
\[
(\Omega_{L_{\mathcal K}})^\flat_{v_q}:
\mathcal K_{v_q}\to\mathcal K_{v_q}^*,
\qquad
w_{v_q}\mapsto \mathbf{i}_{w_{v_q}}\Omega_{L_{\mathcal K}}(v_q).
\]

Since $\Omega_{L_{\mathcal K}}$ depends smoothly on $v_q$, the bundle map
$
(\Omega_{L_{\mathcal K}})^\flat:\mathcal K\to\mathcal K^*
$
is smooth. Therefore its graph
\[
D_{L_{\mathcal K}}=\operatorname{graph}\!\big((\Omega_{L_{\mathcal K}})^\flat\big)
\subset
\mathcal K\oplus\mathcal K^*
\]
is a smooth vector subbundle.

\quad We now verify isotropy with respect to the canonical symmetric pairing on
$\mathcal K\oplus\mathcal K^*$: for each $v_{q} \in \Delta_{Q}$,
\[
\big\langle\!\big\langle (w_{v_{q}},\alpha_{v_{q}}),(u_{v_{q}},\beta_{v_{q}})\big\rangle\!\big\rangle
:=
\langle\alpha_{v_{q}},u_{v_{q}}\rangle+\langle\beta_{v_{q}},w_{v_{q}}\rangle.
\]
Take two arbitrary elements
$(w_{v_{q}},\mathbf{i}_{w_{v_{q}}}\Omega_{L_{\mathcal K}}(v_q)),
(u_{v_{q}},\mathbf{i}_{u_{v_{q}}}\Omega_{L_{\mathcal K}}(v_q))
\in D_{L_{\mathcal K}}(v_q)$.
Then
\[
\begin{aligned}
\big\langle\!\big\langle
(w_{v_{q}},\mathbf{i}_{w_{v_{q}}}\Omega_{L_{\mathcal K}}(v_q)),
(u_{v_{q}},\mathbf{i}_{u_{v_{q}}}\Omega_{L_{\mathcal K}}(v_q))
\big\rangle\!\big\rangle
&=
\Omega_{L_{\mathcal K}}(v_q)(w_{v_{q}},u_{v_{q}})
+
\Omega_{L_{\mathcal K}}(v_q)(u_{v_{q}},w_{v_{q}}) =0,
\end{aligned}
\]
since $\Omega_{L_{\mathcal K}}(v_q)$ is skew-symmetric. Thus each fiber
$D_{L_{\mathcal K}}(v_q)$ is isotropic.

\quad Finally,
\[
\dim D_{L_{\mathcal K}}(v_q)=\dim \mathcal K_{v_q},
\]
because it is the graph of a linear map $(\Omega_{L_{\mathcal K}})_{v_q}^\flat: \mathcal K_{v_q} \to
\mathcal K_{v_q}^*$. On the other hand,
\[
\dim\big(\mathcal K_{v_q}\times\mathcal K_{v_q}^*\big)
=
2\dim\mathcal K_{v_q}.
\]
Hence $D_{L_{\mathcal K}}(v_q)$ is an isotropic subspace of maximal possible
dimension, i.e.\ it is maximally isotropic. Therefore
\[
D_{L_{\mathcal K}}
\subset
\mathcal K\oplus\mathcal K^*
\]
is a smooth maximally isotropic subbundle, namely, a Dirac structure on
$\mathcal K$.
\item[(iii)]

By (ii), the condition
\[
\big(X_{\Delta_Q}(v_q),\mathbf dE_{L_{\mathcal K}}(v_q)\big)
\in D_{L_{\mathcal K}}(v_q)
\]
holds if and only if
\[
\mathbf{i}_{X_{\Delta_Q}(v_q)} \Omega_{L_{\mathcal K}}(v_q)
=
\mathbf dE_{L_{\mathcal K}}(v_q),
\]
which is precisely the constrained Lagrange--d'Alembert--Dirac equation.
\end{itemize}
\end{proof}

\subsection{Gauge covariance of Lagrange--Dirac dynamical systems}
In this subsection, under the assumption of hyperregularity, we investigate a gauge covariance property of Lagrange--Dirac dynamical systems under the flow generated by the constrained dynamics. In particular, we show that the underlying Dirac structure is preserved up to a time-dependent gauge transformation, which encodes a form of invariance corresponding to the apparent failure of symplecticity in the Lagrange--d'Alembert formulation.

\paragraph{Symplectic preservation in unconstrained cases.} 
Let $L$ be a hyperregular Lagrangian on $TQ$ and consider a distribution $\Delta_{Q}$ on $Q$.  For the special case of unconstrained system, i.e., $\Delta_{Q}=TQ$, as was already illustrated, 
there exists a Lagrangian vector field $X_{L}$ on $TQ$ uniquely determined such that the Lagrange--Dirac condition holds:
\[
(X_{L}, \mathbf{d}E_{L}) \in D_{L},
 \]
from which the Euler--Lagrange equations are obtained as
\[
\mathbf{i}_{X_{L}}\Omega_{L}=\mathbf{d}E_{L}.
\]

Then, we obtain $\pounds_{X_{L}}\Omega_{L}=0$, since
\[
\pounds_{X_{L}}\Omega_{L}
= \mathbf{i}_{X_{L}}\mathbf{d}\Omega_{L}
+ \mathbf{d}(\mathbf{i}_{X_{L}}\Omega_{L})
= \mathbf{d}(\mathbf{i}_{X_{L}}\Omega_{L})
= \mathbf{d}\mathbf{d}E_{L}
= 0,
\]
where $\mathbf{d}\Omega_{L}=0$ since the Lagrangian two-form $\Omega_{L}$ is closed. Thus, $\Omega_{L}$ is preserved along the flow:
\[
\varphi^{\ast}_t \; \Omega_L = \Omega_L,
\]
where $\varphi_t: TQ \to TQ$ is the {\it flow} associated with $X_L$.

\paragraph{Time-dependent exact two-form deformations of Lagrangian two-forms.} 
In the presence of nonholonomic constraints $\Delta_{Q} \subset TQ$, 
the restricted two-form $\Omega_{L_{\mathcal K}}$ is typically fiberwise nondegenerate, which uniquely determines the constrained dynamics $X_{\Delta_Q} \in \Gamma(\mathcal K)$ and the associated flow $\varphi_t: \Delta_Q \to \Delta_Q$. Unlike in unconstrained Lagrangian mechanics, the Lagrangian two-form is generally not invariant under the nonholonomic flow $\varphi_t$.
The following proposition clarifies that this failure of invariance is not arbitrary, but is characterized by a time-dependent deformation through an exact two-form.

\begin{framed}
\begin{proposition}\label{Prop:structural property}
Let $(X,\mathbf dE_L)\in D_L$ be a Lagrange--Dirac dynamical system.
Assume that along $\Delta_Q$, the vector field $X$ decomposes relative to a smooth decomposition
$
TTQ|_{\Delta_Q}
=
T\Delta_Q\oplus\mathcal V
$
as
\[
X=X_{\Delta_Q}+X_{\mathcal V},
\]
where $X_{\Delta_Q}\in\Gamma(\mathcal K)$ and $X_{\mathcal V}\in\Gamma(\mathcal V)$, and where the subbundle $\mathcal{V}$ is $\Omega_L$--orthogonal to $T\Delta_Q$, i.e., $\Omega_L(X_{\mathcal{V}}, Y) = 0$ for all $Y \in \Gamma(T\Delta_Q)$.

\medskip

Let $\varphi_t:\Delta_Q\to\Delta_Q$ be the flow of $X_{\Delta_Q}$ and define
\[
\boldsymbol\beta
:=
\mathbf i_X\Omega_L-\mathbf dE_L
\in\Gamma(\Delta_{TQ}^\circ).
\]
Then, we obtain
\begin{equation}\label{GaugeTrans_LagTwoForm}
\varphi^{\ast}_t (\iota_{\Delta_{Q}}^{\ast}\Omega_L)
= (\iota_{\Delta_{Q}}^{\ast}\Omega_L)
+ \int_0^t \varphi_s^{\ast}\mathbf{d}(\iota_{\Delta_{Q}}^{\ast}\boldsymbol{\beta}) ds,
\end{equation}
where $\iota_{\Delta_Q}: \Delta_Q \hookrightarrow TQ$ is the inclusion map.
\end{proposition}
\end{framed}
\begin{proof}
By using Cartan's formula and the closedness of the Lagrangian two-form $\mathbf{d}\Omega_L=0$, the Lie derivative of $\Omega_L$ along $X$ is given by
\[
\pounds_{X}\Omega_L = \mathbf{i}_{X}\mathbf{d}\Omega_L + \mathbf{d}(\mathbf{i}_{X}\Omega_L) = \mathbf{d}(\mathbf{i}_{X}\Omega_L).
\]
Substituting the Lagrange--Dirac condition $\mathbf{i}_{X}\Omega_L = \mathbf{d}E_L + \boldsymbol{\beta}$ into the above, we have
\[
\pounds_{X}\Omega_L = \mathbf{d}(\mathbf{d}E_L + \boldsymbol{\beta}) = \mathbf{d}\boldsymbol{\beta}.
\]
Restricting this relation to the constraint distribution $\Delta_Q$ via the pullback $\iota_{\Delta_Q}^*$, we obtain
\begin{equation}\label{PullBack_Step1}
\iota_{\Delta_Q}^*(\pounds_{X}\Omega_L) = \iota_{\Delta_Q}^*\mathbf{d}\boldsymbol{\beta} = \mathbf{d}(\iota_{\Delta_Q}^*\boldsymbol{\beta}).
\end{equation}

Next, we evaluate the left-hand side of \eqref{PullBack_Step1}. Using the fact that the exterior derivative commutes with the pullback, we have
\[
\iota_{\Delta_Q}^*(\pounds_{X}\Omega_L) = \iota_{\Delta_Q}^* \mathbf{d}(\mathbf{i}_{X}\Omega_L) = \mathbf{d} \left( \iota_{\Delta_Q}^*(\mathbf{i}_{X}\Omega_L) \right).
\]
For any $Y \in \Gamma(T\Delta_Q)$, the term inside the derivative satisfies
\[
\left( \iota_{\Delta_Q}^*(\mathbf{i}_{X}\Omega_L) \right)(Y) = \Omega_L(X, {\iota_{\Delta_Q}}_* Y).
\]
Using the decomposition $X = X_{\Delta_Q} + X_{\mathcal{V}}$ and recalling that the subbundle $\mathcal{V}$ is defined as the $\Omega_L$--orthogonal complement to $T\Delta_Q$ in \eqref{Omega_L_OrthogonalCompSpace} under the normality condition, we have $\Omega_L(X_{\mathcal{V}}, {\iota_{\Delta_Q}}_* Y) = 0$. Thus, we obtain
\[
\Omega_L(X, {\iota_{\Delta_Q}}_* Y) = \Omega_L(X_{\Delta_Q}, {\iota_{\Delta_Q}}_* Y) = \left( \mathbf{i}_{X_{\Delta_Q}} (\iota_{\Delta_Q}^* \Omega_L) \right)(Y).
\]
Thus, we have the identity $\iota_{\Delta_Q}^*(\mathbf{i}_{X}\Omega_L) = \mathbf{i}_{X_{\Delta_Q}} (\iota_{\Delta_Q}^* \Omega_L)$, which leads to
\begin{equation}\label{PullBack_Step2}
\iota_{\Delta_Q}^*(\pounds_{X}\Omega_L) = \mathbf{d} \left( \mathbf{i}_{X_{\Delta_Q}} (\iota_{\Delta_Q}^* \Omega_L) \right).
\end{equation}

On the other hand, since $X_{\Delta_Q}$ is tangent to $\Delta_Q$, the Lie derivative of the restricted two-form $\iota_{\Delta_Q}^* \Omega_L$ along $X_{\Delta_Q}$ is well defined on $\Delta_Q$ and is given by
\[
\pounds_{X_{\Delta_Q}}(\iota_{\Delta_Q}^* \Omega_L) = \mathbf{i}_{X_{\Delta_Q}} \mathbf{d}(\iota_{\Delta_Q}^* \Omega_L) + \mathbf{d}\left(\mathbf{i}_{X_{\Delta_Q}} ( \iota_{\Delta_Q}^* \Omega_L) \right).
\]
Given that $\mathbf{d}(\iota_{\Delta_Q}^* \Omega_L) = \iota_{\Delta_Q}^* \mathbf{d}\Omega_L = 0$, this expression simplifies to
\begin{equation}\label{PullBack_Step3}
\pounds_{X_{\Delta_Q}}(\iota_{\Delta_Q}^* \Omega_L) = \mathbf{d}\left(\mathbf{i}_{X_{\Delta_Q}} (\iota_{\Delta_Q}^* \Omega_L) \right).
\end{equation}
Comparing \eqref{PullBack_Step1}, \eqref{PullBack_Step2}, and \eqref{PullBack_Step3}, we conclude that
\[
\pounds_{X_{\Delta_Q}}(\iota_{\Delta_Q}^* \Omega_L) = \mathbf{d}(\iota_{\Delta_Q}^* \boldsymbol{\beta}).
\]

Finally, by the definition of the Lie derivative, the time evolution of the pulled-back two-form along the flow $\varphi_t$ of $X_{\Delta_Q}$ is
\[
\frac{d}{dt} \varphi_t^* (\iota_{\Delta_Q}^* \Omega_L) = \varphi_t^* \pounds_{X_{\Delta_Q}} (\iota_{\Delta_Q}^* \Omega_L) = \varphi_t^* \mathbf{d}(\iota_{\Delta_Q}^* \boldsymbol{\beta}).
\]
Integrating this equation from $0$ to $t$ yields the desired structural property:
\[
\varphi_t^* (\iota_{\Delta_Q}^* \Omega_L) = \iota_{\Delta_Q}^* \Omega_L + \int_0^t \varphi_s^* \mathbf{d}(\iota_{\Delta_Q}^* \boldsymbol{\beta}) ds.
\]
\end{proof}

\begin{remark}[Consistency with the variational approach]
\label{rem:consistency_variational}
It is worth noting that the identity \eqref{GaugeTrans_LagTwoForm} obtained in Proposition~\ref{Prop:structural property} is formally identical to the one derived in Proposition~\ref{Prop:ConvOmegaL} via the Lagrange--d'Alembert variational principle. In the variational setting, $\boldsymbol{\beta}$ is interpreted as the one-form associated with the constraint force. Here, the same structure emerges naturally from the geometric properties of the Lagrange--Dirac system $(E_L, D_L)$. This consistency highlights that the Lagrange--Dirac framework not only recovers the equations of motion but also preserves the underlying non-symplectic geometric structure inherent in nonholonomic mechanics.
\end{remark}

\paragraph{Gauge transformations of Lagrange--Dirac structures.}
While the structural property in \eqref{GaugeTrans_LagTwoForm} highlights the lack of symplectic invariance of the Lagrangian two-form due to the time-dependent exact two-form deformation, it also suggests that the time evolution of the Lagrange--Dirac structure $D_{L_{\mathcal K}}$ under the nonholonomic flow is precisely characterized by a gauge transformation. This motivates the interpretation of nonholonomic evolution as a \textit{gauge transformation} of Dirac structures, a concept developed in \cite{SeWe2001, BurRad2003}. By viewing the deformation as a gauge change, we can recover a fundamental covariance property for the constrained system.

\begin{definition}[Gauge transformations of Dirac structures]
Let $\operatorname{Dir}(TM)$ be the space of Dirac structures on $M$. 
For a two-form $B$ on a manifold $M$,  we define a map $\tau_{B}: \operatorname{Dir}(TM) \to \operatorname{Dir}(TM)$ by
\[
\tau_{B}(D_{M})=\Big\{ \big(X, \alpha + \mathbf{i}_{X}B\big)\; \big| \; (X, \alpha) \in D_{M} \Big\},
\]
where $D_{M} \subset TM \oplus T^{\ast}M$ is a Dirac structure on $M$; see \cite{SeWe2001, BurRad2003}.

\end{definition}

From this geometric perspective, we now state the main theorem of this section, which characterizes the covariance of Lagrange--Dirac systems under the nonholonomic flow.

\begin{framed}
\begin{theorem}[Gauge covariance of Lagrange--Dirac systems]\label{thm:preservation_LDK}
Let $(X_{\Delta_Q},E_{L_{\mathcal{K}}}, D_{L_{\mathcal{K}}})$ be the constrained Lagrange--Dirac dynamical system.
Let $\varphi_t:\Delta_Q\to\Delta_Q$ be the associated flow of the constrained vector field $X_{\Delta_Q}\in\Gamma(\mathcal K)$. 

Then the constrained Lagrange--Dirac dynamical system is covariant in the sense that
\[
\big( \varphi_t^* X_{\Delta_Q},\, \varphi_t^* \mathbf{d}E_{L_{\mathcal{K}}} \big) \in \tau_{B_t}(D_{L_{\mathcal{K}}}),
\]
where
\[
\tau_{B_t}:  \operatorname{Dir}(\mathcal{K}) \to \operatorname{Dir}(\mathcal{K})
\]
is the gauge transformation of Dirac structures on the subbundle $\mathcal{K}$ associated with the exact two-form $B_t$ restricted to $\mathcal{K}$, given by
\[
B_t = \left. \left( \int_0^t \varphi_s^* \mathbf{d}(\iota_{\Delta_Q}^* \boldsymbol{\beta}) ds \right) \right|_{\mathcal K \times \mathcal K} \in \Gamma(\Lambda^2(\mathcal{K})).
\]
\end{theorem}
\end{framed}
\begin{proof}
First, we observe the transformation of the Dirac structure under the flow by utilizing the restriction to the admissible subbundle $\mathcal{K}$.

By Proposition \ref{Prop:structural property}, the pullback of the full Lagrangian two-form $\iota_{\Delta_Q}^* \Omega_L$ on $\Delta_Q$ along the flow $\varphi_t$ satisfies
\[
\varphi_t^* (\iota_{\Delta_Q}^* \Omega_L) = \iota_{\Delta_Q}^* \Omega_L + \int_0^t \varphi_s^* \mathbf{d}(\iota_{\Delta_Q}^* \boldsymbol{\beta}) ds,
\]
where 
\[
\int_0^t \varphi_s^* \mathbf{d}(\iota_{\Delta_Q}^* \boldsymbol{\beta}) ds \in \Gamma(\Lambda^2(T\Delta_Q))
\]
is the exact two-form on $\Delta_Q$. 
\medskip

Restricting both sides of this identity to the admissible subbundle $\mathcal{K} \times \mathcal{K}$, recalling that $\Omega_{L_{\mathcal{K}}} = \iota_{\Delta_Q}^* \Omega_L |_{\mathcal K \times \mathcal K}$, and setting
\[
B_t := \left. \left( \int_0^t \varphi_s^* \mathbf{d}(\iota_{\Delta_Q}^* \boldsymbol{\beta}) ds \right) \right|_{\mathcal K \times \mathcal K} \in \Gamma(\Lambda^2(\mathcal{K})), 
\]
we obtain the relation for the restricted forms:
\[
\varphi_t^* \Omega_{L_{\mathcal{K}}} = \Omega_{L_{\mathcal{K}}} + B_t.
\]
Since $D_{L_{\mathcal{K}}} = \operatorname{graph}((\Omega_{L_{\mathcal{K}}})^\flat) \subset \mathcal{K} \oplus \mathcal{K}^*$, its pullback by $\varphi_t$ (which is well defined on $\mathcal{K}$ due to the flow being generated by $X_{\Delta_Q} \in \Gamma(\mathcal{K})$) is given by the graph of the pulled-back restricted two-form:
\[
\varphi_t^* D_{L_{\mathcal{K}}} = \operatorname{graph}(\varphi_t^* \Omega_{L_{\mathcal{K}}}) = \operatorname{graph}(\Omega_{L_{\mathcal{K}}} + B_t).
\]
By the definition of the gauge transformation $\tau_{B_t}$ on $\operatorname{Dir}(\mathcal{K})$, this is precisely
\begin{equation}\label{eq:proof_gauge_rel}
\varphi_t^* D_{L_{\mathcal{K}}} = \tau_{B_t}(D_{L_{\mathcal{K}}}).
\end{equation}

Next, we consider the constrained Lagrange--Dirac condition on $\mathcal{K}$ given by \eqref{ConstLagDiracSys}:
\[
\mathbf{i}_{X_{\Delta_Q}} \Omega_{L_{\mathcal{K}}} = \mathbf{d}E_{L_{\mathcal{K}}}.
\] 
Applying the pullback $\varphi_t^*$ to both sides, and using the naturality of the interior product $\varphi_t^*(\mathbf{i}_{X} \alpha) = \mathbf{i}_{\varphi_t^* X} (\varphi_t^* \alpha)$, we obtain
\[
\mathbf{i}_{\varphi_t^* X_{\Delta_Q}} (\varphi_t^* \Omega_{L_{\mathcal{K}}}) = \varphi_t^* \mathbf{d}E_{L_{\mathcal{K}}},
\]
where $\varphi_t^* X_{\Delta_Q}$ denotes the pullback vector field.
\medskip

This equation is equivalent to saying that the pair $(\varphi_t^* X_{\Delta_Q}, \varphi_t^* \mathbf{d}E_{L_{\mathcal{K}}})$ belongs to the graph of $\varphi_t^* \Omega_{L_{\mathcal{K}}}$ over $\mathcal{K}$, which is the pullback Dirac structure:
\[
\big( \varphi_t^* X_{\Delta_Q},\, \varphi_t^* \mathbf{d}E_{L_{\mathcal{K}}} \big) \in \varphi_t^* D_{L_{\mathcal{K}}}.
\]
Substituting the relation \eqref{eq:proof_gauge_rel} into the above, we arrive at the desired covariance property:
\[
\big( \varphi_t^* X_{\Delta_Q},\, \varphi_t^* \mathbf{d}E_{L_{\mathcal{K}}} \big) \in \tau_{B_t} (D_{L_{\mathcal{K}}}).
\]
\end{proof}

\begin{remark}
The nonholonomic flow $\varphi_t: \Delta_{Q} \to \Delta_{Q}$ associated with the constrained vector field $X_{\Delta_{Q}}$ acts on the Lagrange--Dirac structure $D_{L_{\mathcal K}}\subset \mathcal K\oplus\mathcal K^*$ not by a simple pullback in general, i.e., $\varphi_t^{\ast}D_{L_{\mathcal{K}}} \neq D_{L_{\mathcal{K}}}$, but by a time-dependent exact two-form $B_{t}$ transformation. 

The identity
\[
\varphi_t^* D_{L_{\mathcal K}} = \tau_{B_t}(D_{L_{\mathcal K}}),
\]
demonstrates that the Dirac structure is preserved \textit{up to a gauge transformation} determined by the non-closedness of the constraints.

From a physical perspective, this shows that the isotropy property underlying energy conservation is preserved under the gauge transformation. Indeed, if
\[
(\varphi_t^*X_{\Delta_Q},\, \varphi_t^*\mathbf{d}E_{L_{\mathcal K}}) \in \tau_{B_t}(D_{L_{\mathcal K}}),
\]
then by the definition of the gauge transformation and the invariance of the vector field under its own flow, i.e., $\varphi_t^*X_{\Delta_Q}=X_{\Delta_Q}$, we have $\varphi_t^*\mathbf{d}E_{L_{\mathcal K}} = \mathbf{d}E_{L_{\mathcal K}} + \mathbf{i}_{X_{\Delta_Q}}B_t$ for $(X_{\Delta_Q}, \mathbf{d}E_{L_{\mathcal K}}) \in D_{L_{\mathcal{K}}}$. Thus, we obtain the energy conservation:
\[
\begin{split}
\langle \varphi_t^*\mathbf{d}E_{L_{\mathcal K}}, \varphi_t^*X_{\Delta_Q} \rangle 
&= \langle \mathbf{d}E_{L_{\mathcal K}} + \mathbf{i}_{X_{\Delta_Q}}B_t, X_{\Delta_Q} \rangle \\
&= \langle \mathbf{d}E_{L_{\mathcal K}}, X_{\Delta_Q} \rangle + B_t(X_{\Delta_Q}, X_{\Delta_Q}) \\
&= \langle \mathbf{d}E_{L_{\mathcal K}}, X_{\Delta_Q} \rangle \\
&= 0,
\end{split}
\]
where we utilized the skew-symmetry of $B_t$ and the fundamental isotropy of $D_{L_{\mathcal K}}$.

Thus, the isotropy of $D_{L_{\mathcal K}}$, which encodes the workless nature of constraint forces, is robust under the $\varphi_{t}$-induced gauge transformation $\tau_{B_{t}}$.
\end{remark}

\begin{remark}
In the general context, if a Dirac structure $D_{M}$ is integrable (i.e., closed under the Courant bracket), its gauge transform $\tau_{B}(D_M)$ remains integrable provided that $\mathbf{d}B = 0$. In our present context, since $B_t$ is defined as an integral of an exact form, it is naturally closed, i.e., $\mathbf{d}B_t = 0$. However, since nonholonomic constraints are generally nonintegrable, the underlying distribution $\Delta_Q$ and the associated Dirac structures are treated as \emph{almost Dirac structures} throughout this paper. It is important to note that the gauge transformation $\tau_{B_{t}}$ preserves the maximal isotropy property of the subbundle independently of the integrability or the closedness of $B_{t}$.
\end{remark}


\section{Reduction of Lagrange--Dirac structures}
In \cite{YoMa2007a}, the theory of Lie--Dirac reduction was introduced as a reduction procedure for the $G$-invariant Dirac structure on the cotangent bundle $T^*G$ of a Lie group $G$, which is induced by a nonholonomic constraint $\Delta_{G} \subset TG$. In this section, we develop a \textit{reduction theory for Lagrange--Dirac structures on $TG$}, namely Dirac structures depending on a Lagrangian $L$, as the Lagrangian counterpart to the Lie--Dirac reduction on $T^*G$.

\subsection{Lie--Dirac reduction}
\paragraph{Invariance and reduction of Dirac structures.}
We briefly review the general notions of invariance and reduction of Dirac structures, following \cite{YoMa2007a,YoMa2009, BS2001}.

Let $G$ be a Lie group acting freely and properly on a manifold $M$, denoted by $\Phi:G \times M \to M$. For $h\in G$ and $x\in M$, the action is given by
\[
hx \equiv h\cdot x\equiv\Phi(h,x)\equiv\Phi_h(x),\,
\]
and let $\pi^{/G}_{M}: M\to B:=M/G$ be the quotient projection,
\medskip

The action $\Phi_h$ on $M$ naturally lifts to the Pontryagin bundle $TM \oplus T^*M$ as follows:
\begin{equation}
\hat{\Phi}_{h}(v, \alpha) = \left( (\Phi_h)_{*} v, (\Phi_h^{-1})^* \alpha \right),
\end{equation}
where $(\Phi_h)_*: TM \to TM$ is the tangent lift and $(\Phi_h^{-1})^*: T^*M \to T^*M$ is the cotangent lift. 

This combined action is sometimes called the natural lift, and it is characterized by the property that it preserves the natural pairing between $TM$ and $T^*M$, i.e.,$$\langle (\Phi_h^{-1})^* \alpha, (\Phi_h)_* v \rangle = \langle \alpha, v \rangle$$for all $(v, \alpha) \in TM \oplus T^*M$. Consequently, the symmetric pairing on the Pontryagin bundle is invariant under $\hat{\Phi}_h$.
\medskip 

A Dirac structure $D_{M} \subset TM \oplus T^*M$ is then said to be $G$-invariant if it is preserved as a subbundle under this natural lift:
\begin{equation*}
\hat{\Phi}_h(D_M) = D_M \quad \text{for all } h \in G.
\end{equation*}
Equivalently, for all $h \in G$ and $(X, \alpha) \in D_{M}$,
\begin{equation}\label{invariance_Dirac}
((\Phi_h)_* X, (\Phi_h^{-1})^* \alpha) \in D_{M}.
\end{equation} 
We also denote the push-forward of a one-form $\alpha$ by $(\Phi_h)_{*} \alpha := (\Phi_h^{-1})^{*} \alpha$. Then, a Dirac structure $D_{M} \subset TM \oplus T^*M$ is said to be $G$-invariant if
\begin{equation*}
((\Phi_h)_{*} X, (\Phi_h)_{*} \alpha) \in D_{M},
\end{equation*}
for all $h \in G$ and $(X, \alpha) \in D_{M}$.


\begin{definition}[Dirac bundle reduction]
Let $D_{M}$ be a $G$-invariant Dirac structure on $M$.  
The quotient
\[
D^{/G}_{M}:=D_{M}/G \subset (TM \oplus T^{\ast}M)/G\cong (TM)/G \oplus (T^*M)/G
\]
is called the \textit{reduced Dirac structure} over the reduced bundle $TM/G \to M/G$. Figure~\ref{DiracBundleReduction} illustrates this Dirac reduction.
\end{definition}

\medskip

\begin{figure}[h]
\centering
\begin{tikzcd}
TM \oplus T^{\ast}M \arrow[d, "/G"'] 
& D_{M} \arrow[l, hook'] \arrow[d, "/G"] \\
F=(TM\oplus T^{\ast}M)/G \arrow[d]
& D_{M}/G \arrow[l, hook']\\
B = M/G
\end{tikzcd}
\caption{Dirac reduction under a free and proper $G$-action.}
\label{DiracBundleReduction}
\end{figure}

For the special case $M=T^{\ast}G$, this construction 
was developed in \cite{YoMa2007a}, which is referred to as the \textit{Lie--Dirac reduction}. 
We review this case below, since the reduction of Lagrange--Dirac structures on $TG$ will be formulated in direct analogy with it.


\paragraph{Induced Dirac structures on $T^{\ast}G$ and trivialized expressions.}
The left translation action of a Lie group $G$ on itself,
\[
L_h:G\to G,\qquad g\mapsto hg,
\]
induces natural left actions on the tangent bundle $TG$
and the cotangent bundle $T^*G$ through the tangent and cotangent lifts.
The induced actions on $TG$ and $T^*G$ are explicitly given by
\begin{equation*}
h  v_{g} \equiv h \cdot v_{g} \equiv T_{g} L_h \cdot v_{g},\qquad  h p_{g} \equiv h \cdot p_{g}\equiv 
T^\ast_{hg} L_{h^{-1}} \cdot p_{g}.
\end{equation*} 
Here $T_{g}L_{h}:T_{g}G \to T_{hg}G$ is the tangent of the left translation map $L_{h}:G \to G;\, g\mapsto hg$ at the point $g$ and $T^{\ast}_{hg}L_{h^{-1}}:T^{\ast}_{g}G \to T^{\ast}_{hg}G$ is the dual of the map $T_{hg}L_{h^{-1}}:T_{hg}G \to T_{g}G$. 
Throughout the paper, we often use the concatenation notation for the tangent and cotangent lifts of the Lie group actions on $TG$ and $T^{\ast}G$ respectively by writing $h  v_{g}$ and $h  p_{g}$ instead of $T_{g} L_h \cdot v_{g}$ and $T^\ast_{hg} L_{h^{-1}} \cdot p_{g}$.

We assume that the constraint distribution $\Delta_G \subset TG$ is left-invariant under the group action $L_h : G \to G$, $g \mapsto hg$, that is,
\begin{equation}\label{invariance_Delta}
\Delta_G(hg) = T_g L_h \bigl(\Delta_G(g)\bigr)=h \cdot \Delta _G(g), \quad \text{for all } h \in G.
\end{equation}

Now, recall from the definition of induced Dirac structures in \eqref{IndDiracStrCot}, we consider the case in which $Q=G$ to define an induced Dirac structure $D_{\Delta_G}$ on $T^{\ast}G$ from the canonical two-form $\Omega_{T^{\ast}G}$ on $T^{\ast}G$ and the lifted distribution $\Delta_{T^{\ast}G}=(T\pi_{G})^{-1}(\Delta_{G})$ by, for each $p_g \in T^{\ast}G$,
\begin{align}\label{IndDirStrLie}
D_{\Delta_G}(p_{g})
& =\{ (v_{p_{g}}, \alpha_{p_{g}}) \in T_{p_{g}}T^{\ast}G \times
T^{\ast}_{p_{g}}T^{\ast}G  \mid v_{p_{g}} \in
\Delta_{T^{\ast}G}(p_{g}),  \; \mbox{and} \;  \nonumber
\\ & \qquad \qquad
\left<\alpha_{p_{g}}, w_{p_{g}} \right> = \Omega_{T^{\ast}G}( p _g) (v_{p_{g}},w_{p_{g}}) \;\; \mbox{for
all} \;\; w_{p_{g}} \in \Delta_{T^{\ast}G}(p_{g})\},
\end{align}
where $\pi_{G}: T^{\ast}G \to G$ is the canonical projection.
\medskip

By a slight abuse of notation, we sometimes write $p_g=(g,p)$ for
an element of $T^*G$ when separating base and fiber variables is convenient.
Then the cotangent-lifted left action is written as
\[
\Psi_h(p_{g})=\Psi_h(g,p):=(hg,hp).
\]
Since  the left translation action of $G$ on itself is free and proper, its cotangent lift to $T^{\ast}G$ is also free and proper. From the invariance of $\Delta _G$ and $\Omega_{T^*G}$, the induced Dirac structure $D_{\Delta_{G}}$ on $T^{\ast}G$ is also $G$--invariant under $\Psi _h$, i.e., 
\begin{equation*}
((\Psi _{h})_{\ast}X,(\Psi_{h}^{-1})^{\ast}\alpha) \in D_{\Delta_{G}},
\end{equation*}
for all $(X, \alpha) \in D_{\Delta_{G}}$ and $ h \in G$, as in \eqref{invariance_Dirac}.

\paragraph{Trivialization diffeomorphisms.}
For the tangent bundle $TG$ of a Lie group $G$, we consider the left trivialization diffeomorphism:
\begin{equation*}\label{lambda}
\lambda_{TG}: TG  \to G \times \mathfrak{g} : \;\; \lambda_{TG}(g,v)=(g, T_{g}L_{g^{-1}}\cdot v),
\end{equation*}
in which the tangent bundle $TG$ is diffeomorphic to $G \times \mathfrak{g}$. 

On the other hand, for the cotangent bundle $T^{\ast}G$ of $G$,  the left trivialization diffeomorphism is given by:
\begin{eqnarray*}\label{bar_lambda}
\lambda_{T^{\ast}G}:  T^{\ast}G \to G \times \mathfrak{g}^{\ast} : \;\; \lambda_{T^{\ast}G}(g,p)=(g,T_{e}^{\ast}L_{g} \cdot p),
\end{eqnarray*}
where the cotangent bundle $T^{\ast} G$ is diffeomorphic to $G \times \mathfrak{g}^{\ast}$.
\medskip

Associated with the lifted distribution $\Delta_{T^{\ast}G}=(T\pi_{G})^{-1}(\Delta_{G})$, the trivialized lifted distribution $\Delta_{G \times \mathfrak{g}^{\ast}}$ on $G \times \mathfrak{g}^\ast $ is defined by 
\[
\Delta_{G \times \mathfrak{g}^{\ast}}:=(T\overline{\pi}_{G})^{-1}(\Delta_{G}), 
\]
where $\overline{\pi}_{G}: G \times \mathfrak{g}^{\ast} \to G$ is the trivialized projection  defined by $\pi_{G}=\overline{\pi}_{G} \circ \lambda_{T^{\ast}G}$. 
For each $(g, \mu )= \lambda_{T^{\ast}G}(g,p) \in G \times \mathfrak{g}^\ast $, we obtain
\[
\Delta_{G \times \mathfrak{g}^{\ast}}(g,\mu)=\left\{(v ,\rho) \in T_{g}G \times \mathfrak{g}^{\ast} \mid v \in \Delta_G(g) \right\}= \Delta _G(g)\times \mathfrak{g}^\ast.
\]
Employing the left trivialization diffeomorphism $\lambda_{T^{\ast}G}:T^{\ast}G \to G \times \mathfrak{g}^{\ast}$, the one-form $\theta$ on $G \times \mathfrak{g}^{\ast}$ and the two-form $\omega$ on $G \times \mathfrak{g}^{\ast}$ are defined by
\[
\theta=(\lambda_{T^{\ast}G}^{-1})^{\ast}\Theta \quad \textrm{and} \quad \omega=(\lambda_{T^{\ast}G}^{-1})^{\ast}\Omega,
 \]
 where $\Theta$ is the canonical one-form on $T^{\ast}G$ and $\Omega=-\mathbf{d}\Theta$ is the canonical two-form. In fact, the one-form $\theta$ on $G \times \mathfrak{g}^{\ast}$ is locally expressed, at each $(g,\mu) \in G \times \mathfrak{g}^{\ast}$ and $(v,\rho) \in T_{(g,\mu)}(G \times \mathfrak{g}^{\ast}) \cong T_{g}G \times \mathfrak{g}^{\ast}$, by
\[
\theta(g,\mu) \cdot (v,\rho)=\mu(T_{g}L_{g^{-1}} \, v).
\]
The symplectic structure $\omega=-\mathbf{d}\theta$ on $G \times \mathfrak{g}^{\ast}$,  at  each point $(g,\mu) \in G \times \mathfrak{g}^{\ast}$, is therefore represented by
 \begin{equation*} 
\label{symGtig}
\omega(g,\mu)((v ,\rho),(w ,\sigma))
= \langle \sigma, g ^{-1}v  \rangle- \langle \rho, g ^{-1} w  \rangle
     + \langle \mu, [g ^{-1}v , g ^{-1}w ] \rangle,
\end{equation*}
where $(w,\sigma) \in T_{(g,\mu)}(G \times \mathfrak{g}^{\ast}) \cong T_{g}G \times \mathfrak{g}^{\ast}$.
\medskip

Hence, the trivialized Dirac structure $\overline{D}_{\Delta_{G}}\subset T(G \times\mathfrak{g}^\ast ) \oplus T^*(G \times \mathfrak{g}^\ast )$ on $G \times \mathfrak{g}^{\ast}$ is obtained as
\[
\overline{D}_{\Delta_{G}}=\mathcal{F}(T\lambda_{T^{\ast}G})(D_{\Delta_{G}})=(\lambda_{T^{\ast}G})_{\ast}D_{\Delta_{G}},
\]
which is locally given, for each $(g,\mu) \in G \times \mathfrak{g}^\ast$, by
\[
\begin{split}
\overline{D}_{\Delta_{G}}(g,\mu)&=\{  \left((v ,\rho),(p ,\eta)\right) \in (T_{g}G \times \mathfrak{g}^{\ast}) \times (T_{g}^{\ast}G \times \mathfrak{g}) \mid (v,\rho) \in  \Delta_{G \times \mathfrak{g}^{\ast}}(g,\mu) , \\
\text{and} &\;\; \langle p , w  \rangle +  \langle \sigma, \eta \rangle =\omega(g,\mu)\left((v,\rho),(w ,\sigma)\right), \;\; \text{for all} \;\;  (w ,\sigma) \in \Delta_{G \times \mathfrak{g}^{\ast}}(g,\mu) \}.
\end{split}
\]

\paragraph{Invariance of the induced Dirac structure.} Associated with the $G$-invariance \eqref{invariance_Delta} of the distribution $ \Delta _G$, we can uniquely determine its value at $e$ as
\[
\mathfrak{g}  ^ \Delta  := \Delta _G(e)\subset T_e G= \mathfrak{g}  .
\]
Similarly, the distribution
$ \Delta _{G \times \mathfrak{g}^\ast }\subset T( G \times \mathfrak{g}^\ast )$ is completely determined by its value at $(e, \mu ) \in G \times \mathfrak{g}^\ast$ as
\[
\Delta _{G \times \mathfrak{g}^\ast }( e , \mu )= \mathfrak{g}  ^ \Delta \times \mathfrak{g}^\ast\subset T_{(e, \mu )}(G \times \mathfrak{g}^\ast ).
\]
Denote by $\bar{\Psi } _h(g, \mu ):=(hg, \mu )$ the left $G$-action on $G \times \mathfrak{g}^\ast$ induced, via $\lambda_{T^{\ast}G}$, by the cotangent lift of left translation by $G$ on $T^*G$. 
Since both $\omega$ and $\Delta_{G\times\mathfrak g^*}$ are $G$-invariant, the induced Dirac structure $\overline{D}_{\Delta_{G}}$ on $G \times \mathfrak{g}^{\ast}$ is also $G$--invariant under $\bar{\Psi}_h$, i.e., 
\begin{equation*}
((\bar{\Psi }_{h})_{\ast}X,(\bar{\Psi }^{-1} _{h})^{\ast}\alpha) \in \overline{D}_{\Delta_{G}},
\end{equation*}
for all $(X, \alpha) \in \overline{D}_{\Delta_{G}}$ and $ h \in G$, as in \eqref{invariance_Dirac}. 
\medskip

Using the natural identifications
\[
T_{(e,\mu)}(G\times\mathfrak g^\ast)\cong \mathfrak{g}  \times \mathfrak{g}^\ast,
\qquad
T^*_{(e,\mu)}(G\times\mathfrak g^\ast)\cong \mathfrak{g}^\ast  \times \mathfrak{g},
\]
the fiber $\overline D_{\Delta_G}(e,\mu)$ is uniquely determined by its expression at $(e,\mu)$ as
\begin{equation*}
\begin{split}
\overline{D}_{\Delta_{G}}(e,\mu)=&\{ ((\xi,\rho),(\beta ,\eta)) \in (\mathfrak{g}  \times \mathfrak{g}^\ast) \times ( \mathfrak{g}^\ast  \times \mathfrak{g}) \mid  (\xi,\rho) \in  \mathfrak{g}^{\Delta} \times \mathfrak{g}^{\ast}, \\
& \; \text{and} \;  \langle \beta , \zeta \rangle +  \langle \sigma, \eta \rangle=\omega(e,\mu)((\xi,\rho),(\zeta,\sigma)) \;\;\text{for all} \;\; (\zeta,\sigma) \in \mathfrak{g}^{\Delta} \times \mathfrak{g}^{\ast} \}, \nonumber
\end{split}
\end{equation*}
where $\xi=g^{-1}v$,  $\zeta=g^{-1}w$, and $\beta=g^{-1}p$.  
\paragraph{Lie--Dirac bundle reduction.} 
Using the left trivialization $T^\ast G \cong G \times \mathfrak{g}^\ast $, 
the quotient of the Pontryagin bundle $TT^{\ast}G \oplus T^{\ast}T^{\ast}G$ by $G$ is identified with
\[
F:=\left( TT^*G \oplus T^*T^*G \right) /G\cong Y \oplus Y^{\ast}\rightarrow B:=T^{\ast}G/G \cong \mathfrak{g}^\ast.
\]
In the above, we have the identifications:
\[
TT^*G/G \cong Y:=\mathfrak{g}^* \times (\mathfrak{g}  \times \mathfrak{g}^\ast), 
\quad \textrm{and}\quad 
T^{\ast}T^{\ast}G/G \cong Y^{\ast}:=\mathfrak{g}^* \times (\mathfrak{g}^\ast  \times \mathfrak{g}).
\]

Then, by taking the quotient of $\overline{D}_{\Delta_{G}}$ by $G$, we define the quotient bundle 
\[
D_{\Delta_{G}}^{/G}:=\overline{D}_{\Delta_{G}}/G \subset Y \oplus Y^{\ast}, 
\]
whose fiber, at each $\mu \in \mathfrak{g}^{\ast}$, is given by
\begin{equation}\label{reducedInducedDirac}
\begin{split}
D_{\Delta_{G}}^{/G}(\mu)=& \{ ((\xi,\rho),(\beta ,\eta)) \in (\mathfrak{g}  \times \mathfrak{g}^\ast) \times  (\mathfrak{g}^\ast  \times \mathfrak{g}) \mid 
(\xi,\rho) \in \mathfrak{g}^{\Delta} \times \mathfrak{g}^{\ast}, \\
&\text{and} \;\;  \langle \beta , \zeta \rangle +  \langle \sigma, \eta \rangle =\omega^{/G}(\mu)((\xi,\rho),(\zeta,\sigma)) \;\; \text{for all} \;\; (\zeta,\sigma) \in \mathfrak{g}^{\Delta} \times \mathfrak{g}^{\ast}
\}.
\end{split}
\end{equation}
Here, $\omega^{/G}$ is the {\it reduced two-form on the bundle $Y=\mathfrak{g}^* \times (\mathfrak{g}  \times \mathfrak{g}^\ast)$} given fiberwise, for each $\mu \in \mathfrak{g}^{\ast}$, by
\begin{equation*}
\omega^{/G}(\mu)((\xi,\rho),(\zeta,\sigma))= \left\langle \sigma, \xi \right\rangle
- \left\langle \rho, \zeta \right\rangle + \left\langle \mu, [\xi,\zeta] \right\rangle,
\end{equation*}
which combines the canonical symplectic term with the Lie--Poisson term. 
\medskip

Furthermore, the quotient bundle $D_{\Delta_{G}}^{/G}(\mu) \subset Y \oplus Y^{\ast}$ is locally given, for each $\mu \in \mathfrak{g}^{\ast}$, by
\begin{equation*}
D_{\Delta_{G}}^{/G}(\mu)= \{ ((\xi,\rho),(\beta ,\eta )) \in (\mathfrak{g}  \times \mathfrak{g}^\ast) \times  (\mathfrak{g}^\ast  \times \mathfrak{g})\mid  \xi=\eta,\;\;  \xi \in \mathfrak{g}^{\Delta},\;
\beta  + \rho- \operatorname{ad}_{\xi}^{\ast}\mu \in (\mathfrak{g}^{\Delta})^{\circ} \},
\end{equation*}
where $(\mathfrak g^\Delta)^\circ \subset \mathfrak g^*$ denotes the annihilator of $\mathfrak g^\Delta$. It was shown by \cite{YoMa2007a} that the quotient bundle $D_{\Delta_{G}}^{/G} \subset Y \oplus Y^{\ast}$ is a Dirac structure on the bundle $Y= \mathfrak{g}^* \times (\mathfrak{g}  \times \mathfrak{g}^\ast) \to \mathfrak g^*$.
\medskip

Figure \ref{DiracBundleReduction_dualLieAlgebra} illustrates this reduction procedure. The reduced Dirac structure $D_{\Delta_{G}}^{/G} \subset Y \oplus Y^{\ast}$ is called the \textit{Lie--Dirac structure}  on the reduced bundle $Y$ over $\mathfrak{g}^\ast$. For further details, see \cite{YoMa2007a, GBYo2015}.

\begin{figure}[htpb]
\adjustbox{scale=1.0,center}{
\begin{tikzcd}
TT^{\ast}G \oplus T^{\ast}T^{\ast}G \arrow[d, "/G"'] 
  & D_{\Delta_{G}} \arrow[l, hook'] \arrow[d, "/G"] \\
F=(TT^{\ast}G/G) \oplus (T^{\ast}T^{\ast}G/G)\; \cong\;  Y \oplus Y^{\ast}  \arrow[d] 
  & D_{\Delta_{G}}^{/G} \arrow[l, hook']\\
 B=T^{\ast}G/G \cong\mathfrak{g}^{\ast}
\end{tikzcd}
}
\caption{Lie--Dirac reduction.}
\label{DiracBundleReduction_dualLieAlgebra}
\end{figure}

%
\subsection{Reduction of  Lagrange--Dirac structures}\label{Sec:LagDiracStr_LieGroup}
\paragraph{Lagrange--Dirac structures over Lie groups.} 
Let $L$ be a Lagrangian on $TQ$, \textit{possibly degenerate}. Recall from  \eqref{LagDirac} that a Lagrange--Dirac structure $D_{L}\subset TTQ\oplus T^*TQ$ on $TQ$ is defined by a distribution $\Delta_Q \subset TQ$ and the Lagrangian two-form $\Omega_{L}$ associated with $L$. We now specialize to the case $Q=G$ and consider the reduction of the Lagrange--Dirac structure on $TG$. We assume that $L$ is left invariant and also that a given distribution $\Delta_{G}$ on $G$ is left invariant, as before. Then, we define a lifted distribution on $TG$ by $\Delta_{TG}=(T\tau_G)^{-1}(\Delta_G) \subset TTG$, where $\tau_G: TG \to G$ is the canonical projection.

Recall that the Lagrangian forms on $TG$ are defined by $\Theta_{L}=(\mathbb{F}L)^{\ast}\Theta_{T^{\ast}G}$ and $\Omega_{L}=(\mathbb{F}L)^{\ast}\Omega_{T^{\ast}G}$.
Then the \textit{Lagrange--Dirac structure} $D_L$ on $TG$ is defined by, for each $v_g \in TG$,
\begin{align}\label{LagDiracStr_LieGroup}
D_{L}(v_g) 
& =\big\{ (w_{v_{g}}, \alpha_{v_{g}}) \in T_{v_g}TG \times
T^{\ast}_{v_g}TG\mid w_{v_{g}} \in
\Delta_{TG}(v_g), \;   \mbox{and}\nonumber\\
&\hspace{1cm}
\left<\alpha_{v_{g}}, u_{v_{g}}\right>=\Omega_L(v_g) (w_{v_{g}},u_{v_{g}}), \;\textrm{for all} \;\; u_{v_{g}} \in \Delta_{TG}(v_g)\big\}.
\end{align}
Recall also that this definition is well defined even for degenerate Lagrangians.

\begin{remark}
We can also construct the Lagrange--Dirac structure $D_{L}$ on $TG$ from the induced Dirac structure $D_{\Delta_{G}}$ on $T^{\ast}G$ by the backward Dirac map under the Legendre transformation $\mathbb{F}L:TG \to T^{\ast}G$ as:
\begin{equation*}
D_{L}=\mathcal{B}(T\mathbb{F}L)(D_{\Delta_{G}})= (\mathbb{F}L)^{\ast} D_{\Delta_G}.
\end{equation*}
This expression is locally written as, for all $v_{g} \in TG$ and $p_{g}=\mathbb{F}L(v_{g})$,
\begin{equation}\label{eq:DLpullback}
\begin{split}
D_{L}(v_{g})=\{ (w_{v_{g}},(T_{v_{g}}\mathbb{F}L)^{\ast}(\gamma_{p_{g}}) ) \, \mid \, 
w_{v_{g}}\in T_{v_{g}}TG,\; \gamma_{p_{g}} \in T^{\ast}_{p_{g}}T^{\ast}G, \\
(T_{v_{g}}\mathbb{F}L(w_{v_{g}}), \gamma_{p_{g}}) \in D_{\Delta_{G}}(p_{g}) \}.
\end{split}
\end{equation}
\end{remark}

\paragraph{Invariance of the Lagrange--Dirac structure.}
Denote by $\Phi _h(g,v):=(hg, T_{g}L_{h} \cdot v)$ the left $G$-action on $TG$ induced by the tangent lift of left translation by $G$. Since  the left translation action of $G$ on itself is free and proper, its tangent lifts to $TG$ is also free and proper. From the invariance of $ \Delta _G$, the Lagrange--Dirac structure $D_{L}$ on $TG$ is also $G$--invariant under $\Phi _h$, i.e., 
\begin{equation*}
((\Phi _{h})_{\ast}X,(\Phi ^{-1} _{h})^{\ast}\alpha) \in D_{L},
\end{equation*}
for all $(X, \alpha) \in D_{L}$ and $ h \in G$.

\paragraph{Trivialization of the Lagrange--Dirac structure.}
Since the Lagrangian $L:TG \to \mathbb{R}$ is left invariant on $TG$, we have
\[
L(hg, T_{g}L_{h} \cdot v)=L(g,v)
\]
for all $g,h \in G$ and $v \in T_{g}G$. Using the left trivializing diffeomorphism 
\[
\lambda_{TG}:TG \to G \times \mathfrak{g}; \; (g,v) \mapsto \left(g, T_gL_{g^{-1}}v \right), 
\]
we define the trivialized Lagrangian $\overline{L}=L \circ \lambda_{TG}^{-1}$ induced on $G \times \mathfrak{g}$, which is also left $G$--invariant. Therefore, we define the reduced Lagrangian $\ell: \mathfrak{g} \to \mathbb{R}$ by
\[
\ell(\eta):=\overline{L}(e,\eta).
\]

Using the left trivializing diffeomorphism $\lambda_{TG}:TG \to G \times \mathfrak{g}$, we define the Lagrangian one-form $\theta$ on $G \times \mathfrak{g}$ and the Lagrangian two-form $\omega$ on $G \times \mathfrak{g}$ as
\[
\theta_{\ell}=(\lambda_{TG}^{-1})^{\ast}\Theta_L \quad \textrm{and} \quad \omega_{\ell}=(\lambda_{TG}^{-1})^{\ast}\Omega_L.
 \]
By direct computations, it follows that the Lagrangian one-form $\theta_{\ell}$ on $G \times \mathfrak{g}$ is locally represented at each $(g,\eta) \in G \times \mathfrak{g}$ and $(\delta{g},\delta{\eta}) \in T_{(g,\eta)}(G \times \mathfrak{g}) \cong T_{g}G \times \mathfrak{g}$, as
\[
\theta_{\ell}(g,\eta) \cdot (\delta{g},\delta{\eta})=\langle\mathbf{D} \ell(\eta), T_{g}L_{g^{-1}} \, \delta{g}\rangle.
\]
The trivialized Lagrangian two-form $\omega_{\ell}=-\mathbf{d}\theta_{\ell}$ on $G \times \mathfrak{g}$ is obtained by, at  each point $(g,\eta) \in G \times \mathfrak{g}$, 
\begin{equation} \label{Triv_LagTwoForm}
\begin{split}
\omega_{\ell}(g,\eta)((\dot{g},\dot{\eta}),(\delta{g},\delta{\eta}))&= \left< \mathbf{D}^{2} \ell(\eta) \cdot \delta \eta, T_{g}L_{g^{-1}}\dot{g}  \right>- \left< \mathbf{D}^{2}\ell(\eta) \cdot \dot \eta, T_{g}L_{g^{-1}}\delta{g}  \right>\\
&
  \hspace{3cm}   + \left<\mathbf{D} \ell(\eta), [T_{g}L_{g^{-1}}\dot{g}, T_{g}L_{g^{-1}}\delta{g}] \right>,
\end{split}
\end{equation}
where $(\dot{g},\dot{\eta}) , (\delta{g},\delta{\eta}) \in T_{(g,\eta)}(G \times \mathfrak{g}) \cong T_{g}G \times \mathfrak{g}$. Note that we utilized the formula
\[
\omega_{\ell}(X,Y)=-\mathbf{d}\theta_{\ell}(X,Y)=-X[\theta_{\ell}(Y)]+Y[\theta_{\ell}(X)]+\theta_{\ell}([X,Y]),
\]
for vector fields $X, Y \in \mathfrak{X}(G \times \mathfrak{g})$.

The above expression  \eqref{Triv_LagTwoForm} is the Lagrangian counterpart of the reduced symplectic form on $G \times \mathfrak{g}^{\ast}$, where the Lie--Poisson term is replaced by the Lie bracket term weighted by $\mathbf{D} \ell$.
\medskip

Associated with the given distribution $\Delta_{G}$ on $G$, the constraint distribution lifts to
\[
\Delta_{TG}=(T\tau_{G})^{-1}(\Delta_{G}), 
\]
and hence we obtain the trivialized lifted distribution $\Delta_{G \times \mathfrak{g}}$ on $G \times \mathfrak{g}$ by 
\[
\Delta_{G \times \mathfrak{g}}:=(T\overline{\tau}_{G})^{-1}(\Delta_{G}), 
\]
where $\overline{\tau}_{G}: G \times \mathfrak{g} \to G$ is the projection defined by $\tau_{G}=\overline{\tau}_{G} \circ \lambda_{TG}$. For each $(g, \eta ) \in G \times \mathfrak{g}$, we obtain
\[
\Delta_{G \times \mathfrak{g}}(g,\eta)=\left\{(\dot{g} ,\dot{\eta}) \in T_{g}G \times \mathfrak{g} \mid \dot{g} \in \Delta_G(g) \right\}= \Delta _G(g) \times \mathfrak{g}.
\]

The trivialized Lagrange--Dirac structure on $G \times \mathfrak{g}$ is obtained by using the trivialized diffeomorphism $\lambda_{TG}:TG \to G \times \mathfrak{g}$ as
\[
\overline{D}_{\ell} := \mathcal{F}(T\lambda_{TG}) = (\lambda_{TG})_{\ast}D_{L},
\]
such that, for each $v_{g} \in TG$,
\begin{equation*}
\begin{split}
\overline{D}_{\ell}(\lambda_{TG}(v_{g}))=\big\{ ( T_{v_{g}}\lambda_{TG}(w_{v_{g}}),\alpha_{\lambda_{TG}(v_{g})} ) \mid w_{v_{g}}& \in T_{v_{g}}TG,\;  \alpha_{\lambda_{TG}(v_{g})}  \in T^{\ast}_{\lambda_{TG}(v_{g})}(G \times \mathfrak{g}),\; \\ 
& \quad (w_{v_{g}}, (T_{v_{g}}\lambda_{TG})^{\ast}(\alpha_{\lambda_{TG}(v_{g})}))  \in D_{L}(v_{g}) \big\},
\end{split}
\end{equation*}
where $T_{v_{g}}\lambda_{TG}: T_{v_{g}}TG \to T_{\lambda_{TG}(v_{g})}(G \times \mathfrak{g})$ is the tangent map of $\lambda_{TG} : TG \to G \times \mathfrak{g}$.
\medskip

Thus the local expression of $\overline{D}_{\ell}$ is given by, for each $(g,\eta)=\lambda_{TG}(g,v)\in G \times \mathfrak{g}$, 
\begin{equation*}
\begin{split}
\overline{D}_{\ell}(g,\eta)&=\big\{ (\dot{g} ,\dot{\eta}),(p ,\rho) \in (T_{g}G \times \mathfrak{g}) \times (T_{g}^{\ast}G \times \mathfrak{g}^{\ast})  \mid (\dot{g} ,\dot{\eta}) \in  \Delta_{G \times \mathfrak{g}}(g,\eta), \\
\text{and} &\;\; \langle p, \delta{g}  \rangle +  \langle \rho, \delta{\eta} \rangle =\omega_{\ell}(g,\eta)((\dot{g},\dot{\eta}),(\delta{g},\delta{\eta})), \;\; \text{for all} \;\; (\delta{g},\delta{\eta}) \in \Delta_{G \times \mathfrak{g}}(g,\eta) \big\}.
\end{split}
\end{equation*}

\paragraph{Invariance of the Lagrange--Dirac structure.} Since the distribution $ \Delta _G$ is $G$-invariant as in \eqref{invariance_Delta}, we define
\[
\mathfrak{g}  ^ \Delta := \Delta _G(e)\subset T_e G= \mathfrak{g}.
\]
The distribution $\Delta_{G \times \mathfrak g} \subset T(G \times \mathfrak g)$ is determined by its value at $(e,\eta) \in G \times  \mathfrak{g}$ as 
\[
\Delta _{G \times \mathfrak{g} }( e , \eta )= \mathfrak{g}  ^ \Delta \times \mathfrak{g}\subset T_{(e, \eta )}(G \times \mathfrak{g}  ).
\]
Let us denote by $\bar{\Phi }_h(g, \eta):=(hg, \eta)$ the left $G$-action on $G \times \mathfrak{g}$ induced, via $\lambda_{TG}:TG \to G \times \mathfrak{g} $, by the tangent lift of left translation by $G$ on $TG$. From the $G$-invariance of $D_{L}$ on $TG$, the trivialized Lagrange--Dirac structure $\overline{D}_{\ell}$ on $G \times \mathfrak{g}$ is also $G$--invariant, i.e.,
\begin{equation*}
((\bar{\Phi}_{h})_{\ast}X,(\bar{\Phi}_{h}^{-1})^{\ast}\alpha) \in \overline{D}_{\ell},
\end{equation*}
for all $(X, \alpha) \in \overline{D}_{\ell}$ and $ h \in G$, as in \eqref{invariance_Dirac}.
\medskip

Using the natural identifications 
\[
T_{(e, \eta )}(G \times \mathfrak{g}  ) \cong \mathfrak{g}  \times \mathfrak{g}, \qquad T^{\ast}_{(e, \eta )}(G \times \mathfrak{g}  ) \cong  \mathfrak{g}^{\ast} \times \mathfrak{g}^{\ast},
\]
the fiber  $\overline{D}_{\ell} (e,\eta)$ is uniquely determined by its value at $(e,\eta) \in G \times \mathfrak{g}$ as
\begin{equation*}
\begin{split}
&\overline{D}_{\ell}(e,\eta)=\big\{ ((\xi,\dot{\eta}),(\beta ,\rho)) \in (\mathfrak{g}  \times \mathfrak{g})\times ( \mathfrak{g}^{\ast} \times \mathfrak{g}^{\ast}) \mid  (\xi,\dot{\eta}) \in \Delta _{G \times \mathfrak{g} }( e , \eta ) = \mathfrak{g}^{\Delta} \times \mathfrak{g} \, \\
& \quad \text{and} \;  \langle \beta , \zeta \rangle +  \langle \rho, \delta{\eta} \rangle=\omega_{\ell}(e,\eta)((\xi,\dot{\eta}),(\zeta,\delta{\eta})) \;\;\text{for all} \;\; (\zeta,\delta{\eta}) \in \Delta _{G \times \mathfrak{g} }( e , \eta ) =\mathfrak{g}^{\Delta} \times \mathfrak{g} \big\}, \nonumber
\end{split}
\end{equation*}
where $\xi=g^{-1}\dot{g}$,  $\zeta=g^{-1}\delta{g}$, and $\beta=g^{-1}p$.  

\paragraph{Reduction of the Lagrangian forms on $TG$.}
Using the trivialization $G\times\mathfrak g \cong TG$, the quotient of the Pontryagin bundle $TTG \oplus T^{\ast}TG$ by $G$ is identified with
\[
F:=\left( TTG \oplus T^{\ast}TG \right) /G\cong A \oplus A^{\ast}\rightarrow B:=TG/G \cong \mathfrak{g}.
\]
In the above, the quotient of the bundle $TTG \to TG$ by $G$ is identified with the trivial bundle
\[
A:=TTG/G \cong \mathfrak{g} \times (\mathfrak g \times \mathfrak g) \to TG/G \cong \mathfrak{g}.
\]
We also define the dual bundle of $A$ by
\[
A^{\ast}:= T^{\ast}TG/G \cong \mathfrak{g} \times (\mathfrak g^{\ast} \times \mathfrak g^{\ast}) \to TG/G \cong \mathfrak{g}.
\]

\begin{definition}[Reduced symmetric pairing]\label{Def:RedSymPair}
Associated with the symmetric pairing $\langle \! \langle\, \cdot, \cdot \rangle\!\rangle$ on $TTG \oplus T^\ast TG$ given in \eqref{SymPairing}, we define the reduced symmetric paring  $\langle \! \langle\, \cdot, \cdot \rangle\!\rangle_{\mathrm{red}}$ on $F=A \oplus A^\ast$, defined for each $\eta \in \mathfrak{g}$ by
\begin{equation}\label{eq:RedSymPair}
\big\langle \! \big\langle\, \big((\xi,\zeta),(\beta, \rho) \big), \big((\bar{\xi},\bar{\zeta}),(\bar{\beta}, \bar{\rho}) \big)\big\rangle\!\big\rangle_{\mathrm{red}}
=\langle \beta, \bar{\xi}\rangle + \langle \rho, \bar{\zeta} \rangle + \langle \bar{\beta}, \xi \rangle + \langle \bar{\rho}, \zeta \rangle,
\end{equation}
where $(\xi,\zeta), (\bar{\xi},\bar{\zeta}) \in \mathfrak{g} \times \mathfrak{g}$ and $(\beta, \rho), (\bar{\beta}, \bar{\rho}) \in \mathfrak{g}^{\ast} \times \mathfrak{g}^{\ast}$.
\end{definition}

\begin{definition}[Reduced Lagrangian one-form]\label{Def:RedLagOneForm}

Let $\ell \colon \mathfrak g \to \mathbb R$ be a reduced Lagrangian. The trivialized Lagrangian one-form $\theta_\ell: G \times \mathfrak g \to T^{\ast}(G \times \mathfrak g)$ is $G$-invariant. Therefore, we define the reduced Lagrangian one-form by taking the quotient of $\theta_\ell$ with respect to the $G$-action as
\[
\theta_\ell^{/G}
:= \theta_\ell/G: \mathfrak{g} \to A^{\ast}=\mathfrak g \times (\mathfrak g^{\ast} \times \mathfrak g^{\ast}),
\]
which is a smooth section of the vector bundle $A^{\ast}=\mathfrak g \times (\mathfrak g^{\ast} \times \mathfrak g^{\ast}) \to \mathfrak g$, locally denoted by, for each $\eta \in \mathfrak{g}$, 
\begin{equation*}
\theta^{/G}_\ell(\eta)(\zeta, \delta{\eta})
:= \left\langle \mathbf D \ell(\eta),\, \zeta\right\rangle,
\end{equation*}
where $(\zeta,\delta{\eta})$ is a tangent vector of $\mathfrak g \times \mathfrak g$. Here, $\theta^{/G}_\ell(\eta)$ depends on $\eta$ only and pairs 
$\mathbf D \ell(\eta)\in\mathfrak g^*$ with $\zeta$, the $\zeta$--component of
the tangent vector.

\end{definition}
\begin{definition}\label{Def:RedLagTwoForm}
Since the trivialized Lagrangian two-form $\omega_{\ell}$ on $G \times \mathfrak{g}$
is $G$-invariant, it induces a reduced Lagrangian two-form by taking the quotient
with respect to the $G$-action. We define
\[
\omega_\ell^{/G}
:= \omega_\ell/G
\in \Gamma \left(\Lambda^2(A)\right),
\]
as a fiberwise skew-symmetric bilinear form on the vector bundle $A$, 
rather than as a differential two-form on the manifold $\mathfrak g$, since each fiber at $\eta\in\mathfrak g$ is naturally identified with
$\mathfrak g \times \mathfrak g$. Here, $\Lambda^2 (A)$ denotes the vector bundle over $\mathfrak{g}$, whose fiber at $\eta\in\mathfrak g$
consists of skew-symmetric bilinear forms on $\mathfrak g \times \mathfrak g$.

In particular, the reduced Lagrangian two-form $\omega_{\ell}^{/G}$ is locally given by, for each $\eta\in\mathfrak g$,
\begin{equation}\label{RedLagForm}
\omega_{\ell}^{/G}(\eta)((\xi,\dot{\eta}),(\zeta,\delta{\eta}))
= \left\langle \mathbf{D}^{2} \ell(\eta) \cdot \delta{\eta}, \xi \right\rangle
- \left\langle \mathbf{D}^{2} \ell(\eta) \cdot \dot{\eta}, \zeta \right\rangle
+ \left\langle \mathbf{D}\ell(\eta), [\xi,\zeta] \right\rangle,
\end{equation}
where $(\xi,\dot{\eta}),(\zeta,\delta{\eta}) \in \mathfrak g \times \mathfrak g$. On the right-hand side, the first two terms correspond to the canonical  pairing induced by the Hessian of $\ell$,
whereas the last term is the Lagrangian analogue of the Lie--Poisson structure.
\end{definition}
%
\begin{definition}\label{Def:RedExtDer}
We define the reduced exterior derivative
\[
\mathbf d^{/G}:\Gamma\left(\Lambda^k(A) \right)\to\Gamma\left(\Lambda^{k+1}(A)\right)
\]
by requiring the diagram
\[
\begin{tikzcd}
\Omega^k\left(G\times\mathfrak g \right)^G 
\arrow[r,"\mathbf d"] 
\arrow[d,"/G"']  
& 
\Omega^{k+1}(G\times\mathfrak g)^G 
\arrow[d,"/G"]\\
\Gamma\left(\Lambda^k(A)\right) 
\arrow[r,"\mathbf d^{/G}"] 
& 
\Gamma\left(\Lambda^{k+1}(A)\right)  
\end{tikzcd}
\]
to commute, where $\Omega^k(G\times\mathfrak g)^G$ denotes the space of $G$-invariant $k$-forms on $G\times\mathfrak g$, i.e., those satisfying
$\Phi_h^*\alpha=\alpha$ for all $h\in G$. Here we note that since the quotient map $/G:\Omega^k(G\times\mathfrak g)^G
\to \Gamma\left(\Lambda^k(A)\right)$
is a vector space isomorphism via the left trivialization $G\times\mathfrak g \cong TG$, the operator
$\mathbf d^{/G}$ is well defined.
\end{definition}

\begin{remark}
Throughout this paper, we follow the notation of \cite{AbMaRa1988} and write
$\Gamma(\Lambda^k(A))$ for the space of exterior $k$-forms on $A$, where $\Lambda^k(A):=L^k_a(A; \mathbb{R})$ denotes the space of continuous multilinear alternating maps.
\end{remark}

\begin{proposition}
For the reduced Lagrangian one-form $\theta_\ell^{/G}:= \theta_\ell/G\in \Gamma\!\left(\Lambda^{1}(A)\right)$,  the reduced Lagrangian two-form $\omega_{\ell}^{/G}\in \Gamma\left(\Lambda^{2}(A)\right)$ can also be given by
\[
\omega_{\ell}^{/G}=-\mathbf{d}^{/G}\theta_\ell^{/G}.
\]
\end{proposition}
\begin{proof}
This follows directly from Definition \ref{Def:RedLagOneForm}, \ref{Def:RedLagTwoForm}, and \ref{Def:RedExtDer}. Specifically, since the diagram in Definition  \ref{Def:RedExtDer} commutes and the unreduced two-form satisfies
\[
\omega_{\ell}^{/G} = \omega_\ell / G = (-\mathbf{d}\theta_\ell) / G = -\mathbf{d}^{/G}(\theta_\ell / G) = -\mathbf{d}^{/G}\theta_\ell^{/G}.
\]
\end{proof}
The reduced exterior derivative $\mathbf{d}^{/G}$ is well defined because the standard exterior derivative $\mathbf{d}$ preserves $G$-invariance of forms on $G\times \mathfrak{g}$. Consequently, the explicit expression of the reduced Lagrangian two-form \eqref{RedLagForm} is consistent with the relation  $\omega_{\ell}^{/G}=-\mathbf{d}^{/G}\theta_\ell^{/G}$, effectively incorporating both the Hessian and the Lie--Poisson structure (the Lie bracket and the coadjoint action).

\begin{remark}
Consider a $G$-principal bundle $M \to M/G$. The quotient bundle $TM/G \to M/G$ carries a natural Lie algebroid structure, which is known as the Atiyah algebroid, see \cite{Mack2005, YoMa2009}.
In the case $M = TG$, under the left trivialization $TG \cong G \times \mathfrak{g}$, the quotient bundle that is induced by the lifted $G$-action on $TTG$ , i.e., 
\[
A := TTG/G\cong \mathfrak{g} \times (\mathfrak g \times \mathfrak g) \to B := TG/G \cong \mathfrak{g},
\]
is the Atiyah algebroid. This observation provides a natural geometric interpretation of the reduced exterior derivative $\mathbf d^{/G}$; namely, it corresponds to the Lie algebroid differential associated with $A$. The nilpotency $(\mathbf d^{/G})^2=0$ follows from this identification. Further, the Lie bracket term appearing in the reduced Lagrangian two-form $\omega^{/G}_\ell$ originates from the Lie algebroid bracket on $A$.
\end{remark}

%
\paragraph{Reduction of the Lagrange--Dirac structure.} 
Here we illustrate that the reduction of the Lagrange--Dirac structure $D_{L} \subset TTG \oplus T^{\ast}TG$ on $TG$ can be constructed by taking the quotient.

\begin{framed}
\begin{definition}[Lagrange--Dirac reduction]\label{Def:LagDiracRed}
For the Lagrange--Dirac structure $D_{L} \subset TTG \oplus T^{\ast}TG$, by taking the quotient of $\overline{D}_{\ell}=(\lambda_{TG})_{\ast}D_{L}$ by $G$,  define a quotient subbundle 
\begin{equation}\label{quotientLagDirac}
D_{\ell}^{/G}:=\overline{D}_{\ell}/G \subset A \oplus A^{\ast}, 
\end{equation}
whose fiber each $\eta \in \mathfrak{g}$ is given by 
\begin{equation}\label{reduced_LagDirac}
\begin{split}
D_{\ell}^{/G}(\eta)=& \big\{ ((\xi,\dot{\eta}),(\beta ,\rho))\in (\mathfrak g \times \mathfrak g) \times (\mathfrak g^{\ast} \times \mathfrak g^{\ast}) \mid 
(\xi,\dot{\eta}) \in  \Delta_{\mathfrak{g}}(\eta), \\
&\;\;\text{and} \;\;  \langle \beta , \zeta \rangle +  \langle \rho, \delta{\eta} \rangle=\omega_{\ell}^{/G}(\eta)((\xi,\dot{\eta}),(\zeta,\delta{\eta})) \;\;\text{for all} \;\; (\zeta,\delta{\eta}) \in \Delta_{\mathfrak{g}}(\eta) \big\}.
\end{split}
\end{equation}
Equivalently, it is given by
\begin{equation*}
\begin{split}
D_{\ell}^{/G}(\eta)&= \big\{ ((\xi,\dot{\eta}),(\beta ,\rho )) \in (\mathfrak g \times \mathfrak g) \times (\mathfrak g^{\ast} \times \mathfrak g^{\ast}) \mid (\xi,\dot{\eta}) \in \Delta_{\mathfrak{g}}(\eta)\big.\\
&\big. \hspace{5cm} (\beta ,\rho )-\mathbf{i}_{(\xi,\dot{\eta})}\omega_{\ell}^{/G}(\eta) \in (\Delta_{\mathfrak{g}}(\eta))^{\circ}, \big\},
\end{split}
\end{equation*}
where
$\Delta_{\mathfrak{g}}(\eta)$ is the \textit{reduced constraint subspace} defined by the quotient of $\Delta _{G \times \mathfrak{g} }( e , \eta )$ as
\[
\Delta_{\mathfrak{g}}(\eta):=\Delta _{G \times \mathfrak{g} }( e , \eta )/G \cong \mathfrak{g}^ \Delta   \times \mathfrak{g} .
\]
\end{definition}
\end{framed}
We consider the local expression of the structure $D_{\ell}^{/G} \subset A \oplus A^{\ast}$ in the following proposition. 
\begin{proposition}
The quotient subbundle $D_{\ell}^{/G} \subset A \oplus A^{\ast}$ in \eqref{quotientLagDirac} is given, for each $\eta \in \mathfrak{g}$, by
\begin{equation}\label{LocReducedLagDiracAlt}
\begin{split}
D_{\ell}^{/G}(\eta)&= \big\{ ((\xi,\dot{\eta}),(\beta ,\rho )) \in (\mathfrak g \times \mathfrak g) \times (\mathfrak g^{\ast} \times \mathfrak g^{\ast}) \mid  \rho=\mathbf{D}^{2} \ell(\eta) \cdot  \xi, \;\; \xi \in \mathfrak{g}^{\Delta},\big.\\
&\hspace{5cm}\big.\beta +\mathbf{D}^{2} \ell(\eta) \cdot \dot{\eta}-\mathrm{ad}^{\ast}_{\xi}\,\mathbf{D} \ell(\eta) \in (\mathfrak{g}^{\Delta})^{\circ} \big\}.
\end{split}
\end{equation}
\end{proposition}
\begin{proof}
Suppose that the condition 
\[
((\xi,\dot{\eta}),(\beta ,\rho)) \in D_{\ell}^{/G}(\eta).
\]
From \eqref{RedLagForm} and \eqref{reduced_LagDirac}, then it follows
\[
\begin{split}
\langle \beta , \zeta \rangle +  \langle \rho, \delta{\eta} \rangle&=\omega_{\ell}^{/G}(\eta)((\xi,\dot{\eta}),(\zeta,\delta{\eta}))\\
&=\left\langle \mathbf{D}^{2} \ell (\eta)\cdot \delta{\eta}, \xi \right\rangle
- \left\langle \mathbf{D}^{2} \ell(\eta) \cdot \dot{\eta}, \zeta \right\rangle + \left\langle \mathbf{D} \ell(\eta), [\xi,\zeta] \right\rangle
\end{split}
\]
for all $(\zeta,\delta{\eta}) \in \mathfrak{g}^{\Delta} \times \mathfrak{g}$. 
Using the coadjoint operator $\mathrm{ad}^*_{\xi}: \mathfrak{g}^* \to \mathfrak{g}^*$ defined by
\[
\langle \mathrm{ad}^*_{\xi}\mu,\eta\rangle = \langle \mu,[\xi,\eta]\rangle,
\]
we rewrite the above identity
\[
\langle \beta +\mathbf{D}^{2} \ell (\eta)\cdot \dot{\eta}-\mathrm{ad}^{\ast}_{\xi}\,\mathbf{D} \ell(\eta),\; \zeta \rangle+\langle \rho-\mathbf{D}^{2} \ell (\eta)\cdot  \xi , \; \delta{\eta} \rangle=0,
\]
for all $(\zeta,\delta{\eta}) \in \mathfrak{g}^{\Delta} \times \mathfrak{g}$. 
\medskip

Thus, we obtain
\[
 \beta +\mathbf{D}^{2} \ell(\eta) \cdot \dot{\eta}-\mathrm{ad}^{\ast}_{\xi}\,\mathbf{D} \ell (\eta)\in (\mathfrak{g}^{\Delta})^{\circ} ,\quad  \rho=\mathbf{D}^{2} \ell (\eta)\cdot  \xi, \quad \xi \in \mathfrak{g}^{\Delta}.
\]
\end{proof}
\begin{remark}
As for the coadjoint operator $\mathrm{ad}^*_{\xi}: \mathfrak{g}^* \to \mathfrak{g}^*$, in the right-invariant convention, one has
\[
\langle \mathrm{ad}^*_{\xi}\mu,\eta\rangle = -\langle \mu,[\xi,\eta]\rangle.
\]
\end{remark}

\begin{theorem}
The quotient subbundle $D_{\ell}^{/G} \subset A \oplus A^{\ast}$ that is given by \eqref{LocReducedLagDiracAlt} is a reduced Lagrange--Dirac structure on the bundle $A$. 
\end{theorem}
\begin{proof}
For each fixed $\eta \in \mathfrak{g}$, we need to prove that $D_{\ell}^{/G}(\eta)=(D_{\ell}^{/G}(\eta))^{\perp}$.
\medskip

First, let us check $D_{\ell}^{/G}(\eta) \subset (D_{\ell}^{/G}(\eta))^{\perp}$. To do this, let $((\xi,\dot{\eta}),(\beta ,\rho )) \in D_{\ell}^{/G}(\eta)$, where 
\[
\rho=\mathbf{D}^{2}\ell(\eta) \cdot  \xi, \quad \xi \in \mathfrak{g}^{\Delta},\quad \beta +\mathbf{D}^{2} \ell(\eta) \cdot \dot{\eta}-\mathrm{ad}^{\ast}_{\xi}\,\mathbf{D} \ell(\eta) \in (\mathfrak{g}^{\Delta})^{\circ},
\]
and let $((\xi^{\prime},\dot{\eta}^{\prime}),(\beta^{\prime} ,\rho^{\prime} )) \in D_{\ell}^{/G}(\eta)$, where 
\[
\rho^{\prime}=\mathbf{D}^{2} \ell(\eta) \cdot  \xi^{\prime}, \quad \xi^{\prime} \in \mathfrak{g}^{\Delta},\quad \beta^{\prime} +\mathbf{D}^{2} \ell(\eta) \cdot \dot{\eta}^{\prime}-\mathrm{ad}^{\ast}_{\xi^{\prime}}\,\mathbf{D} \ell(\eta) \in (\mathfrak{g}^{\Delta})^{\circ}.
\]
Now, it follows from \eqref{eq:RedSymPair} that the reduced symmetric paring $\langle \! \langle\, \cdot, \cdot \rangle\!\rangle_{\mathrm{red}}$ on $A \oplus A^{\ast}$ is given, at each $\eta \in \mathfrak{g}$, by
\[
\begin{split}
&\langle \! \langle\,  \left((\xi,\dot{\eta}),(\beta ,\rho )\right), \left((\xi^{\prime},\dot{\eta}^{\prime}),(\beta^{\prime} ,\rho^{\prime} )\right) \rangle \! \rangle_{\mathrm{red}}\, \\
&\quad =\langle\beta^{\prime} , \xi\rangle  + \langle\rho^{\prime}, \dot{\eta} \rangle + \langle \beta , \xi^{\prime} \rangle + \langle \rho , \dot{\eta}^{\prime}\rangle\\
&\quad =-\langle\mathbf{D}^{2} \ell(\eta) \cdot \dot{\eta}^{\prime}-\mathrm{ad}^{\ast}_{\xi^{\prime}}\,\mathbf{D}\ell(\eta) , \xi\rangle + \langle\mathbf{D}^{2} \ell(\eta) \cdot  \xi^{\prime}, \dot{\eta} \rangle \\
&\hspace{5cm}- \langle\mathbf{D}^{2} \ell(\eta) \cdot \dot{\eta}-\mathrm{ad}^{\ast}_{\xi}\,\mathbf{D} \ell(\eta) , \xi^{\prime} \rangle + \langle  \mathbf{D}^{2} \ell(\eta) \cdot  \xi , \dot{\eta}^{\prime}\rangle\\
&\quad =\big\{\langle\mathbf{D}^{2} \ell(\eta) \cdot  \dot{\eta}, \xi^{\prime} \rangle
-\langle\mathbf{D}^{2} \ell(\eta) \cdot \dot{\eta}^{\prime}, \xi\rangle
+\langle\mathrm{ad}^{\ast}_{\xi^{\prime}}\,\mathbf{D} \ell(\eta) , \xi\rangle \big\}\\
&\hspace{5cm}+\left\{\langle \mathbf{D}^{2} \ell(\eta) \cdot  \dot\eta^{\prime}, \xi \rangle 
-\langle \mathbf{D}^{2} \ell(\eta) \cdot \dot{\eta}, \xi^{\prime}\rangle 
+\langle \mathrm{ad}^{\ast}_{\xi}\,\mathbf{D} \ell(\eta) , \xi^{\prime}\rangle  \right\}\\
&\quad = \omega_{\ell}^{/G}(\eta)((\xi^{\prime},\dot{\eta}^{\prime}),(\xi,\dot{\eta})) +\omega_{\ell}^{/G}(\eta)((\xi,\dot{\eta}),(\xi^{\prime},\dot{\eta}^{\prime}))=0.
\end{split}
\]
In the above, we used the skew-symmetric property of $\omega_{\ell}^{/G}(\eta)$, and the vanishing follows from the symmetry of the Hessian $\mathbf{D}^{2}\ell(\eta)$ and the skew-symmetry of the coadjoint operator $\mathrm{ad}^{\ast}_{\xi}: \mathfrak{g}^{\ast} \to \mathfrak{g}^{\ast}$. Thus, it follows $D_{\ell}^{/G}(\eta) \subset (D_{\ell}^{/G}(\eta))^{\perp}$.

\smallskip

Second, let us check $(D_{\ell}^{/G}(\eta))^{\perp} \subset D_{\ell}^{/G}(\eta)$. Let $ \left((\xi,\dot{\eta}),(\beta ,\rho )\right) \in (D_{\ell}^{/G}(\eta))^{\perp}$ for each $\eta \in \mathfrak{g}$. Then, we have 
\[
\begin{split}
\langle \! \langle\,  \left((\xi,\dot{\eta}),(\beta ,\rho )\right), \left((w,\nu),(\gamma,\mu)\right) \rangle \! \rangle_{\mathrm{red}}\, =\langle\gamma, \xi\rangle  + \langle\mu, \dot{\eta} \rangle + \langle \beta , w \rangle + \langle \rho , \nu\rangle =0,
\end{split}
\]
for all $ \left((w,\nu),(\gamma,\mu)\right) \in D_{\ell}^{/G}(\eta)$; namely, for all
\[
\gamma + \mathbf{D}^{2} \ell(\eta) \cdot \nu- \mathrm{ad}^{\ast}_{w}\,\mathbf{D} \ell(\eta) \in (\mathfrak{g}^{\Delta})^{\circ}, \quad \mu=\mathbf{D}^{2} \ell(\eta) \cdot w, \quad w \in \mathfrak{g}^{\Delta}.
\]
Therefore, setting $w=0$, it follows from $\mu=\mathbf{D}^{2} \ell(\eta) \cdot w=0$ that we have
\[
\begin{split}
\langle\gamma, \xi\rangle  + \langle \rho , \nu\rangle &=-\langle \mathbf{D}^{2} \ell(\eta) \cdot \nu- \mathrm{ad}^{\ast}_{w} \mathbf{D} \ell(\eta), \xi\rangle+\langle \rho , \nu\rangle\\
&=-\langle \mathbf{D}^{2}\ell(\eta) \cdot \nu, \xi\rangle+\langle \rho , \nu\rangle \\
&=\langle \rho-\mathbf{D}^{2} \ell(\eta) \cdot \xi, \nu\rangle=0,
\end{split}
\]
for all $\nu$. Hence we get $\rho=\mathbf{D}^{2} \ell(\eta) \cdot \xi$. Next, setting $\rho=\mathbf{D}^{2} \ell(\eta) \cdot \xi$ and $w \in \mathfrak{g}^{\Delta}$ is arbitrary, it follows  from $\mu=\mathbf{D}^{2}\ell(\eta) \cdot w$ that
\[
\begin{split}
&\langle\gamma, \xi\rangle  + \langle \mu, \dot{\eta} \rangle + \langle \beta , w \rangle + \langle \rho , \nu\rangle \\
&\quad =-\langle \mathbf{D}^{2} \ell(\eta) \cdot \nu- \mathrm{ad}^{\ast}_{w}\,\mathbf{D} \ell(\eta), \xi\rangle  + \langle \mathbf{D}^{2} \ell(\eta) \cdot w, \dot{\eta} \rangle + \langle\beta , w \rangle + \langle\mathbf{D}^{2} \ell(\eta) \cdot \xi , \nu\rangle\\
&\quad =\langle \mathrm{ad}^{\ast}_{w}\,\mathbf{D} \ell(\eta), \xi\rangle  + \langle\mathbf{D}^{2} \ell(\eta) \cdot w, \dot{\eta} \rangle + \langle\beta , w \rangle \\
&\quad =\langle -\mathrm{ad}^{\ast}_{\xi}\,\mathbf{D} \ell(\eta)+\mathbf{D}^{2}\ell(\eta) \cdot \dot{\eta}+\beta , w \rangle =0,
\end{split}
\]
for all $w \in \mathfrak{g}^{\Delta}$. Therefore, we get
$
-\mathrm{ad}^{\ast}_{\xi}\,\mathbf{D} \ell(\eta)+\mathbf{D}^{2} \ell(\eta) \cdot \dot{\eta}+\beta \in  (\mathfrak{g}^{\Delta})^{\circ}.
$
Hence, we get $(D_{\ell}^{/G}(\eta))^{\perp} \subset D_{\ell}^{/G}(\eta)$. Finally, we prove $D_{\ell}^{/G}(\eta)=(D_{\ell}^{/G}(\eta))^{\perp}$ for each $\eta \in \mathfrak{g}$. Thus $D_{\ell}^{/G} \subset A \oplus A^{\ast}$ is a Dirac structure on $A$.
\end{proof}

We call $D_{\ell}^{/G} \subset A \oplus A^{\ast}$ the \textbf{reduced Lagrange--Dirac structure}, or the \textbf{reduced Lagrange--Dirac bundle}, over the quotient bundle (Atiyah algebroid) $A$. Figure \ref{DiracBundleReduction_LieAlgebra} illustrates the Lagrange--Dirac reduction.
\begin{figure}[h]
\adjustbox{scale=1.0,center}{
\begin{tikzcd}
TTG \oplus T^{\ast}TG \arrow[d, "/G"']
& D_{L} \arrow[l, hook'] \arrow[d, "/G"] \\
F=(TTG \oplus T^{\ast}TG)/G \cong A \oplus A^{\ast} \arrow[d] 
&D^{/G}_{\ell} \arrow[l, hook']\\
B= TG/G \cong \mathfrak{g}
\end{tikzcd}
}
\caption{Lagrange--Dirac reduction.}
\label{DiracBundleReduction_LieAlgebra}
\end{figure}

\begin{remark}[Link to Courant algebroids.]
We can view the Lagrange--Dirac reduction in the context of the reduction of \textit{Courant algebroids}. Let $M$ be a manifold which is a $G$-principal bundle over $B=M/G$ and we consider the bundle $F=(TM \oplus T^{\ast}M)/G$ over $B$. The natural lift of the $G$-action to $TM \oplus T^{\ast}M$ preserves the symmetric pairing $\langle\!\langle \cdot, \cdot \rangle\!\rangle$ as well as the Courant bracket developed in \cite{Cour1990a}. 
Therefore, $F$ naturally inherits the structure of a Courant algebroid over $B$ in the general sense of \cite{LiuWeiXu1998}. Though this Courant algebroid is not of the form $TB \oplus T^{\ast}B$, we still have a Dirac structure in $F$. In fact, if $D_{M} \subset TM \oplus T^{\ast}M$ is a $G$-invariant integrable Dirac structure on $M$, then $D_{M}^{/G}:=D_{M}/G$ is a Dirac subbundle of $F$ and note that $D_{M}^{/G}$ is integrable if $D_{M}$ is, since $G$ preserves the Courant bracket. 
\end{remark}

In the case $M=T^{\ast}G$, the induced Dirac structure $D_{\Delta_{G}} \subset TM \oplus T^{\ast}M $ in \eqref{IndDirStrLie} is reduced to $D_{\Delta_{G}}^{/G}:=D_{\Delta_{G}}/G$, which is a subbundle of $F=(TM \oplus T^{\ast}M)/G \cong Y \oplus Y^{\ast}$ over $B=M/G\cong \mathfrak{g}^{\ast}$. This is the reduced Dirac structure $D_{\Delta_{G}}^{/G} \subset Y \oplus Y^{\ast}$ on the reduced bundle $Y \to \mathfrak{g}^{\ast}$. 
\medskip

In contrast, in the case $M=TG$, the Lagrange--Dirac structure $D_{L} \subset TM \oplus T^{\ast}M $ defined in \eqref{LagDiracStr_LieGroup} is reduced to $D_{\ell}^{/G}:=D_{L}/G$, which is a subbundle of $F=(TM \oplus T^{\ast}M)/G=A \oplus A^{\ast}$ over $B=M/G\cong \mathfrak{g}$. This is the reduced Lagrange--Dirac structure $D_{\ell}^{/G} \subset A \oplus A^{\ast}$ on the reduced bundle $A \to \mathfrak{g}$.

\subsection{Dirac morphisms in the Lagrange--Dirac reduction}
In this subsection, we illustrate the Lagrange--Dirac reduction in terms of Dirac morphisms. 
%
\paragraph{Dirac morphisms relating induced and Lagrange--Dirac structures.}
The Lagrange--Dirac structure $D_{L}$ on $TG$ is obtained from the induced Dirac structure $D_{\Delta_G}$ on $T^{\ast}G$ through the backward Dirac map $\mathcal{B}(T\mathbb{F}L):\mathrm{Dir}\left(TT^{\ast}G\right) \to \mathrm{Dir}\left(TTG\right)$ associated with the Legendre transformation $\mathbb{F}L: TG \to T^{\ast}G$ as
\begin{equation*}
%
D_{L}=\mathcal{B}(T\mathbb{F}L)\big(D_{\Delta_{G}}\big).
\end{equation*}

We obtain the trivialized expression via the backward Dirac map $\mathcal{B}(T\mathbb{F}\overline{L}):  \mathrm{Dir}\left(T(G\times \mathfrak{g}^{\ast})\right) \to \mathrm{Dir}\left(T(G\times \mathfrak{g})\right)$ associated with the trivialized Legendre transformation $\mathbb{F}\overline{L}: G \times \mathfrak{g} \to G \times \mathfrak{g}^{\ast}$ as
\[
%
\overline{D}_{\ell}=\mathcal{B}(T\mathbb{F}\overline{L})\big(\overline{D}_{\Delta_{G}}\big).
\]

We also have relations between $D_{\Delta_{G}}$ and $\overline{D}_{\Delta_{G}}$
through the forward Dirac map $\mathcal{F}\left(T\lambda_{T^{\ast}G}\right): \mathrm{Dir}\left(TT^{\ast}G\right) \allowbreak \to \mathrm{Dir}\left(T(G\times \mathfrak{g}^{\ast})\right)$,
and between $D_{L}$ and $\overline{D}_{\ell}$
through $\mathcal{F}\left(T\lambda_{TG}\right): \mathrm{Dir}\left(TTG \right) \to \mathrm{Dir}\left(T(G\times \mathfrak{g})\right)$. Therefore we get the following commutative diagram.

\begin{figure}[h]
\centering
\stackinset{l}{18ex}{b}{7ex}{%
\scalebox{.8}
{%
} 
}{%
\begin{tikzcd}[row sep = 15ex, column sep = 12ex, ampersand replacement=\&]
\mathrm{Dir}\left(TT^{\ast}G \right) 
    \arrow{r}{\mathcal{B}(T\mathbb{F}L)}
    \arrow[swap]{d}{\mathcal{F}(T\lambda_{T^{\ast}G})}
\& \mathrm{Dir}\left(TTG \right)
    \arrow{d}{\mathcal{F}(T\lambda_{TG})}  \\
\mathrm{Dir}\left(T(G \times \mathfrak{g}^{\ast})\right)
    \arrow[swap]{r}{\mathcal{B}(T\mathbb{F}\overline{L})}
\& \mathrm{Dir}\left(T(G \times \mathfrak{g})\right)
\end{tikzcd}
}
\caption{Dirac maps between $D_{\Delta_{G}}$ and $D_{L}$.}
\end{figure}

This diagram shows that the Legendre transformation and trivialization commute at the level of Dirac structures through Dirac morphisms.

\paragraph{Link to the Lie--Dirac bundle reduction via Dirac morphisms.}
Here we consider the relation between the Dirac structures on $T^{\ast}G$ and $TG$ together with their reductions in the context of the Dirac maps associated with the Legendre transformation.

\begin{itemize}
\item[(i)]
For the induced Dirac structure $D_{\Delta_{G}}$ on $T^{\ast}G$, as discussed previously, taking the quotient of $D_{\Delta_{G}} \subset TT^{\ast}G \oplus T^{\ast}T^{\ast}G$ by $G$ leads to the quotient bundle
\[
D_{\Delta_{G}}/G \subset (TT^{\ast}G \oplus T^{\ast}T^{\ast}G)/G \cong  (TT^{\ast}G/G) \oplus (T^{\ast}T^{\ast}G/G),
\]
which defines a reduced Dirac structure over the quotient bundle $TT^{\ast}G/G \to T^{\ast}G/G$. On the other hand, through the trivialization $T^{\ast}G \cong G \times \mathfrak{g}^{\ast}$, the induced Dirac structure $D_{\Delta_{G}} \subset TT^{\ast}G \oplus T^{\ast}T^{\ast}G$ is trivialized 
as $\overline{D}_{\Delta_{G}} \subset T(G \times \mathfrak{g}^{\ast}) \oplus T^{\ast}(G \times \mathfrak{g}^{\ast})$, and we define the reduced Dirac structure by the quotient bundle
\[
D_{\Delta_{G}}^{/G}:=\overline{D}_{\Delta_{G}}/G \subset \big(T(G \times \mathfrak{g}^{\ast})/G\big) \oplus \big(T^{\ast}(G \times \mathfrak{g}^{\ast})/G\big) \cong Y\oplus Y^{\ast},
\]
which defines the Lie--Dirac structure on the vector bundle $Y=\mathfrak{g}^{\ast} \times (\mathfrak{g} \times \mathfrak{g}^{\ast}) \to \mathfrak{g}^{\ast}$.

\quad This reduction procedure can be understood in the context of the Dirac maps as follows.

We introduce the projection $\pi^{/G}_{T(G \times \mathfrak{g}^{\ast})}: T(G \times \mathfrak{g}^{\ast}) \to T(G \times \mathfrak{g}^{\ast})/G \cong Y = \mathfrak{g}^\ast \times (\mathfrak{g} \times \mathfrak{g}^{\ast})$,
which is  a vector bundle morphism covering the quotient map $G \times \mathfrak{g}^* \to \mathfrak{g}^*$, given for each $(g,\mu) \in G \times \mathfrak{g}^{\ast}$ by 
\[
\pi^{/G}_{T_{(g,\mu)}(G \times \mathfrak{g}^{\ast})}: T_{(g,\mu)}(G \times \mathfrak{g}^{\ast}) \to (T(G \times \mathfrak{g}^{\ast})/G)_{\mu} \cong  \mathfrak{g} \times \mathfrak{g}^{\ast},
\]
and is a fiberwise isomorphism. Then, by using the forward Dirac map
\[
\mathcal{F}\pi^{/G}_{T(G \times \mathfrak{g}^{\ast})}: \mathrm{Dir}(T(G \times \mathfrak{g}^{\ast})) \to \mathrm{Dir}(Y),
\] 
the reduced Dirac structure $D_{\Delta_{G}}^{/G}\subset Y \oplus Y^{\ast}$ can be obtained by, for each $(g,\mu) \in G \times \mathfrak{g}^{\ast}$,
\[
D_{\Delta_{G}}^{/G}(\mu):=(\overline{D}_{\Delta_{G}}/G)(\mu) =\mathcal{F}\pi^{/G}_{T_{(g,\mu)}(G \times \mathfrak{g}^{\ast})}(\overline{D}_{\Delta_{G}})(g,\mu),
\]
and it follows that
\[
D_{\Delta_{G}}^{/G}= \mathcal{F}\pi^{/G}_{T(G \times \mathfrak{g}^{\ast})}(\overline{D}_{\Delta_{G}}).
\]
In the above, since the quotient projection $\pi^{/G}_{T(G \times \mathfrak{g}^{\ast})}$ is a fiberwise isomorphism, the associated forward Dirac map $\mathcal{F}\pi^{/G}_{T(G \times \mathfrak{g}^{\ast})}$ is well defined on each fiber.

\item[(ii)]
Similarly, for a Lagrange--Dirac structure $D_{L}$ over $TG$, taking quotients $D_{L} \subset TTG \oplus T^{\ast}TG$ by $G$ reads
\[
D_{L}/G \subset (TTG \oplus T^{\ast}TG)/G \cong  (TTG/G) \oplus (T^{\ast}TG/G),
\]
which defines a reduced Lagrange--Dirac structure on the bundle $TTG/G \to TG/G$. On the other hand, through the trivialization $TG \cong G \times \mathfrak{g}$, the Lagrange--Dirac structure $D_{L} \subset TTG\oplus T^{\ast}TG$ is trivialized as 
\[
\overline{D}_{\ell} \subset T(G \times \mathfrak{g}) \oplus T^{\ast}(G \times \mathfrak{g}).
\]
Then we define the reduced Lagrange--Dirac structure by the quotient bundle
\[
D_{\ell}^{/G} := \overline{D}_{\ell}/G \subset (T(G \times \mathfrak{g})/G) \oplus (T^{\ast}(G \times \mathfrak{g})/G) \cong A \oplus A^{\ast},
\]
which is the trivialized expression of the \textit{reduced Lagrange--Dirac structure} on $A=\mathfrak{g} \times (\mathfrak{g} \times \mathfrak{g})\to \mathfrak{g}$.

For fixed $\eta \in \mathfrak{g}$, we get the reduced Lagrange--Dirac structure
\[
D_{\ell}^{/G}(\eta) \subset (\mathfrak{g} \times \mathfrak{g}) \oplus (\mathfrak{g}^{\ast} \times \mathfrak{g}^{\ast}).
\]

\quad This reduction procedure can also be understood in the context of the Dirac maps as follows. We introduce the projection 
$
\pi_{T(G \times \mathfrak{g})}^{/G}: T(G \times \mathfrak{g}) \to T(G \times \mathfrak{g})/G \cong A=\mathfrak{g} \times (\mathfrak{g} \times \mathfrak{g}),
$
which is  a vector bundle morphism covering the quotient map $G \times \mathfrak{g} \to \mathfrak{g}$, given for each $(g, \eta) \in G \times \mathfrak{g}$ by
\[
\pi_{T_{(g,\eta)}(G \times \mathfrak{g})}^{/G}: T_{(g,\eta)}(G \times \mathfrak{g}) \to (T(G \times \mathfrak{g})/G)_{\eta}  \cong \mathfrak{g} \times \mathfrak{g},
\]
and is a fiberwise isomorphism. Then, by using the forward Dirac map
\[
\mathcal{F}\pi_{T(G \times \mathfrak{g})}^{/G}: \mathrm{Dir}(T(G \times \mathfrak{g})) \to \mathrm{Dir}(A),
\] 
the reduced Lagrange--Dirac structure $D_{\ell}^{/G} \subset A \oplus A^{\ast}$ can be obtained by, for each $(g,\eta) \in G \times \mathfrak{g}$,
\[
D_{\ell}^{/G}(\eta):=(\overline{D}_{\ell}/G)(\eta) = \mathcal{F}\pi_{T_{(g,\eta)}(G \times \mathfrak{g})}^{/G}(\overline{D}_{\ell})(g,\eta),
\]
and it follows that
\[
D_{\ell}^{/G}= 
\mathcal{F}\pi_{T(G \times \mathfrak{g})}^{/G}
(\overline{D}_{\ell}).
\]
In the above, since the quotient projection $\pi^{/G}_{T(G \times \mathfrak{g})}$ is a fiberwise isomorphism, the associated forward Dirac map $\mathcal{F}\pi^{/G}_{T(G \times \mathfrak{g})}$ is well defined on each fiber.

\item[(iii)] Regarding the relations between the Dirac structures $D_{\Delta_{G}}$ on $G \times \mathfrak{g}^{\ast}$ and $D_{L}$ on $G \times \mathfrak{g}$, we consider the backward Dirac map associated with the trivialized Legendre transformation $\mathbb{F}\overline{L}: G \times \mathfrak{g} \to G \times \mathfrak{g}^{\ast}$ as
\[
\mathcal{B}(T\mathbb{F}\overline{L}): \mathrm{Dir}\left(T(G \times \mathfrak{g}^{\ast})\right)  \to \mathrm{Dir}\left(T(G \times \mathfrak{g})\right).
\] 
Since the tangent map $T\mathbb{F}\overline{L}: T(G \times \mathfrak{g})  \to T(G \times \mathfrak{g}^{\ast})$ is $G$-equivariant, 
it induces a well defined quotient map
\[
T\mathbb{F}\overline{L}/G : A \to Y,
\]
where the following identifications hold:
\[
T(G \times \mathfrak{g})/G \cong  A=\mathfrak{g} \times (\mathfrak{g} \times \mathfrak{g}) \;\; \text{and} \;\; T(G \times \mathfrak{g}^{\ast})/G \cong Y=\mathfrak{g}^{\ast} \times (\mathfrak{g} \times \mathfrak{g}^{\ast}).
\]
The induced quotient map gives rise to the backward Dirac map, and the relation between the reduced Dirac structures $D_{\Delta_{G}}^{/G} \subset Y\oplus Y^{\ast}$ and $D_{\ell}^{/G} \subset A \oplus A^{\ast}$ is given by
\[
\mathcal{B}(T\mathbb{F}\overline{L}/G): \mathrm{Dir}(Y) \to \mathrm{Dir}(A).
\]

Thus, we have
\[
D_{\ell}^{/G}=\mathcal{B}(T\mathbb{F}\overline{L}/G)(D_{\Delta_{G}}^{/G}).
\]
\end{itemize}

We can summarize the relation between the Lie--Dirac reduction and the Lagrange--Dirac reduction in the following proposition.

\begin{proposition}\label{CommDiagDiracReduction}
Assume that the Legendre transformation $\mathbb{F}L: TG \to T^{\ast}G$ is $G$-equivariant and the following diagram commutes in the sense of Dirac maps:
\[
\stackinset{l}{18ex}{b}{7ex}{%
\scalebox{.8}
{%
} 
}{%
\begin{tikzcd}[row sep = 15ex, column sep = 12ex, ampersand replacement=\&]
\mathrm{Dir}\left(T(G \times \mathfrak{g}^{\ast})\right)
    \arrow{r}{\mathcal{B}(T\mathbb{F}\overline{L})}
    \arrow[swap]{d}{\mathcal{F}\pi^{/G}_{T(G \times \mathfrak{g}^{\ast})}}
\& \mathrm{Dir}\left(T(G \times \mathfrak{g})\right)
    \arrow{d}{\mathcal{F}\pi^{/G}_{T(G \times \mathfrak{g})}}  \\
\mathrm{Dir}\left(Y\right)
    \arrow[swap]{r}{\mathcal{B}(T\mathbb{F}\overline{L}/G)}
\& \mathrm{Dir}\left(A\right)
\end{tikzcd}
}
\]
\end{proposition}

\begin{proof}
The commutativity follows from the $G$-equivariance of the Legendre transformation, which ensures that $T\mathbb{F}\overline{L}$ descends to the quotient, together with the fact that the quotient maps associated with the forward Dirac maps are fiberwise isomorphisms.
\end{proof}

This diagram shows that the reduction and the Legendre transformation commute at the level of Dirac structures via forward and backward Dirac maps.
Importantly, the above construction remains valid even for degenerate Lagrangians, since the pullback Dirac structure defined by the backward Dirac map
does not require the Legendre transformation to be invertible.

\subsection{Reduction of Lagrange--Dirac dynamical systems on Lie groups}
In this subsection, we develop the \textit{reduction of Lagrange--Dirac dynamical systems} on the tangent bundle $TG$ of a Lie group $G$ by using the reduced Lagrange--Dirac structure $D_{\ell}^{/G} \subset A\oplus A^\ast$ on the bundle $A=\mathfrak{g} \times V$. This reduction theory naturally yields an implicit reduced dynamics, referred to as \textit{Euler--Poincar\'e--Dirac equations},  for a general mechanical systems including degenerate Lagrangians, nonholonomic constraints as well as symmetry groups. In the particular case in which a given Lagrangian is hyperregular, the reduced Lagrange--Dirac dynamics recovers the \textit{Euler--Poincar\'e--Suslov equations}. 
\paragraph{Lagrange--Dirac dynamical systems on Lie groups $G$.}  
Let $L: TG \to \mathbb{R}$ be a left invariant Lagrangian, possibly degenerate. Let $\Delta_{G} \subset TG$ be a left invariant constraint distribution on $G$. 

Let $X: TG \to TTG$ be a section of the tangent bundle $TTG \to TG$, and assume that $X$ is left invariant, namely, for all $h \in G$,
\[
T_{(g,v)}\Phi_h\bigl(X(g,v)\bigr)
=
X(hg,T_gL_h\cdot v),
\]
or equivalently,
\[
h\cdot X(g,v)=X(hg,hv),
\]
where $\Phi_h(g,v):=(hg,hv)$ denotes the left action of $G$ on $TG$, and $T_{(g,v)}\Phi_h$ its tangent map.
\medskip

Let $D_{L}$ be the Lagrange--Dirac structure on $TG$, given in equation \eqref{LagDiracStr_LieGroup}. Using local coordinates $(g,v)$ for $TG$, $X$ induces the tangent vector at each point $(g,v) \in TG$ as
\begin{equation*}
X(g,v)= (g,v,\dot{g}, \dot{v}).
\end{equation*}

Define the energy $E_{L}: TG \to \mathbb{R}$ by 
\[
E_{L}(g,v) = \mathbb{F}L(g,v) \cdot v -L(g,v) = \mathbf{D}_{2}L(g,v)\cdot v-L(g,v).
\]
Then the differential of $E_{L}$ is locally given, for $(\delta g, \delta v) \in T_{(g,v)}TG$, by
\[
\begin{split}
\mathbf{d}E_{L}(g,v)\cdot(\delta g, \delta v)
&=\mathbf{D}_{1}E_{L}(g,v)\cdot \delta g+\mathbf{D}_{2}E_{L}(g,v)\cdot \delta v,
\end{split}
\]
where
$\mathbf{D}_{1}E_{L}(g,v)=\mathbf{D}_{1}\mathbf{D}_{2}L(g,v)\cdot v - \mathbf{D}_{1}L(g,v)$,  and $\mathbf{D}_{2}E_{L}(g,v)= \mathbf{D}_{2}\mathbf{D}_{2} L(g,v)\cdot v$.
\medskip

Recall from Proposition \ref{Prop: IntLDA_Eqn} that the {\it Lagrange--Dirac dynamical system} $(E_{L}, D_{L})$ satisfies the condition which is given, for each $(g,v) \in TG$, by

\begin{equation}\label{LagDiracDynSys_G}
(X(g,v),\mathbf{d}E_{L}(g,v)) \in D_{L}(g,v).
\end{equation}
In view of \eqref{LagrangeDiracCond}, the Lagrange--Dirac condition \eqref{LagDiracDynSys_G} yields an \textit{Lagrange--Dirac dynamics over a Lie group $G$}:
\begin{equation*}
\begin{cases}
\displaystyle\vspace{0.2cm} 
 \frac{d}{dt} \mathbf{D}_2 L(g,v)- \mathbf{D}_{1}L(g,v) - \mathbf{D}_{1}\mathbf{D}_{2}L(g,v)\cdot \left(\dot{g}-v \right) \in \Delta_{G}^\circ(g),\\[2mm]
\mathbf{D}_{2}\mathbf{D}_{2} L(g,v)\cdot (\dot{g}-v)=0, \\[2mm]
\dot{g} \in \Delta_{G}(g).
\end{cases}
\end{equation*}

When $L$ is hyperregular, namely, $\mathbf{D}_{2} \mathbf{D}_{2}L$ is globally nondegenerate, then we get the second-order condition $\dot{g}=v$ and hence we recover the \textit{first-order formulation of the Lagrange--d'Alembert equations $TG$}:
\begin{equation*}
\displaystyle \frac{d}{dt}\mathbf{D}_2 L(g,v) -\mathbf{D}_{1}L(g,v) \in \Delta_{G}^{\circ}(g),\qquad \dot{g} =v,\qquad \dot{g} \in \Delta_{G}(g).
\end{equation*}

\paragraph{Reduction of the differential of Lagrangians.}
Since the Lagrangian $L:TG \to \mathbb{R}$ is left invariant, we have
\[
L(hg, T_{g}L_{h} \cdot v)=L(g,v)
\]
for all $g,h \in G$ and $v \in T_{g}G$.  The energy $E_{L}(g,v)=\mathbb{F}L(g,v) \cdot v -L(g,v)$ is also left invariant, as the Legendre transformation $\mathbb{F}L: TG \to T^{\ast}G$ is equivariant as
\[ 
\mathbb{F}L(hg, T_{hg}L_{h}(v))=T^{\ast}_{hg}L_{h^{-1}}(\mathbb{F}L(g,v)).
\]
Equivalently,  it can be written as
\[
\mathbb{F}L(hg, hv)=h\cdot \mathbb{F}L(g,v).
\]
Hence, it follows that the differential of $E_L$ is the map
\[
\mathbf{d}E_L: TG \to T^{\ast}TG,
\]
written in coordinates as
\[
\mathbf{d}E_L(g,v)=\big(g,v,\mathbf{D}_{1} E_{L}(g,v),\mathbf{D}_{2} E_{L}(g,v) \big).
\]
Since the map $\mathbf{d}E_{L}:TG \to T^{\ast}TG$ is equivariant, taking the quotient by $G$ induces to the map 
\[
\mathbf{d}^{/G}E_{L} := \mathbf{d}E_{L}/G: TG/G \rightarrow (T ^{\ast} TG)/G,
\]
where $\mathbf{d}^{/G}$ denotes the induced differential map on the quotient by the $G$-action. 
\paragraph{Left Trivialized Expressions.}  
Recall the left trivializing diffeomorphism 
\[
\lambda_{TG}:TG \to G \times \mathfrak{g}: v_{g} \mapsto (g, T_{g}L_{g^{-1}} v), 
\]
and one has the left invariant energy $\overline{E}_{L}=E_{L} \circ \lambda_{TG}^{-1}$ induced on $G \times \mathfrak{g}$ as
\[
\overline{E}_{L}(g,\eta)
 = \overline{\mathbb{F}L}(g,\eta) \cdot \eta  - \bar{L}(g,\eta)= \mathbf{D}_{2}\overline{L}(g,\eta)\cdot \eta-\overline{L}(g,\eta),
\]
where $\eta=T_{g}L_{g^{-1}} v$
and its differential may be represented by
the map
\[
\overline{\mathbf{d}E}_{L} : G \times \mathfrak{g} \to G \times A^{\ast}, 
\]
where we recall $A^{\ast}= \mathfrak{g} \times V^\ast$. Then, it is expressed in coordinates by
\begin{eqnarray*}\label{d_triv_energy}
\overline{\mathbf{d}E}_{L}(g,\eta)=\big(g,\eta,T_{e}^{\ast}L_{g}\left( \mathbf{D}_{1} \overline{E}_{L} \right),\, \mathbf{D}_{2} \overline{E}_{L}  \big).
\end{eqnarray*}
Since $\overline{E}_{L}$ is $G$--invariant, define the reduced energy $\mathcal{E}_{\ell}(\eta)$ on $\mathfrak{g}$ by
\[
\mathcal{E}_{\ell}(\eta):=\overline{E}_{L}(e,\eta),
\]
which is given by
\[
\mathcal{E}_{\ell}(\eta)=\mathbf{D}\ell(\eta) \cdot \eta - \ell(\eta).
\]
\begin{definition}
Since the maps $\lambda_{TG}:TG \to G \times \mathfrak{g}$ and $\overline{\mathbf{d}E}_{L}: G \times \mathfrak{g} \to  G \times A^{\ast}$ are equivariant,
taking the quotient of the map induces the map 
\[
\mathbf{d}^{/G}\mathcal{E}_{\ell} := \overline{\mathbf{d}E_{L}}/G : \mathfrak{g} \rightarrow A^{\ast},
\]
which is locally given by
\begin{equation}\label{Loc_dReducedEnergy}
\mathbf{d}^{/G}\mathcal{E}_{\ell}(\eta)=\big(\eta,0,\mathbf{D}\mathcal{E}_{\ell}(\eta) \big).
\end{equation}
Here, the vanishing of the second component follows from the $G$-invariance of $\overline{E}_L$, and 
the third component $\mathbf{D}\mathcal{E}_{\ell}(\eta) \in \mathfrak{g}^{\ast}$ denotes the functional derivative given by
\[
\mathbf{D}\mathcal{E}_{\ell}(\eta) \cdot \delta{\eta} = \big(\mathbf{D}^{2}\ell(\eta) \cdot \eta \big) \cdot \delta{\eta},
\]
for all $\delta{\eta} \in \mathfrak{g}$.
\end{definition}
\paragraph{Reduction of vector fields.} Recall the section $X:TG  \to TTG$ is trivialized by $\lambda_{TG}:TG \to G \times \mathfrak{g}$ as
\[
\overline{X}:=X \circ \lambda_{TG}^{-1}: G \times \mathfrak{g} \to T(G \times \mathfrak{g})  \cong G \times A,
\]
where $A=\mathfrak{g} \times V$. Since $X$ is left invariant, it follows that
\[
h \cdot \overline{X}(g,\eta)=\overline{X}(hg,\eta)\; \textrm{for all}\; h \in G,
\]
and hence
\begin{equation*}\label{VectorFieldX}
\overline{X}(e,\eta)=(e,\eta,\xi,\dot{\eta}) \in G \times A.
\end{equation*}
Here we note that $\eta=T_{g} L_{g^{-1}}v$ and $\xi=T_{g} L_{g^{-1}}\dot{g}$, where $\dot{g}$ denotes the $T_{g}G$-component of $X(g,v)=(g,v,\dot{g},\dot{v}) \in TTG$, and $\dot{\eta}$ denotes the time derivative of $\eta (t)$.
\medskip

Thus, it follows from the $G$-invariance that the components $\xi$ and $\dot{\eta}$ 
depend only on $\eta$, and hence
\[
\overline{X}(e,\eta)=(e,\eta,\xi(\eta),\dot{\eta}(\eta)).
\]

\begin{definition}
Since $\lambda_{TG}$ is equivariant and $\overline{X}$ is left $G$-invariant, we define a reduced map on the quotient
\[
X^{/G}:=\overline{X}^{/G}: \mathfrak{g} \to A=\mathfrak{g} \times (\mathfrak{g} \times \mathfrak{g}),
\]
which defines a section of the trivial bundle $A \to \mathfrak{g}$, and is explicitly given by
\begin{equation}\label{reduced_VectorField}
X^{/G}(\eta)=(\eta, \xi(\eta), \dot{\eta}(\eta)).
\end{equation}
\end{definition}

\begin{remark}
The reduced map $X^{/G}$ 
is not a section of the tangent bundle of the Lie algebra, i.e., $T\mathfrak{g} \to \mathfrak{g}$, but is a section of the trivial bundle $A \to \mathfrak{g}$ induced from the $G$-invariant section $X$ of $TTG \to TG$.
\end{remark}

\paragraph{Euler--Poincar\'e--Dirac systems.} In this paragraph, we illustrate the reduction procedure for the Lagrange--Dirac dynamical systems over a Lie group $G$, and we show that the Euler--Poincar\'e--Dirac equations are obtained by using the reduced Lagrange--Dirac structure $D_{\ell}^{/G}$ on $A=\mathfrak{g} \times (\mathfrak{g} \times \mathfrak{g})$. Specifically, for the regular case, the Euler--Poincar\'e--Suslov equations are recovered. Moreover, for the unconstrained case $\Delta_{G}=TG$, the Euler--Poincar\'e equations are naturally obtained through the reduced Lagrange--Dirac structure.

\begin{definition}
Let $(E_{L}, D_{L})$ be a Lagrange--Dirac dynamical system as given in \eqref{LagDiracDynSys_G}, and let $D_{\ell}^{/G}$ be the reduced Lagrange--Dirac structure on $A=\mathfrak{g} \times (\mathfrak{g} \times \mathfrak{g})$, whose fiber at $\eta \in \mathfrak{g}$ is given by equation \eqref{reduced_LagDirac}, or  equivalently \eqref{LocReducedLagDiracAlt}. The \textit{reduction of a Lagrange--Dirac dynamical system} $(E_{L}, D_{L})$ is given by the couple $(\mathcal{E}_{\ell}, D_{\ell}^{/G})$ that satisfies the condition given, for each $\eta \in \mathfrak{g}$, by
\begin{equation}\label{condition_reducedLagDiracSys}
\big(X^{/G}(\eta),\mathbf{d}^{/G}\mathcal{E}_{\ell}(\eta) \big) \in D^{/G}_{\ell}(\eta).
\end{equation}
\end{definition}

\begin{definition}
The curve $\eta(t) \in \mathfrak{g}$ satisfying the condition \eqref{condition_reducedLagDiracSys} is \textit{a solution curve} of the reduced Lagrange--Dirac dynamical systems $(\mathcal{E}_{\ell}, D_{\ell}^{/G})$.
\end{definition}
\begin{proposition}\label{Prop:RedLagDiracDynSys}
Let $(\mathcal{E}_{\ell}, D_{\ell}^{/G})$ be the reduction of a Lagrange--Dirac dynamical system on $TG$. Let $\eta (t) \in \mathfrak{g}$ be a solution curve of the reduced Lagrange--Dirac dynamical systems $(\mathcal{E}_{\ell}, D_{\ell}^{/G})$. Then the curve $\eta (t)$ satisfies  \textbf{Euler--Poincar\'e--Dirac equations} on $\mathfrak{g} \times \mathfrak{g}$:
\begin{equation}\label{RedLagDADirac}
\begin{cases}
\displaystyle \frac{d}{dt}\mathbf{D} \ell (\eta) -\mathrm{ad}^{\ast}_{\xi}\,\mathbf{D} \ell (\eta)\in (\mathfrak{g}^{\Delta})^{\circ} ,\\[4mm] 
\displaystyle \mathbf{D}^{2} \ell (\eta)\cdot  (\xi-\eta )=0, \\[3mm] 
\displaystyle \xi \in \mathfrak{g}^{\Delta}.
\end{cases}
\end{equation}
In the hyperregular case, since the Hessian $\mathbf{D}^{2} \ell (\eta)$ is regular, we recover the \textit{Euler--Poincar\'e--Suslov equations} on $V$:
\begin{eqnarray}\label{EulPoiEqn}
\displaystyle \frac{d}{dt}\mathbf{D} \ell (\eta) - \mathrm{ad}_{\xi}^{\ast} \mathbf{D} \ell (\eta) \in (\mathfrak{g}^{\Delta})^{\circ}, \quad \xi = \eta, \quad \xi  \in \mathfrak{g}^{\Delta}.
\end{eqnarray}
Moreover, in the unconstrained case, we recover the Euler--Poincar\'e equations.
\end{proposition}
\begin{proof}
Substituting \eqref{Loc_dReducedEnergy} and \eqref{reduced_VectorField} into the condition \eqref{condition_reducedLagDiracSys}, it follows that, for each $\eta \in \mathfrak{g}$, 
\[
\bigl( (\eta, \xi, \dot{\eta}), \left(\eta,0,\mathbf{D}\mathcal{E}_{\ell}(\eta)\right) \bigr) \in D^{/G}_{\ell}(\eta),
\]
where $X^{/G}(\eta)=(\eta, \xi, \dot{\eta})$ and $\mathbf{d}^{/G}\mathcal{E}_{\ell}(\eta)=\big(\eta, 0,\mathbf{D}\mathcal{E}_{\ell}(\eta)\big)$. In view of equation \eqref{LocReducedLagDiracAlt}, we obtain the Lagrange--Dirac dynamics:
\begin{equation}\label{RedLagDiracEqn}
\displaystyle \frac{d}{dt}\mathbf{D} \ell (\eta)-\mathrm{ad}^{\ast}_{\xi}\,\mathbf{D}\ell(\eta) \in (\mathfrak{g}^{\Delta})^{\circ} ,\quad  \mathbf{D}^{2}\ell (\eta) \cdot (\xi-\eta)=0, \quad \xi \in \mathfrak{g}^{\Delta}.
\end{equation}
In the above, note that $\frac{d}{dt}\mathbf{D} \ell(\eta) =\mathbf{D}^{2} \ell (\eta)\cdot \dot{\eta}$ by the chain rule. The first equation denotes the reduced dynamics and the third equation indicates the nonholonomic constraint.
The second equation is the \textit{constraint due to the degeneracy of Lagrangians, expressed in the variables of the Lie algebra}.
\medskip

When the given Lagrangian is hyperregular and hence the reduced Lagrangian $\ell: \mathfrak{g}\to \mathbb{R}$ is also hyperregular, the Hessian $\mathbf{D}^{2} \ell(\eta) $ is nondegenerate and then we get $\xi = \eta$ corresponding to the second-order condition in the reduced level. We finally recover the Euler--Poincar\'e--Suslov equations \eqref{EulPoiEqn}, see the papers \cite{Bl2003, YoMa2007a}. Moreover, in the unconstrained case where $\Delta_{G}=TG$, the \textit{Euler--Poincar\'e equations on $\mathfrak{g} \times \mathfrak{g}$} are naturally obtained, see \cite{MaRa1999}.
\end{proof}
As shown in the above, the reduced framework of the Lagrange--Dirac dynamical system yields the Euler--Poincar\'e--Dirac equations on $\mathfrak{g} \times \mathfrak{g}$, 
which unifies various formulations for mechanical systems with nonholonomic constraints and constraints due to degeneracy of the reduced Lagrangians, including the Euler--Poincar\'e--Suslov equations for the regular case and the Euler--Poincar\'e equations for the unconstrained case.

\paragraph{Energy conservation of the reduced Lagrange--Dirac system.}
In the unreduced level, we already showed the energy conservation of the Lagrange--Dirac dynamical systems as in Proposition \ref{EneCon_UnreducedLDDS}. Here we show that the energy conservation also holds n the reduced level.

\begin{proposition}
Along the solution curve $\eta(t) \in \mathfrak{g}, \; t\in [0,T]$ of the reduced Lagrange--Dirac dynamical system in \eqref{condition_reducedLagDiracSys}, the reduced energy $\mathcal{E}_{\ell}(\eta)$ is conserved.
\end{proposition}
\begin{proof}
Let us check by directly computing the time derivative of the reduced energy $\mathcal{E}_{\ell}(\eta) = \mathbf{D}\ell(\eta) \cdot \eta - \ell(\eta)$ along the solution curve $\eta(t) \in \mathfrak{g}, \; t\in [0,T]$. 
\medskip

Differentiating $E_\ell(\eta(t))$ with respect to time yields
\begin{equation}
\frac{d}{dt}\mathcal{E}_{\ell} = \dot{\mathbf{D}}\ell \cdot \eta + \mathbf{D}\ell \cdot \dot{\eta} - \mathbf{D}\ell \cdot \dot{\eta} = \dot{\mathbf{D}}\ell \cdot \eta,
\end{equation}
where we set $\dot{\mathbf{D}}\ell:=\frac{d}{dt}\left( \mathbf{D}\ell\right)$.
Substituting the first equation of the reduced Lagrange--Dirac (or Euler--Poincar\'e--Dirac) system, $\dot{\mathbf{D}}\ell = \text{ad}^*_\xi \mathbf{D}\ell + \alpha$ where $\alpha \in (\mathfrak{g}^\Delta)^\circ$, yields:
\begin{equation}\label{dEelldt}
\frac{d}{dt}\mathcal{E}_{\ell} = (\text{ad}^*_\xi \mathbf{D}\ell + \alpha) \cdot \eta.
\end{equation}

We decompose $\eta$ as $\eta = \xi - (\xi - \eta)$ and then \eqref {dEelldt} is transformed into 
\begin{equation}\label{dEelldt_mismatch}
\begin{split}
\frac{d}{dt}\mathcal{E}_{\ell} &= \text{ad}^*_\xi \mathbf{D}\ell \cdot \xi + \alpha \cdot \xi - \text{ad}^*_\xi \mathbf{D}\ell \cdot (\xi - \eta) - \alpha \cdot (\xi - \eta)\\
&=  -(\text{ad}^*_\xi \mathbf{D}\ell + \alpha) \cdot (\xi - \eta),
\end{split}
\end{equation}
where $\text{ad}^*_\xi \mathbf{D}\ell \cdot \xi = \mathbf{D}\ell \cdot [\xi, \xi] = 0$, and the term $\alpha \cdot \xi = 0$ since $\xi \in \mathfrak{g}^\Delta$ and $\alpha \in (\mathfrak{g}^\Delta)^\circ$.
\medskip

Using the first equation again to replace the term $(\text{ad}^*_\xi \mathbf{D}\ell + \alpha)$ in \eqref{dEelldt_mismatch} with $\dot{\mathbf{D}}\ell$, it follows that
\begin{equation*}
\frac{d}{dt}\mathcal{E}_{\ell} = -\dot{\mathbf{D}}\ell \cdot (\xi - \eta).
\end{equation*}
Since $\dot{\mathbf{D}}\ell = \mathbf{D}^2 \ell(\eta) \cdot \dot{\eta}$, by applying the symmetry of the Hessian $\mathbf{D}^2 \ell(\eta)$, we get
\begin{equation*}
\frac{d}{dt}\mathcal{E}_{\ell} = -(\mathbf{D}^2 \ell(\eta) \cdot \dot{\eta}) \cdot (\xi - \eta) = -(\mathbf{D}^2 \ell(\eta) \cdot (\xi - \eta)) \cdot \dot{\eta}.
\end{equation*}
By the second equation of the reduced system, $\mathbf{D}^2 \ell(\eta) \cdot (\xi - \eta) = 0$, we conclude:
\begin{equation*}
\frac{d}{dt}\mathcal{E}_{\ell} = 0.
\end{equation*}
\end{proof}

\paragraph{Euler-Poincar\'e--Dirac reconstruction of dynamics.} As mentioned above, for the hyperregular case, the reduced Lagrange--Dirac system recovers the Euler-Poincar\'e-Suslov equations as $\frac{d}{dt}\mathbf{D} \ell (\eta) - \mathrm{ad}^{\ast}_{\xi}\,\mathbf{D}\ell(\eta) \in \big(\mathfrak{g}^{\Delta}\big)^{\circ}$ together with $\xi=\eta$ and $\xi \in \mathfrak{g}^{\Delta}$.
Noting that $\xi(t) =T_{g}L_{g^{-1}}\,\dot{g}(t)  \in \mathfrak{g}$ and  $\eta(t) =T_{g}L_{g^{-1}}\,v(t)  \in \mathfrak{g}$, the kinematic equation $\xi(t)=\eta(t)$ and the constraint $\xi(t) \in \mathfrak{g}^{\Delta}$, respectively,  correspond to the reduced forms of the second-order condition $\dot{g}(t) = v(t)$ and the nonholonomic constraint $\dot{g} \in \Delta_{G}(g)$ in the equations of motion of the Lagrange--Dirac system on $TG$:
\[
\displaystyle \frac{d}{dt}\mathbf{D}_{2} L  - \mathbf{D}_{1}L \in \Delta_{G}^{\circ}(g),\qquad \dot{g}=v, \qquad \dot{g} \in \Delta_{G}(g).
\]
This can be summarized in the context of Euler--Poincar\'e--Dirac reconstruction of dynamics given by the following theorem.
\begin{theorem}
Let $G$ be a Lie group. Let $L$ be a left invariant hyperregular Lagrangian on $TG$ and $\ell=L|\mathfrak{g}$ be the reduced Lagrangian. Let 
$\Delta_{G}$ be a left invariant distribution on $G$ and let $\mathfrak{g}^{\Delta}=\Delta_{G}(e)$ be the reduced constraint subspace of the Lie algebra $\mathfrak{g}=T_{e}G$. 

Suppose $g _0\in G $, $v_{0} \in T_{g_{0}}G$ and $\eta_{0}=T_{g_{0}} L_{g_{0}^{-1}}v_{0} \in \mathfrak{g}$.  
Let $\eta(t)\in \mathfrak{g}$ be a solution of the reduced Lagrange--Dirac system
\begin{eqnarray*}
\frac{d}{dt}\mathbf{D} \ell (\eta(t))- \mathrm{ad}_{\xi(t)}^{\ast} \mathbf{D} \ell (\eta(t)) \in \big(\mathfrak{g}^{\Delta}\big)^{\circ}, \qquad \xi(t)=\eta(t), \qquad 
\xi(t) \in \mathfrak{g}^{\Delta}
\end{eqnarray*}
with $\eta(0)=\eta_0$.
Then the solution $(g(t), v(t))$ of the original Lagrange--Dirac system on $TG$
\begin{equation*}
\frac{d}{dt}\mathbf{D}_{2}L(g(t),v(t))-\mathbf{D}_{1}L(g(t),v(t)) \in \Delta^{\circ}_{G}(g), \qquad    \dot{g}(t) = v(t), \qquad \dot{g}(t) \in \Delta_{G}(g(t)), 
\end{equation*} 
with $(g(0),v(0))=(g_0,v_0)$ is obtained by solving
\[
\dot{g}(t) = T_e L_{g(t)} \eta(t), \quad v(t)=T_e L_{g(t)}\eta(t).
\]
\end{theorem}
\begin{proof}
The commutativity of the Dirac morphisms established in Proposition \ref{CommDiagDiracReduction} ensures that the unreduced Lagrange--Dirac dynamics on $TG$ is reconstructed from the reduced dynamics on $A=\mathfrak{g} \times (\mathfrak{g} \times \mathfrak{g})$. Specifically, one can verify that the curve
$(g(t),v(t))$
defined by the reconstruction equations satisfies the given initial conditions
$(g_0,v_0)\in TG$,
maintains the second-order condition
$\dot g(t)=v(t)$,
and satisfies the nonholonomic constraint
$\dot g(t)\in\Delta_G(g(t))$
for all $t$.
\end{proof}
\begin{remark}[Reconstructibility of Degenerate Systems]\label{rem:recon_degsys}
It is important to emphasize that the geometric framework developed in this paper remains valid even for degenerate systems. 
As long as the reduced implicit system \eqref{RedLagDADirac} is solvable for $(\eta(t), \xi(t))$---as is typical in physical applications such as electric circuits---the original dynamics on $TG$ can be reconstructed. 
The commutativity of Dirac morphisms ensures that this reconstruction is not merely a formal calculation but is geometrically consistent with the original Lagrange--Dirac structure. 
While the determination of the final constraint manifold requires constraint algorithms in general, the practical solvability of the reduced DAEs provides a direct path to reconstructing the full space dynamics.
\end{remark}


\section{Variational principles for Lagrange--Dirac dynamics}\label{sec:variational}

In this section, we present a \textit{new variational principle on the velocity phase space that formulates the Lagrange--Dirac dynamics on $TQ$}. 
This variational principle generalizes both the Lagrange--d'Alembert principle and Hamilton's principle, allowing degenerate Lagrangians. 
In the hyperregular case, it naturally recovers the standard Lagrange--d'Alembert equations and, in the unconstrained case, the Euler--Lagrange equations.

In the case where the configuration manifold is a Lie group, namely $Q=G$, we further develop a reduced variational principle that yields the Euler--Poincar\'e--Dirac equations on $\mathfrak{g}\times\mathfrak{g}$. 
In particular, for the hyperregular case, the reduced variational principle recovers the Euler--Poincar\'e--Suslov equations, while in the unconstrained case it reduces to the Euler--Poincar\'e equations.

\subsection{The Lagrange--Dirac variational principle}
\label{SecHamVelocityPhaseSpacePrin}

In this subsection, we introduce a variational principle on the tangent bundle $TQ$, referred to as the \textit{Lagrange--Dirac variational principle}. 
This principle naturally yields a first-order formulation of the Euler--Lagrange--Dirac equations in the unconstrained case, allowing degenerate Lagrangians.

\paragraph{A generalization of Hamilton's principle.}

Hamilton's principle is formulated as the criticality condition of the action functional $\mathfrak{S}(q)$ defined on the space of curves on $Q$. 
Here, we extend the standard Hamilton principle from the configuration manifold $Q$ to the tangent bundle $TQ$ (velocity phase space). 
We refer to this extension as the \textit{Lagrange--Dirac variational principle on the velocity phase space}.

\begin{framed}
\begin{proposition}[The Lagrange--Dirac variational principle]\label{LagrangeDiracVarPrin}
Let $L: TQ\to \mathbb{R}$ be a Lagrangian, possibly degenerate. Consider the action functional 
\begin{equation}\label{ActionInt_HamVelocityPhaseSpacePrin}
\mathcal{S}(q,v)=\int_{0}^{T} \bigg\{ L(q(t),v(t)) + \mathbf{D}_{2}L(q(t),v(t))\cdot \big(\dot{q}(t)-v(t)\big) \bigg\} dt,
%
\end{equation}
for curves $(q(t), v(t)),\;t \in [0,T]$ on $TQ$, whose base curve $q(t)=\tau_{Q}(q(t),v(t))$ on $Q$ joins two distinct points 
$q_{0}=\tau_{Q}(q(0),v(0))$ and $q_{1}=\tau_{Q}(q(T),v(T))$. 
If a curve $(q(t), v(t)),\;t \in [0,T]$ is critical of the action functional, i.e.,
\[
\delta \mathcal{S}(q,v)=0
\]
for arbitrary variations $\delta q$ and $\delta v$ satisfying
$\delta q(0)=\delta q(T)=0$,
then it satisfies the \textit{Euler--Lagrange--Dirac equations}:
\begin{equation}\label{Unconst_ImpLagDiracDynamics}
\begin{cases}
\displaystyle \frac{d}{dt}\mathbf{D}_{2}L(q,v)-\mathbf{D}_{1}L(q,v) - \mathbf{D}_{1}\mathbf{D}_{2}L(q,v) \cdot \left(\dot{q}-v \right)=0,\\[3mm]
\mathbf{D}_{2}\mathbf{D}_{2}L(q,v) \cdot \left( \dot{q}-v \right)=0.
\end{cases}
\end{equation}
For the hyperregular case, the Euler--Lagrange--Dirac equations \eqref{Unconst_ImpLagDiracDynamics} recover the first-order formulation of the Euler--Lagrange equations on $TQ$:
\begin{equation*}
\frac{d}{dt} \mathbf{D}_2 L(q,v) = \mathbf{D}_1 L(q,v), \qquad \dot{q}=v.
\end{equation*} 
By eliminating $v$-variables, these equations on $TQ$ are equivalent to the usual Euler--Lagrange equations on $Q$ as in \eqref{EulerLagEqn}.
\end{proposition}
\end{framed}
\begin{proof}
By direct computations, the critical condition is 
\begin{equation*}
\begin{split}
&\delta \int_{0}^{T} \bigg\{ L(q(t),v(t)) + \mathbf{D}_{2}L(q(t),v(t)) \cdot \big(\dot{q}(t)-v(t)\big) \bigg\}dt\\
&= \int_{0}^{T} \bigg\{ \mathbf{D}_{1}L(q,v) \cdot\delta{q}  + \mathbf{D}_{2}L(q,v) \cdot \delta{v} \bigg.\\
&\hspace{2cm}\bigg. 
+ \big(\mathbf{D}_{2}\mathbf{D}_{2}L(q,v) \cdot \delta{v}+\mathbf{D}_{1}\mathbf{D}_{2}L(q,v) \cdot \delta{q}\big) \cdot \big(\dot{q}-v\big)
+ \mathbf{D}_{2}L(q,v) \cdot \big(\delta\dot{q}-\delta v\big) \bigg\}dt\\
&= \int_{0}^{T} \bigg\{ \mathbf{D}_{1}L(q,v) \cdot \delta{q}  + \big( \mathbf{D}_{2}\mathbf{D}_{2}L(q,v) \cdot \delta{v}  \big) \cdot \left(\dot{q}-v\right) \bigg.\\
&\bigg. \hspace{4cm} +\big( \mathbf{D}_{1}\mathbf{D}_{2}L(q,v) \cdot \delta{q} \big) \cdot  \big(\dot{q}-v \big) + \mathbf{D}_{2}L(q,v) \cdot \frac{d}{dt} \delta{q} \bigg\}dt\\
&= \int_{0}^{T} \bigg\{  \big(\mathbf{D}_{2}\mathbf{D}_{2}L(q,v)\left(\dot{q}-v\right)\big) \cdot \delta{v} \bigg.\\
&\hspace{1cm}\bigg. + \bigg( \mathbf{D}_{1}\mathbf{D}_{2}L(q,v) \cdot \left(\dot{q}-v \right)+\mathbf{D}_{1}L(q,v)-\frac{d}{dt}\mathbf{D}_{2}L(q,v)\bigg) \cdot \delta{q}\bigg\}dt+ \mathbf{D}_{2}L(q,v) \cdot \delta{q} \biggr{\arrowvert}_{0}^{T}=0,
\end{split}
\end{equation*}
for all $\delta{q}$ and $\delta{v}$. Then, taking the variations leads to 
\begin{equation*}
\begin{cases}
\displaystyle \frac{d}{dt}\mathbf{D}_{2}L(q,v)-\mathbf{D}_{1}L(q,v) - \mathbf{D}_{1}\mathbf{D}_{2}L(q,v) \cdot \left(\dot{q} -v\right)=0,\\[3mm]
\mathbf{D}_{2}\mathbf{D}_{2}L(q,v) \cdot  \left(\dot{q} - v\right)=0,
\end{cases}
\end{equation*}
where the boundary term vanishes because $\delta q(0)=\delta q(T)=0$.
When the fiber Hessian $\mathbf{D}_{2}\mathbf{D}_{2}L(q,v)$ is not regular, the second equation represents the constraint due to degeneracy of $L$.

In the hyperregular case, $\mathbf{D}_2\mathbf{D}_2 L(q,v)$ is nonsingular and hence we obtain $\dot{q}=v$. By 
substituting $\dot{q}=v$ into the first equation, the system recovers the first-order formulation of the Euler--Lagrange equations on $TQ$ obtained in \eqref{EulLagEqn}:
\[
\frac{d}{dt} \mathbf{D}_2 L(q,v) = \mathbf{D}_1 L(q,v), \qquad \dot{q}=v.
\]
\end{proof}
\begin{remark}
In the degenerate case, the second equation $\mathbf{D}_{2}\mathbf{D}_{2}L(q,v) \cdot  \left(\dot{q} - v\right)=0$ in \eqref{Unconst_ImpLagDiracDynamics} reveals that the second-order condition $\dot{q}=v$ is not 
enforced for all directions but is restricted to the complement of $\operatorname{ker}\,\mathbf{D}_{2}\mathbf{D}_{2}L(q,v)$. As discussed in Remark \ref{rem:recon_degsys}, the solvability of the implicit DAE system ensures the existence of a consistent trajectory $(q(t),v(t))$ on $TQ$.
\end{remark}

\paragraph{Finite dimensional cases.}
By direct computations using local coordinates $q^{i}, v^{i},\;i=1,\dots,n$ for $u=(q,v) \in TQ$, it follows that 
\[
\begin{split}
\displaystyle
\delta&\int_{0}^{T}\bigg\{ L(q,v) +\frac{\partial L}{\partial v^{j}}\big( \dot{q}^{j} -v^{j} \big) \bigg\}dt\\[2mm]
\displaystyle&=\int_{0}^{T} \bigg\{   \frac{\partial L}{\partial q^{i}} \delta{q}^{i} +  \frac{\partial L}{\partial v^{i}} \delta{v}^{i} 
+ \frac{\partial^{2} L}{\partial q^{i} \partial v^{j}} \left(\dot{q}^{j}-v^{j} \right)\delta{q}^{i} 
+ \frac{\partial^{2} L}{\partial v^{i} \partial v^{j}} \left(\dot{q}^{j}-v^{j} \right)\delta{v}^{i} 
+ \frac{\partial L}{\partial v^{j}}\big( \delta\dot{q}^{j} - \delta v^{j} \big)
\bigg\}dt\\[3mm]
\displaystyle&=\int_{0}^{T} \bigg[ \frac{\partial^{2} L}{\partial v^{i}v^{j}}\left(\dot{q}^{j}-v^{j} \right)\delta{v}^{i} + \bigg\{  \frac{\partial L}{\partial q^{i}}+ \frac{\partial^{2} L}{\partial q^{i} \partial v^{j}}\left(\dot{q}^{j}-v^{j}\right) -  \frac{d}{dt}\frac{\partial L}{\partial v^{i}}\bigg\} \delta{q}^{i} \bigg]dt=0,
\end{split}
\]
for all $\delta{q}^{i}$ and $\delta{v}^{i}$, together with the fixed endpoint conditions $\delta{q}^{i}(0)=\delta{q}^{i}(T)=0$. Taking variations $\delta{q}^{i}$ and $\delta{v}^{i}$ leads to
the Lagrange--Dirac equations:
\begin{equation*}
\begin{cases}
\displaystyle \frac{d}{dt}\frac{\partial L}{\partial v^{i}} - \frac{\partial L}{\partial q^{i}}  -  \frac{\partial^{2} L}{\partial q^{i} \partial v^{j}}\left( \dot{q}^{j} - v^{j}\right)=0,\\[5mm]
\displaystyle \frac{\partial^{2} L}{\partial v^{i} \partial v^{j}} \left(\dot{q}^{j}-v^{j}\right)=0,
\end{cases}
\end{equation*}
In the hyperregular case, these equations recover the local first-order formulation of the Euler-Lagrange equations:
\[
\displaystyle \frac{d}{dt} \frac{\partial L}{\partial v^{i}}  - \frac{\partial L}{\partial q^{i}}=0,\qquad 
\dot{q}^{i}=v^{i}. 
\]

\begin{remark}\rm
Notice that using the Legendre transform $p=\frac{\partial L}{\partial v}$, the action functional given in \eqref{ActionInt_HamVelocityPhaseSpacePrin} can be rewritten in terms of the action functional for the curves $(q(t), v(t), p(t))$ in the Pontryagin bundle $TQ \oplus T^{\ast}Q$ as
\[
\mathcal{S}(q,v, p)=\int_{0}^{T} \bigg\{ L(q(t),v(t)) +  p(t) \cdot \big(\dot{q}(t)-v(t)\big) \bigg\}dt,
\]
which is called the action functional of the \textit{Hamilton-Pontryagin principle} proposed in \cite{YoMa2006a}.
\end{remark}

\paragraph{Intrinsic expressions.}
Here, we consider the intrinsic expression of the Lagrange--Dirac variational principle given in Proposition \ref{LagrangeDiracVarPrin}.

\begin{framed}
\begin{proposition}[The intrinsic Lagrange--Dirac variational principle]
Let $u=(q,v) \in TQ$ and let $E_{L}: TQ \to \mathbb{R}$ be the energy defined by $E_{L}(u):=\mathbb{F}L(u) \cdot u - L(u)$. Denote by $\Theta_{L}$ the Lagrangian one-form on $TQ$.
\medskip

If a curve $u(t)$ on $TQ$ is critical of the action functional, i.e., 
\begin{equation}\label{VarPrinTan}
\delta\mathcal{S}(u)=\delta \int_{0}^{T} \bigg(\Theta_{L}(u)(\dot{u})-E_{L}(u) \bigg)dt=0,
\end{equation}
for all variations $\delta{u}$ of the curves $u(t)$ on $TQ$, together with the fixed endpoint conditions, i.e., $T_{u}\tau_{Q}(\delta{u})(0)=T_{u}\tau_{Q}(\delta{u})(T)=0$, then the critical curve $u(t)$ satisfies the intrinsic Euler--Lagrange--Dirac equation:
\begin{equation}\label{IntrinsicEulerLagrangeEqn}
\mathbf{i}_{\dot{u}(t)}\Omega_{L}(u)=\mathbf{d}E_{L}(u).
\end{equation} 
\end{proposition}
\end{framed}

\begin{proof}
Consider the variations of a curve $u(t),\;t \in [0,T]$ by $u_{\epsilon}(t)=u(t,\epsilon)$, where $\epsilon \in (-a,a)$ and hence  
$u_{0}(t)=u(t)$.  
The critical condition \eqref{VarPrinTan} becomes 
\[
\begin{split}
\delta \mathcal{S}(u) &= \mathbf{d}\mathcal{S}(u) \cdot \delta{u} = \left.\frac{d}{d\epsilon}\right|_{\epsilon =0}\mathcal{S}(u_{\epsilon})\\[2mm]
&=\left.\frac{d}{d\epsilon}\right|_{\epsilon =0}  \int_{0}^{T} \bigg(\Theta_{L}(u_{\epsilon}(t))(\dot{u}_{\epsilon}(t))-E_{L}(u_{\epsilon}(t)) \bigg)dt\\[2mm]
&=\int_{0}^{T} \bigg(-\mathbf{d}\Theta_{L}(u(t))(\dot{u}(t),\delta{u}(t))-\mathbf{d}E_{L}(u(t))\cdot \delta{u}(t) \bigg) dt+\Theta_{L}(u(t))\cdot \delta{u}(t)\biggr|_{0}^{T}\\[2mm]
&=\int_{0}^{T} \bigg(\mathbf{i}_{\dot{u}}\Omega_{L}(u)-\mathbf{d}E_{L}(u) \bigg) \cdot \delta{u}\,dt+\Theta_{L}(u)\cdot \delta{u}\biggr|_{0}^{T}\\[2mm]
&=0,
\end{split}
\]
for all variations $\delta{u}(t)$ of the curves $u(t) \in TQ$, together with the fixed endpoint  conditions $T_{u}\tau_{Q}(\delta{u})(0)=T_{u}\tau_{Q}(\delta{u})(T)=0$, i.e., $\delta{q}(0)=\delta{q}(T)=0$. 

Hence, by the fixed endpoint conditions, we have
\[
\Theta_L(u) \cdot \delta u \Bigr|_{0}^{T} =  \mathbb{F}L(u) \cdot T_u \tau_Q (\delta u)  \Bigr|_{0}^{T} = 0,
\]
and thus we get \eqref{IntrinsicEulerLagrangeEqn} follows.
Note that $\Omega_{L}=-\mathbf{d}\Theta_{L}$ holds.
\end{proof}
Note that \eqref{IntrinsicEulerLagrangeEqn} coincides with the intrinsic Euler--Lagrange equation in \eqref{InstrinsicLagrangeSystem}, except that the present formulation allows degenerate Lagrangians.

\begin{remark}
We provide a detailed derivation of the transition from the second to the third line in the variation, which involves integration by parts. Since $\Theta_{L}$ is a one-form on $TQ$, 
its Fr\'echet derivative at $u \in TQ$ in the direction $\delta u$ is again a 
one-form, denoted $\mathbf{D}\Theta_{L}(u) \cdot \delta u$. Hence, we have
\[
\left.\frac{d}{d\varepsilon}\right|_{\varepsilon=0}
\Theta_{L}(u_{\varepsilon})\cdot \dot u_{\varepsilon}
=
\big(\mathbf{D}\Theta_{L}(u) \cdot \delta u \big) \cdot \dot{u}
+
\Theta_{L}(u)\!\left(\left.\frac{d}{d\varepsilon}\right|_{\varepsilon=0}
\dot u_{\varepsilon}\right).
\]
Since the variation  and time derivative commute,
$\delta{\dot{u}}=\frac{d}{dt}\delta u$, we obtain
\[
\int_{0}^{T}
\Theta_{L}(u)\cdot \left(\frac{d}{dt}\delta u\right)\,dt
=
\Theta_{L}(u)\cdot\delta u\bigg|_{0}^{T}
-
\int_{0}^{T}
\frac{d}{dt}\!\left(\Theta_{L}(u)\right)\cdot\delta u\,dt,
\]
where the boundary term appears from integration by parts.

Using the chain rule,
$\tfrac{d}{dt}(\Theta_{L}(u(t)))=\mathbf{D}\Theta_{L}(u)\cdot \dot u$, we have
\[
\big(\mathbf{D}\Theta_{L}(u) \cdot \delta u\big) \cdot \dot u
-
\big(\mathbf{D}\Theta_{L}(u) \cdot \dot u \big) \cdot \delta u
=
-\,\mathbf d\Theta_{L}(u)(\dot u,\delta u).
\]
Here we used the standard identity for the exterior derivative of a one-form $\alpha \in \Lambda^{1}(M)$:
\[
\mathbf{d}\alpha(X, Y)=X[\alpha(Y)]-Y[\alpha(X)]-\alpha([X,Y]),\quad \textrm{for}\quad X, Y \in \mathfrak{X}(M).
\]
In the present setting, the Lie bracket term $\Theta_L([\dot u, \delta u])$ does not contribute, since $\dot u=\partial_t u(t,\varepsilon)$ and $\delta u=\partial_\varepsilon u(t,\varepsilon)$ arise as partial derivatives of a smooth two-parameter variation $u(t,\varepsilon)$ on $TQ$. Therefore, the corresponding directional derivatives commute, i.e., $\partial_t\partial_\varepsilon u=\partial_\varepsilon\partial_t u$. Hence, the corresponding vector fields satisfy $[\dot u,\delta u]=0$.
\medskip

Combining these expressions yields
\[
\delta\mathcal S(u)
=
\int_{0}^{T}
\Big(
-\,\mathbf d\Theta_{L}(u)(\dot u,\delta u)
-
\mathbf dE_{L}(u)\cdot\delta u
\Big)\,dt
+
\Theta_{L}(u)\cdot\delta u\Big|_{0}^{T}.
\]

Finally, using 
$\Theta_{L}(u)\cdot\delta u
=\langle\mathbb{F}L(u),T_{u}\tau_{Q}(\delta u)\rangle$,
the fixed endpoint condition 
$T_{u}\tau_{Q}(\delta u)(0)=
 T_{u}\tau_{Q}(\delta u)(T)=0$
implies that the boundary term vanishes. Hence, the critical condition
\[
\delta\mathcal{S}(u)=0
\]
leads directly to the intrinsic Euler--Lagrange--Dirac equations \eqref{IntrinsicEulerLagrangeEqn}. A similar derivation for Hamilton's phase space principle was given in \cite{CeMa1987}.
\end{remark}

\paragraph{Variational link to the Lagrangian systems.}
In the hyperregular case, the Lagrange--Dirac equation \eqref{IntrinsicEulerLagrangeEqn} 
coincides with the intrinsic Euler--Lagrange equation \eqref{InstrinsicLagrangeSystem} obtained from the Lagrangian condition \eqref{LagCond}. 
From this intrinsic form, one can recover the first-order system of local Euler--Lagrange equations \eqref{EulLagEqn}. 
Of course, this corresponds to the special case of a Lagrange--Dirac system in which the constraint distribution is trivial, $\Delta_{Q}=TQ$.

\subsection{The Lagrange--d'Alembert--Dirac variational principle}
We shall extend the Lagrange--Dirac variational principle in \S\ref{SecHamVelocityPhaseSpacePrin} to the d'Alembert-type variational principle for the case in which a constraint distribution $\Delta_{Q} \subset TQ$ exists. Specifically, we shall call this the \textit{Lagrange-d'Alembert--Dirac principle}, because it allows us to consider nonholonomic constraints in addition to the degeneracy due to Lagrangians. This principle naturally yields the Lagrange--Dirac dynamics on the tangent bundle. In the hyperregular case, we recover a first-order formulation of the Lagrange-d'Alembert equations.

\begin{framed}
\begin{proposition}[The Lagrange--d'Alembert--Dirac variational principle]
Let $L: TQ \to \mathbb{R}$ be a Lagrangian, possibly degenerate. Consider the following action functional
\begin{equation}\label{ActionInt_LagDADiracPrin}
\mathcal{S}(q,v)=\int_{0}^{T} \bigg\{ L(q(t),v(t)) + \mathbf{D}_{2}L(q(t),v(t)) \cdot \big(\dot{q}(t)-v(t)\big) \bigg\}dt.
\end{equation}
If a curve $(q(t), v(t)),\;t \in [0,T]$ on $TQ$ is critical, i.e., 
\[
\delta \mathcal{S}(q,v)=0,
\]
for all chosen variations  $\delta{q}(t) \in \Delta_{Q}(q(t))$ with the fixed endpoint conditions $\delta{q}(0)=\delta{q}(T)=0$, for all $\delta{v}$, and also subject to the kinematic constraint $\dot{q}(t)\in \Delta_{Q}(q(t))$, then the curve 
$(q(t), v(t)),\;t \in [0,T]$ satisfies the \textit{Lagrange--d'Alembert--Dirac equations} in \eqref{LagrangeDiracCond}:
\begin{equation*}
\begin{cases}
\displaystyle \frac{d}{dt}\mathbf{D}_{2}L(q,v)-\mathbf{D}_{1}L(q,v) - \mathbf{D}_{1}\mathbf{D}_{2}L(q,v) \cdot \left(\dot{q}-v \right) \in \Delta_{Q}^{\circ}(q),\\[3mm]
\displaystyle \mathbf{D}_{2}\mathbf{D}_{2}L(q,v) \cdot  \left( \dot{q}-v \right)=0,\\[2mm]
\displaystyle \dot{q} \in \Delta_{Q}(q).
\end{cases}
\end{equation*}
For the hyperregular case, the system above recovers the first-order formulation of the Lagrange--d'Alembert equations on $TQ$:
\begin{equation*}
\frac{d}{dt} \mathbf{D}_2 L(q,v) - \mathbf{D}_1 L(q,v)\in \Delta_{Q}^{\circ}(q), \qquad \dot{q}=v, \qquad \dot{q} \in \Delta_{Q}(q).
\end{equation*} 
\end{proposition}
\end{framed}
 \begin{proof}
By direct computations, it follows that the critical condition is 
\begin{equation*}
\begin{split}
&\delta \int_{0}^{T} \Big\{ L(q(t),v(t)) +  \mathbf{D}_{2}L(q(t),v(t)) \cdot \big(\dot{q}(t)-v(t)\big) \Big\}dt\\
&= \int_{0}^{T} \bigg\{ \Big(\mathbf{D}_{2}\mathbf{D}_{2}L(q,v)\cdot \left(\dot{q}-v\right)\Big) \cdot \delta{v}  \bigg.\\
&\hspace{1cm}\bigg. + \Big( \mathbf{D}_{1}\mathbf{D}_{2}L(q,v) \cdot \left(\dot{q}-v \right)+\mathbf{D}_{1}L(q,v)-\frac{d}{dt}\mathbf{D}_{2}L(q,v) \Big) \cdot \delta{q} \bigg\}dt + \mathbf{D}_{2}L(q,v) \cdot \delta{q} \biggr{\arrowvert}_{0}^{T}=0,
\end{split}
\end{equation*}
for all chosen variations  $\delta{q}(t) \in \Delta_{Q}(q(t))$ with the fixed endpoint conditions $\delta{q}(0)=\delta{q}(T)=0$ and for all $\delta{v}$. Taking the variations directly leads to the Lagrange--d'Alembert--Dirac equations. 

In the hyperregular case, we recover the second-order condition $\dot{q}=v$ and hence we obtain the first-order formulation of the Lagrange--d'Alembert equations on $TQ$. Of course, in the unconstrained case, we recover the first-order system of the Euler--Lagrange equations.
\end{proof}

\paragraph{Finite dimensional cases.} 
We suppose that the nonholonomic constraints $\Delta_{Q}\subset TQ$ are given by
\[
\Delta_{Q}(q):=\Big\{ (q, \dot{q})\in TQ \;\; \Big| \;\;\big< \omega^{r}(q), \dot{q} \big>=0 \Big\},
\]
where $\omega^{r}(q)=\omega^{r}_{i}(q)dq^{i},\; r=1,\dots,m<n$ are $m$ constraint one-forms on $Q$.

Using local coordinates $(q^{i}, v^{i}),\;i=1,\dots,n$ for $u=(q,v) \in TQ$, the critical condition of the action functional is computed under the fixed endpoints by:
\[
\begin{split}
\displaystyle
\delta&\int_{0}^{T}\bigg\{ L(q,v) +\frac{\partial L}{\partial v^{j}}\big( \dot{q}^{j} -v^{j} \big) \bigg\}dt\\[2mm]
\displaystyle&=\int_{0}^{T} \bigg[ \frac{\partial^{2} L}{\partial v^{i}v^{j}}\left(\dot{q}^{j}-v^{j} \right)\delta{v}^{i} + \bigg\{  \frac{\partial L}{\partial q^{i}}+ \frac{\partial^{2} L}{\partial q^{i} \partial v^{j}}\left(\dot{q}^{j}-v^{j}\right) -  \frac{d}{dt}\frac{\partial L}{\partial v^{i}}\bigg\} \delta{q}^{i} \bigg]dt=0,
\end{split}
\]
for all $\delta{v}^{j}$ and for all $\delta{q}^{j}$ that satisfy
\[
\omega^{r}_{i}(q)\delta q^{i}=0,\quad r=1,\dots, m<n,
\]
together with the nonholonomic constraints 
\[
\omega_{i}^{r}(q)\dot{q}^{i}=0.
\]
Introducing Lagrange multipliers $\mu_{r}$, we thus obtain the local Lagrange--d'Alembert--Dirac equations:
\begin{equation*}
\begin{cases}
\displaystyle \frac{d}{dt}\left( \frac{\partial L}{\partial v^{i}}\right)- \frac{\partial L}{\partial q^{i}} -  \left(\frac{\partial^{2} L}{\partial q^{i} \partial v^{j}}\right)\left( \dot{q}^{j} - v^{j}\right)=\mu_{r}\omega^{r}_{i}(q),\\[5mm]
\displaystyle \frac{\partial^{2} L}{\partial v^{i} \partial v^{j}} \left(\dot{q}^{j}-v^{j}\right)=0,\\[4mm]
\displaystyle \omega_{i}^{r}(q)\dot{q}^{i}=0.
\end{cases}
\end{equation*}
In the hyperregular case, we recover the first-order formulation of the local Lagrange-d'Alembert equations on $TQ$:
\begin{equation*}
\left\{
\begin{array}{l}
\displaystyle\frac{d}{dt}\frac{\partial L}{\partial v^{i}} -\frac{\partial L}{\partial {q}^{i}}=\mu_{r}\omega^{r}_{i}(q),\\[4mm]
\displaystyle\frac{dq^{i}}{dt}=v^{i},\\[4mm]
\omega^{r}_{i}(q)\dot{q}^{i}=0.
\end{array} 
\right.
\end{equation*}

\paragraph{Intrinsic expressions.}
Now we consider the intrinsic expression of the Lagrange--d'Alembert--Dirac variational principle. Recall the distribution $\Delta_{TQ}$ on $TQ$ is defined by
\[
\Delta_{TQ}=(T\tau_{Q})^{-1}(\Delta_{Q}),
\]
which is locally given by, for each $(q,v) \in TQ$,
\[
\Delta_{TQ}(q,v)=\left\{ (q,v,\delta{q},\delta{v}) \in T_{(q,v)}TQ \mid (q,\delta{q}) \in \Delta_{Q}(q)\right\}.
\]
\begin{framed}
\begin{proposition}[The intrinsic Lagrange--d'Alembert--Dirac variational principle]
Consider the action functional for curves $u(t),\;t \in [0,T],$ on $TQ$  joining to fixed endpoints $q_{0}=\tau_{Q}(u(0))$ and $q_{T}=\tau_{Q}(u(T))$ in $Q$: 
\begin{equation*}
\mathcal{S}(u)=\int_{0}^{T} \bigg(\Theta_{L}(u(t))\cdot \dot{u}(t)-E_{L}(u(t)) \bigg)dt.
\end{equation*}

If a curve $u(t)$ is critical of the action functional, i.e, 
\begin{equation*}
\begin{split}
\delta\mathcal{S}(u)=0,
\end{split}
\end{equation*}
with respect to $\delta{u} \in \Delta_{TQ}(u)$ with $T_{u}\tau_{Q}(\delta{u})(0)=T_{u}\tau_{Q} (\delta{u})(T)=0$ and also subject to the kinematic constraint $\dot{u}\in \Delta_{TQ}(u)$, then the curve $u(t)$ satisfies 
the intrinsic Lagrange--d'Alembert--Dirac equations:
\begin{equation}\label{IntrinsicLADEqn}
\mathbf{i}_{\dot{u}(t)}\Omega_{L}(u)-\mathbf{d}E_{L}(u) \in \Delta^{\circ}_{TQ}(u), \quad \dot{u}\in \Delta_{TQ}(u).
\end{equation} 
\end{proposition}
\end{framed}

\begin{proof}
The critical condition yields
\[
\begin{split}
\delta \mathcal{S}(u)
&=\left.\frac{d}{d\epsilon}\right|_{\epsilon =0}  \int_{0}^{T} \bigg(\Theta_{L}(u_{\epsilon}(t)) \cdot \dot{u}_{\epsilon}(t)-E_{L}(u_{\epsilon}(t)) \bigg) dt\\[2mm]
&=\int_{0}^{T} \bigg(\mathbf{i}_{\dot{u}(t)}\Omega_{L}(u)-\mathbf{d}E_{L}(u)\bigg) \cdot \delta{u}\, dt+\Theta_{L}(u)(t)\cdot \delta{u}_{0}(t)\biggr|_{0}^{T}\\[2mm]
&=0,
\end{split}
\]
for all variations $\delta{u}(t) \in \Delta_{TQ}(u)$, together with the fixed endpoint conditions $T_{u}\tau_{Q} (\delta{u})(0)=T_{u}\tau_{Q} (\delta{u})(T)=0$. Since $\Theta_{L}(u) \cdot \delta{u}\bigr|_{0}^{T}= \mathbb{F}L(u) \cdot T_{u}\tau_{Q}(\delta{u}) \bigr|_{0}^{T}$ vanishes, we get \eqref{IntrinsicLADEqn}.
\end{proof}
\paragraph{Energy conservation.}
Let $u(t),\, t\in [0,T]$ be an integral curve of the Lagrange--d'Alembert--Dirac equations \eqref{IntrinsicLADEqn}. Then, the energy $E_{L}$ is conserved such that
\begin{equation*}
\begin{aligned}
\frac{d}{dt}E_{L}(u(t))&=\mathbf{d} E_{L}(u(t))\cdot \dot{u}(t) \nonumber\\
&=\left(\mathbf{i}_{\dot{u}(t)}\Omega_{L}(u)-\mathbf{d}E_{L}(u)\right) \cdot \dot{u}(t)=0,
\end{aligned}
\end{equation*}
where $\mathbf{i}_{\dot{u}(t)}\Omega_{L}(u)-\mathbf{d}E_{L}(u) \in \Delta^{\circ}_{TQ}(u)$ and $\dot{u}(t) \in \Delta_{TQ}(u(t))$.
%

\subsection{The Euler--Poincar\'e--Dirac variational principle}\label{def:EulPoinDiracVarPrin}
Here, we consider the reduction of the Lagrange--d'Alembert--Dirac principle on the tangent bundle $TG$ of a Lie group $G$.
Assume $L: TG \to \mathbb{R}$ is left invariant and also assume a given constraint distribution $\Delta_{G}$ on $G$ is left invariant.
\begin{definition}
The Lagrange--d'Alembert--Dirac action functional $\mathcal{S}(g,v)$ for curves $(g(t),v(t))$ on $TG$
\begin{equation*}
\mathcal{S}(g,v)=\int_{0}^{T} \bigg\{ L(g(t),v(t)) + \mathbf{D}_{2}L(g(t),v(t))  \cdot \big (\dot{g}(t)-v(t)\big) \bigg\}dt
\end{equation*}
is reduced to the following action functional defined on curves
$(\eta(t),\xi(t))$ in the extended space $\mathfrak{g}\times\mathfrak{g}$:
\begin{equation*}
\mathcal{S}^{/G}(\eta,\xi)=
\int_{0}^{T} \bigg\{
\ell(\eta(t))
+\mathbf{D}\ell(\eta(t)) \cdot \big( \xi(t)-\eta(t) \big)
\bigg\}dt,
\end{equation*}
where $\eta=g^{-1}v$ and $\xi=g^{-1}\dot g$ is subject to the kinematic constraint $\xi\in\mathfrak{g}^\Delta$ and where $\ell(\eta)=L(e,\eta)$ is the reduced Lagrangian.
\medskip

In the above, the variations $\delta{g} \in \Delta_{G}(g) \subset TG$ and $\delta{v} \in T_{v}TG$ of the curves $g(t)$ and $v(t)$ induce the variations on the reduced space as $\zeta=g^{-1}\delta g\in \mathfrak{g}^\Delta \subset \mathfrak{g}$ and 
$\delta{\eta}=g^{-1}\delta{v} \in \mathfrak{g}$. Furthermore, $\delta{\dot{g}}$ is reduced to the variation of $\xi=g^{-1}\dot{g}$, which is given by the formula
\[
 \delta{\xi}=\dot \zeta+[\xi, \zeta]=\dot \zeta+\mathrm{ad}_{\xi}\zeta,
\]
where the fixed endpoint conditions $\delta{g}(0)=\delta{g}(T)=0$ are reduced to $\zeta (0)=\zeta (T)=0$.
\end{definition}

Now, we have the following proposition regarding the reduction of the Lagrange--d'Alembert---Dirac principle on the velocity phase space over a Lie group.

\begin{framed}
\begin{proposition}[Euler--Poincar\'e--Dirac variational principle]
Let $(\eta(t),\xi(t))$, $t \in [0, T]$, be a curve on
$\mathfrak{g} \times \mathfrak{g}$. If a curve $(\eta(t),\xi(t))$ is a critical point of the reduced action functional, i.e.,  
\[
\delta \mathcal{S}^{/G}(\eta,\xi)=
\delta\int_{0}^{T} \bigg\{ \ell(\eta(t)) + \mathbf{D}\ell(\eta(t)) \cdot \big(\xi(t)-\eta(t)\big) \bigg\}dt=0
\]
with respect to the variations of the form
\[
 \delta{\xi}=\dot \zeta+[\xi, \zeta]=\dot \zeta+\mathrm{ad}_{\xi}\zeta \in \mathfrak{g}^{\Delta},
\]
where $\zeta \in \mathfrak{g}^\Delta$ vanishes at the endpoints, i.e., $\zeta(0)=\zeta(T)=0$ and also subject to the kinematic constraints $\xi \in \mathfrak{g}^{\Delta}$, then the curve $(\eta, \xi)$ on $\mathfrak{g} \times \mathfrak{g}$ satisfies the Euler--Poincar\'e--Dirac equations \eqref{RedLagDADirac}:
\begin{equation*}
\begin{cases}
\displaystyle \frac{d}{dt} \mathbf{D}\ell(\eta)
- \mathrm{ad}^{\ast}_{\xi} \mathbf{D}\ell(\eta)
\in (\mathfrak{g}^\Delta)^\circ,\\[5mm]
\displaystyle \mathbf{D}^{2}\ell(\eta) \cdot (\xi - \eta) = 0,\\[4mm]
\displaystyle \xi  \in \mathfrak{g}^{\Delta}.
\end{cases}
\end{equation*}

In the hyperregular case, the above system recovers the \textit{Euler--Poincar\'e--Suslov equations}:
\begin{equation*}
\displaystyle
\frac{d}{dt}\mathbf{D}\ell(\eta) - \mathrm{ad}_{\xi}^{\ast} \mathbf{D}\ell(\eta) \in (\mathfrak{g}^{\Delta})^{\circ}, \quad \xi=\eta, \quad \xi  \in \mathfrak{g}^{\Delta},
\end{equation*}
which recovers the Euler--Poincar\'e equations in unconstrained case.
\end{proposition}
\end{framed}
\begin{proof}
Taking the variation of the reduced action functional, we have:
\[
\begin{split}
\delta \mathcal{S}^{/G}(\eta, \xi)
&=\delta\int_{0}^{T} \bigg\{ \ell(\eta) + \mathbf{D}\ell(\eta) \cdot \big (\xi-\eta \big) \bigg\}dt\\[2mm]
&=\int_{0}^{T} \bigg\{
\bigg(\mathbf{D}^2 \ell(\eta)\cdot (\xi-\eta)\bigg) \cdot \delta \eta
+ \mathbf{D}\ell(\eta) \cdot \delta \xi  \bigg\}dt.
\end{split}
\]

Substituting $\delta \xi = \dot{\zeta} + \mathrm{ad}_\xi \zeta$ and using integration by parts together with $\zeta(0)=\zeta(T)=0$, the critical condition is given by
\[
\int_{0}^{T} \bigg\{
\bigg(\mathbf{D}^{2}\ell(\eta) \cdot (\xi -\eta)\bigg) \cdot \delta \eta
+ \bigg( -\frac{d}{dt} \mathbf{D}\ell(\eta) + \mathrm{ad}^{\ast}_{\xi} \mathbf{D}\ell(\eta)\bigg) \cdot \zeta
\bigg\}dt=0,
\]
for all $\delta \eta$ and for all $\zeta \in \mathfrak{g}^\Delta$.

By taking variations with respect to $\delta \eta$ and for all $\zeta \in \mathfrak{g}^\Delta$, we obtain the Euler--Poincar\'e--Dirac equations \eqref{RedLagDADirac}.
\medskip

In the hyperregular case, since the Hessian $\mathbf{D}^{2}\ell(\eta)$ is nonsingular, the second equation implies the second-order condition $\xi=\eta$.
%
%
Then, we recover the conventional Euler--Poincar\'e--Suslov equations.
\end{proof}
The above reduced Lagrange--d'Alembert--Dirac principle provides the variational structure associated with the Euler--Poincar\'e--Dirac equations, given in Proposition~\ref{Prop:RedLagDiracDynSys}, which includes the Euler--Poincar\'e--Suslov equations in the hyperregular case and the Euler--Poincar\'e equations in the unconstrained case. In this sense, we shall call this reduced variational principle the \textbf{Euler--Poincar\'e--Dirac variational principle}.

\begin{definition}[Intrinsic variational derivatives]
Here we define the intrinsic expression for the variation of the reduced Lagrange--d'Alembert--Dirac action functional.

For a curve $u=(g,v)\in TG$, the tangent vectors $\dot{u}=(g,v,\dot{g},\dot{v})$ and $\delta{u}=(g,v,\delta{g},\delta{v})$ are reduced to $(\eta, \xi, \dot{\eta})=(g^{-1}v, g^{-1}\dot{g}, g^{-1}\dot{v})$ and $ (\eta, \zeta, \delta{\eta})=(g^{-1}v, g^{-1}\delta{g}, g^{-1}\delta{v})\in A=\mathfrak{g} \times (\mathfrak{g} \times \mathfrak{g})$. 

\quad Let $\mathcal{E}_{\ell}(\eta)=\mathbb{F}\ell(\eta) \cdot \eta - \ell(\eta)$ be the reduced Lagrangian on $\mathfrak{g}$. Let $\theta_{\ell}^{/G}:=\theta_{\ell}/G \in \Gamma(\Lambda^{1}(A))$ be the reduced Lagrangian one-form given in Definition \ref{Def:RedLagOneForm}, and 
let $\omega_{\ell}^{/G}:= \omega_{\ell}/G \in \Gamma\left(\Lambda^{2}(A)\right)$  the reduced Lagrangian two-form given in Definition \ref{Def:RedLagTwoForm}, where we note that there exists the relation:
\[
\omega_{\ell}^{/G}=-\mathbf{d}^{/G}\theta_\ell^{/G}.
\]

For a reduced curve $(\eta(t),\xi(t)), \,t \in[0,T]$ on $\mathfrak{g} \times \mathfrak{g}$, we consider the reduced action functional:

\begin{equation}\label{Def:RedVarPrinTan}
\mathcal{S}^{/G}(\eta, \xi):=\int_{0}^{T} \bigg(\theta^{/G}_{\ell}(\eta) \cdot (\xi, \dot{\eta})-\mathcal{E}_{\ell}(\eta) \bigg)dt.
\end{equation}

The variation of the reduced action functional is defined by
\begin{equation*}
\begin{split}
\delta \mathcal{S}^{/G}(\eta, \xi) &:= \mathbf{d}^{/G} \mathcal{S}^{/G}(\eta, \xi) \cdot (\zeta,\delta{\eta}) \\[2mm]
&=\int_{0}^{T} \bigg(-\mathbf{i}_{(\xi, \dot{\eta})}\mathbf{d}^{/G}\theta^{/G}_{\ell}(\eta)-\mathbf{d}^{/G}\mathcal{E}_{\ell}(\eta)\bigg)\cdot (\zeta,\delta{\eta}) dt
+\theta^{/G}_{\ell}(\eta)\cdot (\zeta,\delta{\eta})\biggr|_{0}^{T}\\[2mm]
&=\int_{0}^{T} \bigg(\mathbf{i}_{(\xi, \dot{\eta})}\omega^{/G}_{\ell}(\eta)-\mathbf{d}^{/G}\mathcal{E}_{\ell}(\eta) \bigg)\cdot  (\zeta,\delta{\eta}) dt +\theta^{/G}_{\ell}(\eta)\cdot (\zeta,\delta{\eta})\biggr|_{0}^{T},
\end{split}
\end{equation*}
where we use the relation $\omega_{\ell}^{/G}=-\mathbf{d}^{/G}\theta_\ell^{/G}$.
\medskip

\quad Now, imposing the endpoint fixed conditions $\zeta(0)=\zeta(T)=0$, we have   
\begin{equation*}
\begin{split}
\delta \mathcal{S}^{/G}(\eta, \xi) &=\int_{0}^{T} \bigg(\mathbf{i}_{(\xi, \dot{\eta})}\omega^{/G}_{\ell}(\eta)-\mathbf{d}^{/G}\mathcal{E}_{\ell}(\eta)\bigg) \cdot  (\zeta,\delta{\eta}) dt.
\end{split}
\end{equation*}

\end{definition}

\begin{framed}
\begin{proposition}[The intrinsic reduced Lagrange--d'Alembert--Dirac variational principle]
Let $\theta_\ell^{/G}$ be the reduced canonical one-form, and $\mathcal{E}_\ell$ be the reduced energy function as before.
Let $\Delta_{\mathfrak g}=\mathfrak{g} \times (\mathfrak{g}^{\Delta} \times \mathfrak{g}) \subset A=\mathfrak{g} \times (\mathfrak{g} \times \mathfrak{g})$ be a reduced constraint subbundle given in Definition \ref{Def:LagDiracRed}.
 
 Consider the reduced action functional $\mathcal{S}^{/G}(\eta, \xi)$ on $\mathfrak{g} \times \mathfrak{g}$, as in \eqref{Def:RedVarPrinTan}.
If a reduced curve $(\eta(t), \xi(t)),\;t \in [0,T]$ on $\mathfrak{g} \times \mathfrak{g}$ is a critical point of the reduced Lagrange--d'Alembert--Dirac functional $\mathcal{S}^{/G}(\eta, \xi)$, i.e., it satisfies
\begin{equation*}
\delta\mathcal{S}^{/G}(\eta, \xi)=\mathbf{d}^{/G} \mathcal{S}^{/G}(\eta, \xi) \cdot (\zeta,\delta{\eta})=0,
\end{equation*}
for variations $(\zeta,\delta{\eta}) \in \Delta_{\mathfrak{g}}$ with $\zeta(0)=\zeta(T)=0$, and subject to the nonholonomic constraint
\[
(\xi,\dot\eta)
\in \Delta_{\mathfrak g}(\eta),
\]
then the curve $(\eta(t), \xi(t))$ satisfies the intrinsic reduced Lagrange--d'Alembert--Dirac equations:
\begin{equation}\label{IntRedLADeq}
\mathbf{i}_{(\xi, \dot{\eta})}\omega^{/G}_{\ell}(\eta)-\mathbf{d}^{/G}\mathcal{E}_{\ell}(\eta) \in (\Delta_{\mathfrak{g}}(\eta))^{\circ},\quad (\xi,\dot{\eta}) \in \Delta_{\mathfrak{g}}(\eta).
\end{equation}
\end{proposition}
\end{framed}
\begin{proof}

A critical curve $(\eta(t), \xi(t))$ satisfies 
\begin{equation}\label{VarPrinTan_revised}
\begin{split}
\delta \mathcal{S}^{/G}(\eta, \xi) &=\mathbf{d}^{/G} \mathcal{S}^{/G}(\eta, \xi) \cdot (\zeta,\delta{\eta})\\[2mm]
&= \int_{0}^{T} \left( \mathbf{i}_{(\xi, \dot{\eta})}\omega^{/G}_{\ell}(\eta) - \mathbf{d}^{/G}\mathcal{E}_{\ell}(\eta) \right) \cdot (\zeta, \delta{\eta}) \, dt=0,
\end{split}
\end{equation}
for all variations $(\zeta,\delta{\eta}) \in \Delta_{\mathfrak{g}}=\mathfrak{g}^{\Delta} \times \mathfrak{g}$, together with the nonholonomic constraint $(\xi(t),\dot\eta(t)) \in \Delta_{\mathfrak g}$, where the boundary term $\theta^{/G}_{\ell}(\eta)\cdot (\zeta,\delta{\eta})\bigr|_{0}^{T}$ vanishes since the endpoints of the base curve $g(t)$ are fixed, implying $\zeta(0)=\zeta(T)=0$.

Then, the curve $(\eta(t), \xi(t))$ satisfies the intrinsic reduced Lagrange--d'Alembert--Dirac equations:
\[
\mathbf{i}_{(\xi, \dot{\eta})}\omega^{/G}_{\ell}(\eta)-\mathbf{d}^{/G}\mathcal{E}_{\ell}(\eta) \in (\Delta_{\mathfrak{g}}(\eta))^{\circ},\quad (\xi,\dot{\eta}) \in \Delta_{\mathfrak{g}}(\eta),
\]
where $(\Delta_{\mathfrak g}(\eta))^{\circ}$ is the annihilator of $\Delta_{\mathfrak g}(\eta)=\mathfrak{g}^{\Delta} \times \mathfrak{g}$.
\end{proof}

The reduced Lagrange--d'Alembert--Dirac equations obtained in \eqref{IntRedLADeq} admit the degenerate Lagrangian systems, which yields an implicit constrained dynamics including nonholonomic constraints and the constraints due to degeneracy of the Lagrangian on $TQ$, as shown in the following corollary.

\begin{corollary}
The intrinsic reduced Lagrange--d'Alembert--Dirac equation in \eqref{IntRedLADeq} naturally yields the Euler--Poincar\'e--Dirac equations:
\begin{equation*}
\begin{cases}
\displaystyle  \frac{d}{dt}\mathbf{D}\ell(\eta)  - \mathrm{ad}^{\ast}_{\xi}\mathbf{D}\ell(\eta) \in (\mathfrak{g}^\Delta)^{\circ},\\[2mm]
\displaystyle \mathbf{D}^{2}\ell(\eta) \cdot (\eta-\xi)=0,\\[2mm]
\displaystyle \xi \in \mathfrak{g}^{\Delta}.
\end{cases}
\end{equation*}

In the hyperregular case, we recover the Euler--Poincar\'e--Suslov equations:
\[
\frac{d}{dt}\mathbf{D}\ell(\eta)
 - \operatorname{ad}_{\xi}^{\;\ast} \mathbf{D}\ell(\eta)\in (\mathfrak{g}^{\Delta})^{\circ},  \quad \xi = \eta, \quad \xi \in \mathfrak{g}^{\Delta}.
\]
\end{corollary}

\begin{proof}
From the criticality condition \eqref{VarPrinTan_revised}, we have
\begin{equation*}
\begin{split}
\delta \mathcal{S}^{/G}(\eta, \xi) &=\int_{0}^{T} \left( \mathbf{i}_{(\xi, \dot{\eta})}\omega^{/G}_{\ell}(\eta) - \mathbf{d}^{/G}\mathcal{E}_{\ell}(\eta) \right) \cdot (\zeta, \delta{\eta}) \, dt \\[2mm]
&= \int_{0}^{T} \bigg\{ \omega^{/G}_{\ell}(\eta)\big((\xi, \dot{\eta}), (\zeta, \delta{\eta}) \big) - \mathbf{d}^{/G}\mathcal{E}_{\ell}(\eta)  \cdot \big(\zeta, \delta{\eta}\big)\bigg\}dt=0, 
\end{split}
\end{equation*}
for all variations $(\zeta,\delta{\eta}) \in \Delta_{\mathfrak{g}}=\mathfrak{g}^{\Delta} \times \mathfrak{g}$, together with the nonholonomic constraint $(\xi(t),\dot\eta(t)) \in \Delta_{\mathfrak g}$.
\medskip

The reduced Lagrangian two-form $\omega_\ell^{/G}:= \omega_\ell/G \in \Gamma\!\left(\Lambda^2(A)\right)$ is given, for $(\eta,\xi,\dot{\eta})$, $(\eta,\zeta,\delta{\eta}) \in A=\mathfrak{g} \times (\mathfrak{g} \times \mathfrak{g})$, by
\begin{equation*}
\omega_{\ell}^{/G}(\eta)((\xi,\dot{\eta}),(\zeta,\delta{\eta}))
= \left\langle \mathbf{D}^{2} \ell(\eta) \cdot \delta{\eta}, \xi \right\rangle
- \left\langle \mathbf{D}^{2} \ell(\eta) \cdot \dot{\eta}, \zeta \right\rangle
+ \left\langle \mathbf{D} \ell(\eta), [\xi,\zeta] \right\rangle,
\end{equation*}
and the reduced differential of the reduced energy is a map $\mathbf{d}^{/G}\mathcal{E}_{\ell}: \mathfrak{g} \to A^\ast = \mathfrak{g}\times (\mathfrak{g}^{\ast} \times \mathfrak{g}^{\ast})$ given by
\[
\mathbf{d}^{/G}\mathcal{E}_{\ell}(\eta) \cdot (\zeta,\delta{\eta})=\mathbf{D}\mathcal{E}_{\ell}(\eta) \cdot \delta{\eta}=\big(\mathbf{D}^{2}\ell(\eta) \cdot \eta\big) \cdot \delta{\eta}.
\]
Substituting these local expressions into the criticality condition, we have
\begin{equation*}
\left\langle \mathbf{D}^{2} \ell(\eta) \cdot \delta{\eta}, \xi \right\rangle
- \left\langle \mathbf{D}^{2} \ell(\eta) \cdot \dot{\eta}, \zeta \right\rangle
+ \left\langle \mathbf{D} \ell(\eta), [\xi,\zeta] \right\rangle
=\left\langle \mathbf{D}^{2}\ell(\eta) \cdot \eta,  \delta{\eta} \right\rangle,
\end{equation*}
and therefore
\begin{equation*}
\underbrace{\left\langle \mathbf{D}^{2} \ell(\eta) \cdot (\xi - \eta), \delta{\eta} \right\rangle}_{\text{Hessian part}} + \underbrace{\left\langle -\frac{d}{dt}\mathbf{D} \ell(\eta) + \operatorname{ad}^{\ast}_{\xi}\mathbf{D} \ell(\eta), \zeta \right\rangle}_{\text{Euler-Poincar\'e part}} = 0.
\end{equation*}
Since this holds for all variations $\zeta \in \mathfrak{g}^{\Delta}$ and for all $\delta{\eta} \in \mathfrak{g}$, the coefficients must vanish (or lie in the appropriate annihilator), yielding the Euler--Poincar\'e--Dirac equations:
\[
\frac{d}{dt}\mathbf{D} \ell(\eta)-\operatorname{ad}^{\ast}_{\xi}\mathbf{D}\ell(\eta) \in \big(\mathfrak{g}^{\Delta}\big)^{\circ}, \qquad \mathbf{D}^{2}\ell(\eta) \cdot \big(\xi-\eta \big)=0, \qquad \xi \in \mathfrak{g}^{\Delta}.
\]
For the hyperregular case, since the Hessian $\mathbf{D}^{2}\ell(\eta)$ is nondegenerate, the second equation recovers $\xi=\eta$. Thus, we obtain the Euler--Poincar\'e--Suslov equations. 
\end{proof}

\paragraph{Summary of the Euler--Poincar\'e--Dirac reduction theory.} 
Now, we summarize the Euler--Poincar\'e--Dirac reduction theory, which enables us to unify various formulations for the Lagrange--Dirac dynamical systems over a Lie group in the following theorem.

\begin{framed}
\begin{theorem}[Euler--Poincar\'e--Dirac reduction theory] \label{theorem_HVPSPrin}
Let $G$ be a Lie group and $\Delta_{G}$ be a left invariant distribution on $G$. Denote by $D_{L}$ a Lagrange--Dirac structure on $TG$. Let $L: TG \to \mathbb{R}$ be a left invariant Lagrangian, possibly degenerate, and let $E_{L}: TG \to \mathbb{R}$ be the associated energy. Let $\ell : \mathfrak{g} \to \mathbb{R}$ be the reduced Lagrangian of $L$ and $\mathcal{E}_{\ell}: \mathfrak{g} \to \mathbb{R}$ be the reduced energy.
\medskip

\begin{itemize}
\item[(1)]
In the unreduced level, the following statements are equivalent:

\begin{itemize}

\item[(i)]\textbf{Lagrange--Dirac variational principle:} A curve $(g(t),v(t)),\;t \in [0,T]$ on $TG$ is a critical point of the local Lagrange--d'Alembert --Dirac action functional:
\begin{equation*}
\mathcal{S}(g,v)=\int_{0}^{T} \bigg\{ L(g,v) + \mathbf{D}_{2}L(g,v)  \cdot \big(\dot{g}-v\big) \bigg\}dt,
\end{equation*}
namely, if it satisfies
\begin{equation*}
\delta\mathcal{S}(g,v)=0,
\end{equation*}
for all chosen variations  $\delta{g} \in \Delta_{G}(g)$ with the fixed endpoint conditions $\delta{g}(0)=\delta{g}(T)=0$, for all $\delta{v}$, and also subject to the kinematic constraint $\dot{g}\in \Delta_{G}(g)$. 

\item[(ii)] \textbf{Local equations of motion:}
A curve $(g(t), v(t)),\;t \in [0,T]$  on $TG$ satisfies the Lagrange--d'Alembert--Dirac equations: 
\begin{equation*}
\begin{cases}
\displaystyle\vspace{0.2cm} 
 \frac{d}{dt} \mathbf{D}_2 L(g,v)- \mathbf{D}_{1}L(g,v) +\mathbf{D}_{1}\mathbf{D}_{2}L(g,v)\cdot \left(v -
 \dot{g} \right) \in \Delta_{G}^\circ(g),\\[2mm]
\mathbf{D}_{2}\mathbf{D}_{2} L(g,v)\cdot (v-\dot{g})=0, \\[2mm]
\dot{g} \in \Delta_{G}(g).
\end{cases}
\end{equation*}
When $L$ is hyperregular, it satisfies the Lagrange--d'Alembert equations on $TG$:
\begin{equation*}
\displaystyle \frac{d}{dt}\mathbf{D}_2 L(g,v) -\mathbf{D}_{1}L(g,v) \in \Delta_{G}^{\circ}(g),\quad \dot{g} =v,\quad \dot{g} \in \Delta_{G}(g).
\end{equation*}
\item[(iii)] \textbf{Lagrange--Dirac dynamics:} A curve $(g(t), v(t)),\;t \in [0,T]$  on $TG$ is a solution curve of the Lagrange--Dirac dynamical system $(E_{L}, D_{L})$ that satisfies
\[
\big((\dot{g},\dot{v}), \mathbf{d}E_{L}(g, v)\big) \in D_{L}(g, v).
\]

\item[(iv)] \textbf{Intrinsic variational principle:}
A curve $u(t),\;t \in [0,T],$ on $TG$ is a critical point of the action functional 
\begin{equation*}
\begin{split}
\mathcal{S}(u)=\int_{0}^{T} \bigg\{\Theta_{L}(u)(\dot{u})-E_{L}(u) \bigg\}dt
\end{split}
\end{equation*}
namely, 
\[
\delta\mathcal{S}(u)=0,
\]
with respect to $\delta{u} \in \Delta_{TG}(u)$ with $T_{u}\tau_{G}(\delta{u})(0)=T_{u}\tau_{G} (\delta{u})(T)=0$ and also subject to the kinematic constraint $\dot{u}\in \Delta_{TG}(u)$.

\item[(v)] \textbf{Intrinsic equations of motion:}
A curve $u(t),\;t \in [0,T],$ on $TG$ satisfies 
the following equations:
\begin{equation*}
\mathbf{i}_{\dot{u}(t)}\Omega_{L}(u)-\mathbf{d}E_{L}(u) \in \Delta^{\circ}_{TG}(u), \quad \dot{u}\in \Delta_{TG}(u).
\end{equation*} 
\end{itemize}

\item[(2)] The statements in the unreduced level are all reduced to the following equivalent statements:

\begin{itemize}
\item[(i)]\textbf{Euler--Poincar\'e--Dirac variational principle:} A curve $(\eta(t), \xi(t)),\;t \in [0,T]$ on $\mathfrak{g} \times \mathfrak{g}$ is a critical curve of the Euler--Poincar\'e--Dirac action functional:
\[
\mathcal{S}^{/G}(\eta,\xi)=\int_{0}^{T} \bigg\{ \ell(\eta(t))+ 
\mathbf{D} \ell(\eta(t)) \cdot \big( \xi(t) - \eta(t) \big) \bigg\} dt,
\]
namely, if it satisfies
\begin{equation*}
\delta \mathcal{S}^{/G}(\eta,\xi)=0,
\end{equation*}
subject to the kinematic constraint  $\xi=g^{-1}\dot{g} \in \mathfrak{g}^\Delta$ and also to the variational constraint 
\[
\delta{\xi}=\dot\zeta+[\xi, \zeta], 
\]
where $\zeta=g^{-1}\delta{g} \in \mathfrak{g}^\Delta$ vanishing at the endpoints,
where $\eta=g^{-1}v$.

\item[(ii)] \textbf{Local reduced equations of motion:} A curve $(\eta(t), \xi(t)),\;t \in [0,T]$ on $\mathfrak{g} \times \mathfrak{g}$  is a solution curve of the Euler--Poincar\'e--Dirac equations:
\begin{equation*}
\begin{cases}
\displaystyle  \frac{d}{dt}\mathbf{D}\ell(\eta)  - \mathrm{ad}^{\ast}_{\xi}\mathbf{D}\ell(\eta) \in (\mathfrak{g}^\Delta)^{\circ},\\[5mm]
\displaystyle \mathbf{D}^{2}\ell(\eta) \cdot (\xi - \eta)=0,\\[4mm]
\displaystyle \xi \in \mathfrak{g}^{\Delta}.
\end{cases}
\end{equation*}

In the hyperregular case, it is a solution curve of the Euler--Poincar\'e--Suslov equations:
\begin{equation*}
\frac{d}{dt}\mathbf{D}\ell(\eta)
 - \operatorname{ad}_{\xi}^{\;\ast} \mathbf{D}\ell(\eta)\in (\mathfrak{g}^{\Delta})^{\circ},  \quad \xi = \eta, \quad \xi \in \mathfrak{g}^{\Delta}.
\end{equation*} 
Moreover, in the unconstrained case, it is a solution curve of the Euler--Poincar\'e equations.

\item[(iii)] \textbf{Reduced Lagrange--Dirac dynamics:} A curve $(\eta(t), \xi(t)),\;t \in [0,T]$ on $\mathfrak{g} \times \mathfrak{g}$  is a solution curve of the reduced Lagrange--Dirac dynamical system
$(\mathcal{E}_{\ell}, D_{\ell}^{/G})$ that satisfies, for each $\eta \in \mathfrak{g}$, the condition
\begin{equation*}
\big((\eta, \xi, \dot{\eta}),\,\mathbf{d}^{/G}\mathcal{E}_{\ell}(\eta)\big) \in D^{/G}_{\ell}(\eta).
\end{equation*}

\item[(iv)] \textbf{Intrinsic reduced variational principle:}
A curve $(\eta(t), \xi(t))\in \mathfrak{g} \times \mathfrak{g}$ is a critical point of the reduced action functional:
\begin{equation*}
\mathcal{S}^{/G}(\eta, \xi)= \int_0^T \left\{ \theta_\ell^{/G}(\eta) \cdot (\xi,\dot{\eta}) - \mathcal{E}_\ell(\eta) \right\} dt,
\end{equation*}
namely, it satisfies
\begin{equation*}
\delta \mathcal{S}^{/G}(\eta, \xi)=0,
\end{equation*}
for variations $(\zeta,\delta{\eta}) \in \Delta_{\mathfrak{g}}=\mathfrak{g}^{\Delta} \times \mathfrak{g}$, with the fixed endpoint conditions $\zeta (0)=\zeta (T)=0$, where the induced variation of $\xi$ is given by
\[
\delta\xi=\dot\zeta+[\xi,\zeta].
\]

\item[(v)] \textbf{Intrinsic reduced equations of motion:} 
A curve $(\eta(t), \xi(t))$ satisfies the intrinsic reduced Lagrange--d'Alembert--Dirac equations:
\begin{equation*}
\mathbf{i}_{(\xi,\dot{\eta})}\omega_{\ell}^{/G}(\eta)-\mathbf{d}^{/G}\mathcal{E}_{\ell}(\eta) \in (\Delta_{\mathfrak{g}}(\eta))^{\circ},\quad (\xi,\dot{\eta}) \in \Delta_{\mathfrak{g}}(\eta).
\end{equation*}
\end{itemize}
\end{itemize}
\end{theorem}
\end{framed}

This theorem clarifies that the reduced Lagrange--d'Alembert--Dirac equations are not merely a specific reduction of nonholonomic systems, but are fundamentally rooted in the reduced Dirac structure on the reduced velocity phase space $\mathfrak{g} \times \mathfrak{g}$, allowing for a unified treatment of both regular and degenerate Lagrangians, together with nonholonomic constraints.


\section{Examples}
In this section, we illustrate our theory through various examples, including charged particle motion in electromagnetic fields, the Heisenberg system and the Euler top with nonholonomic constraints, electric circuits, and ideal fluids. These examples demonstrate that a wide class of mechanical systems---including nonholonomic systems, degenerate systems, symmetry-reduced systems, and infinite-dimensional systems---can be systematically formulated within the unified geometric and variational framework of Lagrange--Dirac dynamical systems. In particular, the degeneracy of the velocity phase-space Lagrangians in plasma physics is shown to arise intrinsically from the underlying geometric structure of the Lagrange--Dirac dynamical system on $TQ$.

\subsection{Charged particle motion in electromagnetic fields}
In conjunction with plasma physics, let us consider a charged particle in an electromagnetic field. Denote by $\mathbf{x}=(x^{1},x^{2},x^{3}) \in Q=\mathbb{R}^3$ a position vector of a particle with mass $m$ and charge $e$. Let $\mathbf{E}(\mathbf{x})$ and $\mathbf{B}(\mathbf{x})$ be an electric field and a magnetic field, in which we have the relations
\[
\mathbf{E}=-\nabla \Phi, \quad \mathbf{B}=\nabla \times \mathbf{A},
\]
where $\Phi(\mathbf{x})$ is the electric potential and $\mathbf{A}(\mathbf{x})=(A_{1},A_{2},A_{3})$ is the vector potential associated with $\mathbf{B}$.
Then, we define a Lagrangian $L: TQ \to \mathbb{R}$ by, for $(\mathbf{x}, \mathbf{v}) \in TQ$, 
\begin{equation}\label{LagChPEM}
L(\mathbf{x}, \mathbf{v})=\frac{1}{2}m |\mathbf{v}|^2+\frac{e}{c}\mathbf{v}\cdot \mathbf{A}(\mathbf{x})-e\Phi(\mathbf{x}),
\end{equation}
where $c$ denotes the speed of light.
The Legendre transformation is given by
\[
\mathbb{F}L: TQ \to T^{\ast}Q;\quad (\mathbf{x}, \mathbf{v}) \mapsto \left(\mathbf{x}, \frac{\partial L}{\partial \mathbf{v}}\right)=\left(\mathbf{x}, m \mathbf{v} + \frac{e}{c}\mathbf{A}(\mathbf{x})\right)
\]
and therefore we can define the (canonical) Hamiltonian $H$ on $T^{\ast}Q$ by
\[
H(\mathbf{x}, \mathbf{p})=\frac{1}{2m}\left| \mathbf{p} - \frac{e}{c} \mathbf{A}(\mathbf{x})\right|^{2}+e\Phi(\mathbf{x}),
\]
where $\mathbf{p} = \frac{\partial L}{\partial \mathbf{v}}=m \mathbf{v} + \frac{e}{c}\mathbf{A}(\mathbf{x})$ is the conjugate momenta; see, \cite{MaRa1999,CaBr2009}.

\paragraph{Littlejohn's derivation of a phase-space Lagrangian.}
Littlejohn \cite{Lit1983} proposed a concept of \textit{phase-space Lagrangians} on $TT^{\ast}Q$ as follows:
\[
\begin{split}
\mathfrak{L}(\mathbf{x}, \mathbf{p}, \mathbf{\dot x})&:=\mathbf{p}\cdot \mathbf{\dot x}-H(\mathbf{x}, \mathbf{p})\\
&=\mathbf{p}\cdot \mathbf{\dot x}-\frac{1}{2m}\left| \mathbf{p} - \frac{e}{c} \mathbf{A}(\mathbf{x})\right|^{2}-e\Phi(\mathbf{x}).
\end{split}
\]
Further by setting $ \mathbf{v}=\frac{1}{m}\left(\mathbf{p}-\frac{e}{c}\mathbf{A}(\mathbf{x})\right)$, a \textit{velocity phase-space Lagrangian} can be defined on  $TTQ$ as
\begin{equation}\label{VPSLag}
\mathcal{L}(\mathbf{x}, \mathbf{v}, \mathbf{\dot x})=\left(m \mathbf{v} + \frac{e}{c}\mathbf{A}(\mathbf{x}) \right)\cdot \mathbf{\dot x}- \frac{m}{2}\left| \mathbf{v} \right|^{2}-e\Phi(\mathbf{x}).
\end{equation}
In the above, note that $ \mathbf{v}$ is treated as an independent variable of $\mathbf{\dot x}$, while the phase space Euler--Lagrange equations yield $ \frac{\partial \mathcal{L}}{\partial \mathbf{v}}=0$, and assuming $m \ne 0$, it follows that we get the second-order condition:
\[
\mathbf{\dot x}=\mathbf{v}.
\]
This formulation corresponds to a Hamilton--Pontryagin type variational principle, where $\mathbf{v}$ is treated as an independent variable and the constraint $\mathbf{v}=\mathbf{\dot x}$ is enforced variationally. In fact, the phase space Lagrangian in \eqref{VPSLag} can be restated by using the original Lagrangian and the constraint for the second-order condition as,
\begin{equation}\label{VPSLag_AnotherLook}
\begin{split}
\mathcal{L}(\mathbf{x}, \mathbf{v}, \mathbf{\dot x})&=L(\mathbf{x}, \mathbf{v})+ \frac{\partial L}{\partial \mathbf{v}}\cdot 
\big(\mathbf{\dot x} - \mathbf{v} \big).
\end{split}
\end{equation}

\remark[Phase-space Lagrangians]\rm
In the field of plasma physics, the concept of phase-space Lagrangian has been used as a variational formulation that can be described directly in terms of canonical coordinates and momenta,
from which Hamilton's equations follow naturally. In the presence of electromagnetic fields, the canonical momentum includes the vector potential, and the corresponding action principle yields the Lorentz force law. This framework provides a clear geometric interpretation through the canonical one-form and the underlying symplectic structure, making it particularly suitable for describing charged particle dynamics; see, for instance, \cite{HuMo1992, CaBr2009}.
In \cite{HuMo1992}, the Hamiltonian theory of guiding-center motion was developed to clarify the geometric structure of reduced particle dynamics in strong magnetic fields. In parallel, \cite{CHHM1998} expressed the Maxwell--Vlasov equations in Euler--Poincar\'e form, demonstrating how phase space variational principles can be systematically reduced using symmetry and Lie group methods. These works collectively show that the phase-space Lagrangian is not merely a reformulation of classical mechanics but a versatile framework for both single-particle and collective plasma dynamics. It underlies modern approaches to structure-preserving numerical methods, gyrokinetic reductions, and geometric formulations of complex plasma systems, providing a unified perspective on variational and Hamiltonian structures in plasma physics.

\paragraph{The Lagrange--Dirac interpretation of the phase-space Lagrangian.}
The derivation of the phase-space Lagrangian often appears ad hoc, since it is not defined on $TQ$ but rather on $TT^{\ast}Q$ or $TTQ$. Here we clarify its intrinsic nature within both Lagrange--Dirac geometric and variational structures on $TQ$.
\medskip

Now, starting with the Lagrangian $L: TQ \to \mathbb{R}$ from \eqref{LagChPEM},  the Legendre transform provides the following relation:
\begin{equation*}
 \frac{\partial L}{\partial \mathbf{v}}=m \mathbf{v} + \frac{e}{c}\mathbf{A}(\mathbf{x}).
\end{equation*}
The energy $E_L$ on the tangent bundle is defined as:
\begin{equation*}
\begin{split}
E_L(\mathbf{x},\mathbf{v})&= \frac{\partial L}{\partial \mathbf{v}}\cdot \mathbf{v}-L(\mathbf{x},\mathbf{v})\\
&=\left(m \mathbf{v} + \frac{e}{c}\mathbf{A}(\mathbf{x})\right)\cdot \mathbf{v}-\left(\frac{1}{2}m |\mathbf{v}|^2+\frac{e}{c}\mathbf{v}\cdot \mathbf{A}(\mathbf{x})-e\Phi(\mathbf{x}) \right)\\
&=\frac{1}{2}m |\mathbf{v}|^2+e\Phi(\mathbf{x}).
\end{split}
\end{equation*}
Since the Lagrange--Dirac variational principle on $TQ$ is given by
\begin{equation*}
\delta \int_0^T \bigg\{  \Theta_L(\mathbf{x},\mathbf{v})\cdot X(\mathbf{x},\mathbf{v})-E_L(\mathbf{x},\mathbf{v})\bigg\}dt=0,
\end{equation*}
for all variations $\delta\mathbf{v}$ and $\delta\mathbf{x}$ with the fixed endpoint conditions $\delta\mathbf{x}(0)=\delta\mathbf{x}(T)=0$, the Lagrange--Dirac dynamics for the charged particle in a electromagnetic field is given by
\[
\mathbf{i}_{X(\mathbf{x},\mathbf{v})}\Omega_{L}=\mathbf{d}E_L(\mathbf{x},\mathbf{v}),
\]
where $X$ is the section of $TTQ \to TQ$, locally expressed by $X(\mathbf{x},\mathbf{v}) = \dot{x}^i \frac{\partial}{\partial x^i} + \dot{v}^i \frac{\partial}{\partial v^i}$. 
\medskip

By direct computations using local coordinates, $x^{i}, v^{i}; i=1,\dots,n$, the critical condition of the Lagrange--Dirac principle is:
\begin{equation*}
\delta \int_0^T \bigg\{\left(m \mathbf{v} + \frac{e}{c}\mathbf{A}(\mathbf{x}) \right)\cdot \mathbf{\dot x}- \frac{m}{2}\left| \mathbf{v} \right|^{2}-e\Phi(\mathbf{x})\bigg\}dt=0
\end{equation*}
for all variations $\delta\mathbf{x}$ and $\delta\mathbf{v}$, where $\Theta_L=\frac{\partial L}{\partial \mathbf{v}}\cdot d\mathbf{x}=\big(\mathbf{A}(\mathbf{x})+m\mathbf{v}\big) \cdot d\mathbf{x}$. Then, it yields the first-order system of the Euler--Lagrange equations as follows:
\begin{equation*}
\left\{
\begin{aligned}
\dot{x}^{i}&=v^{i},\\
m\dot{v}^{i}&=\frac{e}{c}\left( \frac{\partial A_{j}}{\partial x^{i}}-\frac{\partial A_{i}}{\partial x^{j}} \right)\dot{x}^{j}-e \frac{\partial \Phi}{\partial x^{i}},\quad i=1,2,3,
\end{aligned}
\right.
\end{equation*}
where $\frac{\partial A_{j}}{\partial x^{i}}-\frac{\partial A_{i}}{\partial x^{j}}$ corresponds to the components of the magnetic field tensor. 
\medskip

Here, it is important to note that in the action functional of the Lagrange--Dirac principle, the term $\Theta_L(\mathbf{x},\mathbf{v})\cdot X(\mathbf{x},\mathbf{v})-E_L(\mathbf{x},\mathbf{v})$ exactly corresponds to the intrinsic expression of the phase-space Lagrangian given in \eqref{VPSLag} or \eqref{VPSLag_AnotherLook}; namely, we can deduce
\begin{equation*}
\mathcal{L}(\mathbf{x},\mathbf{v}, \dot{\mathbf{x}})\equiv \Theta_L(\mathbf{x},\mathbf{v})\cdot X(\mathbf{x},\mathbf{v})-E_L(\mathbf{x},\mathbf{v}),
\end{equation*}
where $\Theta_L$ depends only on $d\mathbf{x}$, and therefore the pairing $\Theta_L \cdot X$ involves only the $\dot{\mathbf{x}}$-component of $X$. 
This clarifies why the phase-space Lagrangian depends on $\dot{\mathbf{x}}$ but not on $\dot{\mathbf{v}}$, namely, it is \textit{degenerate} in the sense that it is linear in the velocity variable $\dot{\mathbf{x}}$ and independent of $\dot{\mathbf{v}}$. Furthermore, we also note that the phase-space Lagrangian is completely determined by the pair $(E_L, \Theta_L)$, and hence is an intrinsic object on $TQ$. 
Consequently, the commonly used velocity phase space formulation can be understood 
as a manifestation of the underlying Lagrange--Dirac system $(E_L, \Omega_{L}=-\mathbf{d}\Theta_L)$. This shows that the phase-space Lagrangian is not an auxiliary construction on $TTQ$, but is intrinsically determined by the Lagrange--Dirac data on $TQ$. In particular, the introduction of the auxiliary variable $\mathbf{v}$ does not enlarge the geometric structure, but merely makes explicit the underlying Dirac structure.

While phase-space Lagrangians have been extensively used in plasma physics, notably in the work of Littlejohn, their geometric origin has remained somewhat implicit. The above result clarifies that they are \textit{not ad hoc constructions, but intrinsic objects arising from the Lagrange--Dirac dynamical system}. This  provides a unified geometric framework for plasma physics via the Lagrange--Dirac variational formulation.

\subsection{Nonholonomic mechanical systems}
\paragraph{The classical  Heisenberg system.}
We consider a classical Heisenberg system where a particle moves in the configuration space $Q=\mathbb{R}^3$ under a potential field with a nonholonomic constraint; see \cite{Bl2003, PeYo2025}. Let $\mathbf{x}=(x,y,z) \in Q$ be the local coordinates for a point in $Q$. The Lagrangian is given by
\begin{equation*}
L(\mathbf{x},\dot{\mathbf{x}})=\frac{1}{2}\left(\dot{x}^2+\dot{y}^2+\dot{z}^2\right)-U(\mathbf{x}),
\end{equation*}
where $U(\mathbf{x})=\frac{1}{2}(x^{2}+y^{2})$ is the potential energy of the system. The motion $\mathbf{x}(t) \in Q, \; t \in [0, T]$ of the system is subject to the nonholonomic constraint:
\[
\dot{z}-y\dot{x}+x\dot{y}=0,
\]
defined by the constraint distribution
\[
\Delta_{Q}(q)=\left\{\dot{q} \in T_{q}Q \mid \left< \omega(q), \dot{q} \right>=0 \right\},
\]
 where $\omega(q)=dz-ydx+x dy$. The annihilator $\Delta_{Q}^{\circ} \subset T^{\ast}Q$ is represented locally using a Lagrange multiplier $\mu$ as $\alpha = \mu(dz - ydx + xdy)$.
\medskip

From the Lagrange--Dirac condition, the Lagrange--d'Alembert equations are obtained as:
\begin{equation*}
\begin{split}
&\dot{x}=v_{x}, \qquad \qquad\;  \dot{y}=v_{y}, \qquad\quad\;\, \dot{z}=v_{z}, \\
&\dot{v}_{x}=-\mu y - 2x, \;\; \dot{v}_{y}=\mu x - 2y, \;\; \dot{v}_{z}=\mu, \\
&\dot{z}=y\dot{x}-x\dot{y}.
\end{split}
\end{equation*}

This example illustrates how nonholonomic constraints define a distribution $\Delta_Q$ and its annihilator $\Delta_Q^\circ$, which are naturally incorporated into the Lagrange--Dirac formulation on $TQ$.

\paragraph{The Euler top with a nonholonomic constraint.}
Consider the Suslov problem (see, \cite{Bl2003,ZeBl2000,YoMa2007a}): a rigid body rotating about a fixed point such that the projection of the angular velocity in a body-fixed direction is zero. The configuration space is $G=SO(3)$. Suppose the Lagrangian $L:TG \to \mathbb{R}$ is left-invariant, with the reduced Lagrangian $\ell: \mathfrak{so}(3) \to \mathbb{R}$ given by the kinetic energy as:
\[
\ell(\mathbf{\Sigma})=\frac{1}{2}\mathbb{I} \mathbf{\Sigma} \cdot \mathbf{\Sigma},
\]
where $\mathbf{\Sigma}$ is the body angular velocity and $\mathbb{I}$ is the inertia tensor. The constraint distribution is 
\begin{equation*}
\mathfrak{so}(3)^{\Delta}=\big\{ \mathbf{\Omega}  \in \mathfrak{so}(3) \mid \langle \mathbf{A}, \mathbf{\Omega} \rangle=0 \big\},
\end{equation*}
where $\mathbf{\Omega}$ denotes the body angular velocity, $\mathbf{A}$ is a fixed element of the dual Lie algebra $\mathfrak{so}(3)^{\ast}$ and $\langle \cdot,\cdot\rangle$ denotes the natural paring between the Lie algebra and its dual. Since the subspace $\mathfrak{so}(3)^{\Delta}$ is not necessarily a subalgebra, the constraint is  nonholonomic. 

This system can be formulated by the reduced Lagrange--Dirac system $(\mathcal{E}_{\ell},D^{/G}_{L})$ given in  \eqref{condition_reducedLagDiracSys}. 
 Namely, the system equations for the Suslov problem can be expressed by the Euler-Poincar\'e-Suslov equations, which are given in this example by
\begin{eqnarray*}\label{IEPSEns-NonhSystems}
\Omega^{i}  = \Sigma^{i}, \quad
\frac{d}{dt}\frac{\partial {\ell}}{\partial \Sigma^{i}} - C^{k}_{ji} \Omega^{j}\Pi_{k} =\mu A_{i}, \quad i,j,k=1,2,3.
\end{eqnarray*}
In the above, $C^{k}_{ji}$ are the structure constants of $\mathfrak{so}(3)$ and $\mu$ is the Lagrange multiplier. Further, $\mathbf{\Sigma}=\Sigma^{i}\mathbf{e}_{i},\, \mathbf{\Omega}=\Omega^{i}\mathbf{e}_{i} \in \mathfrak{so}(3)^{\Delta}$, and $\mathbf{A}=A_{i}\mathbf{e}^{i} \in \mathfrak{so}(3)^{\ast}$, where $\mathbf{e}_{i},\,i=1,2,3$ form a basis for $\mathfrak{so}(3)$ and $\mathbf{e}^{i}$ form a basis for $\mathfrak{so}(3)^{\ast}$. 

Choose $\mathbf{e}_{3}=\mathbf{A}/\left|\mathbf{A}\right|$ as the third vector of the body frame and let us pick up two independent vectors $\mathbf{e}_{1}, \mathbf{e}_{2}$ that are orthogonal to $\mathbf{e}_{3}$ in the kinetic energy metric.  Therefore, we have the constraint
$
\Omega_{3}=0.
$
Note that $\mathbf{e}_{1}, \mathbf{e}_{2}$ and  $\mathbf{e}_{3}$ are not necessarily orthogonal relative to the standard metric in $\mathfrak{so}(3) \cong \mathbb{R}^{3}$ unless $\mathbf{e}_{3}$ spans an eigenspace of the inertia tensor $\mathbb{I}=I_{ij} \mathbf{e}_{i} \otimes \mathbf{e}_{j}$, which means that $C^{j}_{ji}$ are not equal to zero. Since the components of the inertia tensor $I_{13}$ and $I_{23}$ are zero, we get
\[
\frac{\partial {\ell}}{\partial \Sigma^{1}}=I_{11}\,\Sigma^{1}+I_{12}\,\Sigma^{2}, \quad \frac{\partial {\ell}}{\partial \Sigma^{2}}=I_{21}\,\Sigma^{1}+I_{22}\,\Sigma^{2}, \quad \frac{\partial {\ell}}{\partial \Sigma^{3}}=I_{33}\,\Sigma^{3}. 
\]
The Euler-Poincar\'e-Suslov equations are given in matrix form by
\[
\frac{d}{dt}
\begin{pmatrix}
\frac{\partial {\ell}}{\partial \Sigma^{1}}\\[1mm]
\frac{\partial {\ell}}{\partial \Sigma^{2}}\\[1mm]
\frac{\partial {\ell}}{\partial \Sigma^{3}}
\end{pmatrix}
=
\begin{pmatrix}
C^{1}_{21} \Omega^{2}\frac{\partial {l}}{\partial \Sigma^{1}}+C^{2}_{21} \Omega^{2}\frac{\partial {l}}{\partial \Sigma^{2}} \\[1mm]
C^{1}_{12} \Omega^{1}\frac{\partial {l}}{\partial \Sigma^{1}}+C^{2}_{12} \Omega^{1}\frac{\partial {l}}{\partial \Sigma^{2}} \\[1mm]
C^{1}_{13} \Omega^{1}\frac{\partial {l}}{\partial \Sigma^{1}}+C^{2}_{13} \Omega^{1}\frac{\partial {l}}{\partial \Sigma^{2}}+C^{1}_{23} \Omega^{2}\frac{\partial {l}}{\partial \Sigma^{1}}+C^{2}_{23} \Omega^{2}\frac{\partial {l}}{\partial \Sigma^{2}} 
\end{pmatrix}
+
\begin{pmatrix}
0 \\
0 \\
 \mu \, |\mathbf{A}| 
\end{pmatrix},
\]
together with 
$
\Omega^{1}=\Sigma^{1}, \Omega^{2}=\Sigma^{2}, \Omega^{3}=\Sigma^{3}=0.
$
By eliminating the Lagrange multiplier, we get the Euler-Poincar\'e-Suslov equations:
\begin{align*}
\frac{d}{dt}
\left(
\begin{array}{c}
\frac{\partial {\ell}}{\partial \Omega^{1}}\\[1mm]
\frac{\partial {\ell}}{\partial \Omega^{2}}\\[1mm]
\end{array}
\right)
&=
\left(
\begin{array}{c}
C^{1}_{21} \Omega^{2}\frac{\partial {\ell}}{\partial \Omega^{1}}+C^{2}_{21} \Omega^{2}\frac{\partial {\ell}}{\partial \Omega^{2}} \\[1mm]
C^{1}_{12} \Omega^{1}\frac{\partial {\ell}}{\partial \Omega^{1}}+C^{2}_{12} \Omega^{1}\frac{\partial {\ell}}{\partial \Omega^{2}} \\
\end{array}
\right),
\end{align*}
which are equivalent with the equations in \cite{ZeBl2000}.
 \medskip
 
If we choose $\mathbf{e}_{3}$ as an eigenvector of the inertia tensor such that $\mathbb{I}\mathbf{e}_{3}=I_{33}\mathbf{e}_{3}$ and if we also choose $\mathbf{e}_{1}$ and $\mathbf{e}_{2}$ as the two remaining eigenvectors, the basis $\mathbf{e}_{1}, \mathbf{e}_{2}, \mathbf{e}_{3}$ in $\mathfrak{so}(3)$ is orthogonal with respect to both the standard and the kinetic energy metrics, where $C^{3}_{12}\!=\!C^{1}_{23}\!=\!C^{2}_{31}\!=\!-C^{3}_{21}\!=\!-C^{1}_{32}\!=\!-C^{2}_{13}\!=\!1$ and $C^{k}_{ji}=0$ otherwise. In this case, it follows
\[
\frac{d}{dt}\frac{\partial {\ell}}{\partial \Omega^{1}}=0, \quad \frac{d}{dt}\frac{\partial {\ell}}{\partial \Omega^{2}}=0.
\] 
Finally, all the solutions of the reduced system are relative equilibria.

\subsection{Degenerate Lagrangian systems}
Here, we consider three examples of degenerate Lagrangian systems. We illustrate how such systems can be systematically formulated in the framework of Lagrange--Dirac dynamical systems on $TQ$, including systems with additional linear constraints arising from electric circuits (e.g., Kirchhoff laws) as well as classical nonholonomic mechanical constraints.

\paragraph{Degenerate Lagrangians.}
Consider a system with a configuration space $Q=\mathbb{R}^{3}$ with the following Lagrangian $L(q,v): TQ \to \mathbb{R}$, see \cite{Ha2011}:
\[
L(q, v)=v^{1}v^{3}-q^{2}v^{3}+q^{1}q^{3}.
\]
The Hessian matrix of $L$ is
\[
\mathbf{D}_{2}\mathbf{D}_{2}L(q,v)
=
\begin{pmatrix}
0 & 0 & 1 \\
0 & 0 & 0 \\
1 & 0 & 0\\
\end{pmatrix},
\]
and hence is singular, the system is degenerate. 

From Proposition \ref{LagrangeDiracVarPrin}, we compute the critical condition of the action functional 
\eqref{ActionInt_HamVelocityPhaseSpacePrin} as
\begin{equation*}
\begin{split}
&\delta\mathcal{S}(q,v)=\delta\int_{0}^{T} \bigg\{ L(q(t),v(t)) + \mathbf{D}_{2}L(q(t),v(t)) \cdot \big(\dot{q}(t)-v(t)\big) \bigg\}dt\\[2mm]
&=\delta\int_{0}^{T} \bigg\{ \big(v^{1}v^{3}-q^{2}v^{3}+q^{1}q^{3} \big)+ v^{3}\big(\dot{q}^{1}-v^{1}\big)+ \big(v^{1}-q^{2}\big)\big(\dot{q}^{3}-v^{3}\big)\bigg\}dt\\[2mm]
&=\delta\int_{0}^{T} \bigg\{ v^{3}\dot{q}^{1}+v^{1}\dot{q}^{3}-q^{2}\dot{q}^{3}+q^{1}q^{3}-v^{1}v^{3}\bigg\}dt \\[2mm]
&=\delta\int_{0}^{T} \bigg\{ \delta{q}^{1}\big(\!-\dot{v}^{3}+q^{3}\big)-\delta{q}^{2}\dot{q}^{3}+\delta{q}^{3}\big(\!-\dot{v}^{1}+\dot{q}^{2}+q^{1}\big)+\delta{v}^{1}\big(\dot{q}^{3}-v^{3}\big)+\delta{v}^{3}\big(\dot{q}^{1}-v^{1}\big)\bigg\}dt=0,
\end{split}
\end{equation*}
for all $\delta{q}^{1}, \delta{q}^{2}, \delta{q}^{3}, \delta{v}^{1}$, and for all $\delta{v}^{3}$ with the fixed endpoint conditions. This yields the following Lagrange--Dirac equations:
\begin{equation*}
\begin{split}
-\dot{v}^{3}+q^{3}=0,\quad \dot{q}^{3}=0, \quad -\dot{v}^{1}+\dot{q}^{2}+q^{1}=0, \quad \dot{q}^{3}-v^{3}=0, \quad \dot{q}^{1}-v^{1}=0.
\end{split}
\end{equation*}
By eliminating $v$-variables, these equations are reduced to 
\[
\ddot{q}^{1}=\dot{q}^{2}+q^{1},\quad \ddot{q}^{3}=q^{3}, \quad \dot{q}^{3}=0,
\]
which are identical to the equations of motion obtained in \cite{Ha2011}. Here we note that there exists a \textit{hidden constraint} $\dot{q}^{3}=0$ in this degenerate Lagrangian system, which is systematically obtained in the framework of  Lagrange--Dirac dynamical systems.

\paragraph{Lagrangians linear in velocity and link to Hamilton's principle.} We consider a Lagrangian linear in velocity, see \cite{RoMa2002, YoMa2007b}. Let $Q$ be an $n$ dimensional manifold and let $L: TQ \to \mathbb{R}$ be a Lagrangian, linear in velocity, given in local coordinates $q^{1},\dots,q^{n},v^{1},\dots, v^{n}$, by:
\[
L(q,v)=\alpha_{i}(q) v^{i}-U(q),
\]
which is completely degenerate since 
\[
\frac{\partial^{2} L}{\partial v^{i} \partial v^{j}} =0.
\]
One can establish the Lagrange--Dirac variational principle:
\[
\begin{split}
\displaystyle
\int_{0}^{T} \bigg[ 
\frac{\partial^{2} L}{\partial v^{i}v^{j}} 
\left(\dot{q}^{j}-v^{j} \right) \delta{v}^{i} 
+
\bigg\{ 
- \frac{d}{dt}\frac{\partial L}{\partial v^{i}} +  \frac{\partial L}{\partial q^{i}} + \frac{\partial^{2} L}{\partial q^{i} \partial v^{j}} \left(\dot{q}^{j}-v^{j}\right) 
\bigg\} \delta{q}^{i} 
\bigg]dt=0,
\end{split}
\]
for all variations $\delta{v}$ and for all $\delta{q}$ with the fixed endpoint conditions $\delta{q}(0)=\delta{q}(T)=0$.

Then, the Lagrange--Dirac equations are given, in local coordinates, by:
\[
\displaystyle \frac{d}{dt}\frac{\partial L}{\partial v^{i}} - \frac{\partial L}{\partial q^{i}}  -  \frac{\partial^{2} L}{\partial q^{i} \partial v^{j}} \left( \dot{q}^{j} - v^{j}\right)=0,\qquad
\displaystyle \frac{\partial^{2} L}{\partial v^{i} \partial v^{j}} \left(\dot{q}^{j}-v^{j}\right)=0,
\]
where 
\[
\begin{split}
&\displaystyle\frac{\partial L}{\partial v^{i}}=\alpha_{i}(q),\quad \frac{d}{dt}\frac{\partial L}{\partial v^{i}} = \frac{d\alpha_{i}}{dt}
= \frac{\partial \alpha_{i}}{\partial q^{j}}\dot{q}^{j},\quad 
\frac{\partial L}{\partial q^{i}} = \frac{\partial \alpha_{j}}{\partial q^{i}} v^{j} - \frac{\partial U}{\partial q^{i}},\\[3mm]
&\displaystyle  \frac{\partial^{2} L}{\partial q^{i} \partial v^{j}}\left( \dot{q}^{j} - v^{j}\right)=
\frac{\partial \alpha_{j}}{\partial q^{i}} \left( \dot{q}^{j} - v^{j}\right).
\end{split}
\]
Thus, we get the same equations obtained in \cite{RoMa2002, YoMa2007b}:
\[
\bigg( \frac{\partial \alpha_{i}}{\partial q^{j}} - \frac{\partial \alpha_{j}}{\partial q^{i}}  \bigg)\dot{q}^{j} = - \frac{\partial U}{\partial q^{i}}.
\]
\begin{remark}[Relationship with Hamilton's principle on $Q$]
It is well known that the standard Hamilton's principle on the configuration space $Q$ can formally yield the Euler--Lagrange equations even for degenerate Lagrangians. 
This is because the standard principle restricts the variation $\delta \dot{q}$ to be the time derivative of $\delta q$ from the outset, thereby a priori imposing the second-order condition $v = \dot{q}$. 

In contrast, our variational principle on the extended space $TQ$ (the velocity phase space) treats $v$ (velocity) and $\dot{q}$ (kinematic derivative) as independent variables. 
The resulting condition $\mathbf{D}_{2}\mathbf{D}_{2}L(q,v) \cdot (\dot{q} - v) = 0$ reveals the following:
\begin{itemize}
    \item In the hyperregular case, the condition $\mathbf{D}_{2}\mathbf{D}_{2}L(q,v) \cdot (\dot{q} - v) = 0$ uniquely determines $\dot{q} = v$, recoverning the standard kinematic relationship of the Hamilton's principle.
    \item In the degenerate case, the condition allows for $\dot{q} - v \in \ker \mathbf{D}_{2}\mathbf{D}_{2}L(q,v)$. In the above example of a first-order Lagrangian, where $\mathbf{D}_{2}L(q,v)=0$, any $v$ and $\dot{q}$ satisfy the condition.
\end{itemize}
Crucially, the standard kinematic requirement $\dot{q} = v$ is always contained within the set of critical curves of our principle as one possible choice. 
While the standard Hamilton's principle views the dynamics through a single ``section'' where $\dot{q} = v$ is enforced, the Lagrange--Dirac formulation provides a broader geometric framework that preserves the underlying Dirac structure even when the Lagrangian is degenerate.
\end{remark}

\paragraph{Electric circuits.}
Consider an electric circuit illustrated in Figure \ref{circuit}, consisting of an inductor $I_{1}$ and three capacitors $C_i,\;i=2,3,4$, where the index corresponds to the number of the branch. This example was explored by Yoshimura and Marsden \cite{YoMa2006a, YoMa2007b}, and also by Cendra et. al. \cite{CeEtFe2014}. 

\begin{figure}[htb]
\centering
\includegraphics[scale=0.8]{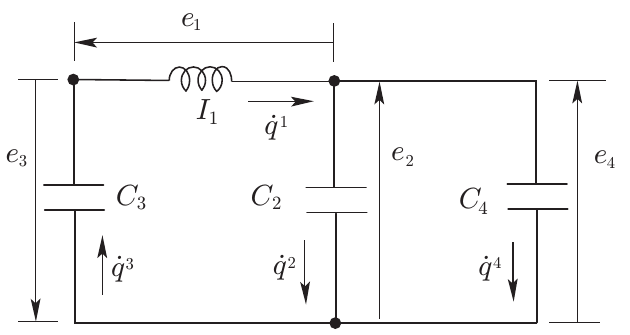}
\caption{Electric circuit}
\label{circuit}
\end{figure}

The configuration space $Q=\mathbb{R}^4$ is the space of branch charges, i.e., for each $q \in Q$ is locally denote by
\[
q=(q^{1},q^{2},q^{3},q^{4}) \equiv(q_{I_{1}},q_{C_{2}},q_{C_{3}},q_{C_{4}}) \in Q.
\]
Then the branch currents are denoted by elements of the tangent bundle $TQ$: for each $q \in Q$,
\[
\dot q=(\dot q^{1},\dot q^{2},\dot q^{3},\dot q^{4}) \equiv(\dot q_{I_{1}},\dot q_{C_{2}},\dot q_{C_{3}},\dot q_{C_{4}}) \in T_{q}Q.
\]
The Lagrangian of this electric circuit $L: TQ \to \mathbb{R}$ is given by  the magnetic energy due to the inductor minus electric potential energy 
due to the capacitors as
\[
L(q,v):=\frac{1}{2}I_{1}\, (v^{1})^2-\bigg(\frac{1}{2}\frac{(q^{2})^2}{C_2}
+\frac{1}{2}\frac{(q^{3})^2}{C_3}+\frac{1}{2}\frac{(q^{4})^2}{C_4}\bigg),
\]
where $v=(v^{1},v^{2},v^{3},v^{4}) \equiv (v_{I_{1}},v_{C_{2}},v_{C_{3}},v_{C_{4}}) \in T_{q}Q$. This Lagrangian is degenerate since the Hessian of $L$ is singular:
\[
\mathbf{D}_{2}\mathbf{D}_{2}L(q,v)
=
\begin{pmatrix}
I_{1} & 0 & 0 & 0 \\
0 & 0 & 0 & 0 \\
0 & 0 & 0 & 0 \\
0 & 0 & 0 & 0 \\ 
\end{pmatrix}.
\]

There exists constraints among the branch currents, called \textit{Kirchhoff Current Law (KCL)} as
\begin{equation}\label{KCL_const}
\begin{split}
-\dot{q}^{1}+\dot{q}^{3}=0,\quad \dot{q}^{2}-\dot{q}^{3} + \dot{q}^{4}=0.
\end{split}
\end{equation}
These constraints can be represented by
\[
\Delta_{Q}(q)=\{ \dot{q} \in T_qQ \mid \langle \omega^{r}(q), \dot{q} \rangle =0, \; r=1, 2 \}.
\]
where $\omega^{r}=\omega_{i}^{r} dq^{i}, i=1,\dots,4; r=1,2$ denote constraint one-forms, and the coefficients $\omega_{i}^{r}$ are given in matrix by
\[
(\omega_{i}^{r})=
\left(
\begin{array}{cccc}
-1 & 0 & 1 & 0 \\
0 & 1 & -1 & 1 \\
\end{array}
\right).
\]
In parallel with KCL, there exists the dual relations among the branch voltages $e_{I_{1}}, e_{C_{2}},$ $e_{C_{3}}, e_{C_{4}}$, which is called \textit{Kirchhoff Voltage Law (KVL)}. The KVL is geometrically understood in the context of the annihilator $\Delta_{Q}^{\circ} \subset T^{\ast}Q$, which gives us the constraint space of branch voltages, corresponding to the space of constraint forces in mechanics. In fact, an element 
$e=(e_{1}, e_{2}, e_{3}, e_{4})\equiv (e_{I_{1}}, e_{C_{2}}, e_{C_{3}}, e_{C_{4}}) \in \Delta_{Q}^{\circ}$ is given by introducing unknown Lagrange multipliers $\mu_{r}, r=1,2$ as
\[
e_{i}=\mu_{r}\omega_{i}^{r}.
\]
 
Thus, the Lagrange--d'Alembert--Dirac equations \eqref{LocalImpDynLagDirSys}
\begin{equation*}
\begin{cases}
\displaystyle \frac{d}{dt}\frac{\partial L}{\partial v^{i}}   
-  \frac{\partial L}{\partial q^{i}}
- \frac{\partial^2 L}{\partial q^{i}\partial v^j } \left(\dot{q}^{j}- v^{j}\right) =\mu_{r}\omega^{r}_{i},\\[5mm]
\displaystyle \frac{\partial^2 L}{\partial v^{i} \partial v^{j}}(\dot{q}^{j}- v^{j})=0,\\[4mm]
\displaystyle \omega_{i}^{r}\dot{q}^{i}=0
\end{cases}
\end{equation*}
are computed by
\begin{equation*}
\begin{cases}
\displaystyle I_{1}\dot{v}^{1}=-\mu_{1},\quad \frac{q^{2}}{C_2}=\mu_{2}, \quad \frac{q^{3}}{C_{3}}=\mu_{1}-\mu_{2}, \quad \frac{q^{4}}{C_{4}}=\mu_{2},\\[3mm]
\displaystyle I_{1} (\dot{q}^{1}- v^{1})=0,\\[3mm]
\displaystyle -\dot{q}^{1}+\dot{q}^{3}=0,\qquad \dot{q}^{2} - \dot{q}^{3} + \dot{q}^{4}=0.
\end{cases}
\end{equation*}
Finally, we get the required implicit dynamics as follows:
\begin{equation*}
\begin{cases}
\displaystyle I_{1}\ddot{q}^{1}=-\mu_{1},\quad \frac{q^{2}}{C_{2}}=\mu_{2}, \quad \frac{q^{3}}{C_{3}}=\mu_{1}-\mu_{2}, \quad \frac{q^{4}}{C_{4}}=\mu_{2},\\[3mm]
\displaystyle \dot{q}^{1}=\dot{q}^{3},\qquad \dot{q}^{3}=\dot{q}^{2}+\dot{q}^{4}.
\end{cases}
\end{equation*}

This example shows that the Lagrange--Dirac structure naturally encodes both Lagrangian degeneracy and nonholonomic constraints.

\subsection{Motion of ideal fluids}
Now we consider the motion of an ideal (incompressible and inviscid) fluid in a compact Riemannian manifold $M$ in $\mathbb{R}^{3}$ with smooth boundary $\partial M$. Let $g$ be a Riemannian metric and $\mu_{g}$ the associated volume form. We assume that the fluid is homogeneous and hence the mass density $\rho_{0}$ is constant.

The motion of a fluid in $M$ is described by a curve $\varphi_{t}$ in $\operatorname{Diff}(M)$, representing a family of diffeomorphisms from the reference configuration to the current configuration in $M$. That is, $x(X,t):=\varphi_{t}(X) \in M$, where $x$ is the current Eulerian spatial location of the particle with label $X$. 

We  define the Lagrangian (material) velocity of the fluid by taking the time derivative of the trajectory while keeping the particle label $X$ fixed:
as 
\[
U(X,t):=\frac{\partial \varphi_{t}(X)}{\partial t}=\dot{\varphi}_t(X).
\]
The Eulerian (spatial) velocity $v$ at the current location $x$ is defined by:
\[
u(x,t) := U(\varphi_{t}^{-1}(x), t).
\]
Therefore, the  relationship between the Lagrangian velocity $(\varphi,\dot{\varphi})$ and the Eulerian velocity $u$ is:
\[
u =\dot{\varphi}_t  \circ \varphi_t ^{-1}, \quad \text{i.e.}, \quad  u = U_{t} \circ \varphi_{t}^{-1}.
\]
Note that the Eulerian velocity $u$ belongs to the Lie algebra $ \mathfrak{g}=\mathfrak{X}  (M)$ of $\operatorname{Diff}(M)$. 
\medskip

Let $G=\mathrm{Diff}(M)$ be the group of all smooth diffeomorphisms of $M$. Its Lie algebra is the space $\mathfrak{g}=\mathfrak{X}(M)$ of all smooth vector fields on $M$ tangent to the boundary, endowed with the Lie bracket 
\[
[u,v]=\nabla_v u-\nabla_u v,
\]
for $u, v \in \mathfrak{g}$, where $\nabla$ denotes the Levi--Civita connection associated with $g$. In this example, we employ the right-invariant convention. 

Consider the weakly nondegenerate $L^2$-paring on $\mathfrak{g}$: 
\[
\left<u,v\right>_{\mathfrak{g}}=\int_{M}g(u,v)\mu_{g}.
\]

We emphasize that this is a weak $L^{2}$-pairing on the Fr\'echet space $ \mathfrak{X}  ( M)$; thus $\mathfrak{g}^{\ast}$ here is understood formally as the space of smooth one-forms paired with vector fields through the $L^{2}$-pairing
%
\[
 \left<\alpha,u\right>=\int_{M}\alpha(u)\mu_{g},
\]
where $\alpha \in \Omega^{1}(M)$.

From the incompressibility condition, the divergence free condition is imposed on velocity $u$, i.e., it is restricted to a subspace $\mathfrak{g}^{\Delta} \subset \mathfrak{g}$ as
\[
\mathfrak{g}^{\Delta}:=\{ u \in \mathfrak{g} \mid \operatorname{div}u=0,\;\; u \cdot n=0 \; \textrm{at $\partial M$} \},
\] 
where $n$ is the outward unit normal vector at each point on $\partial M$. The annihilator of $\mathfrak{g}^{\Delta}$ is defined by
\[
(\mathfrak{g}^{\Delta})^{\circ}:=\{ \alpha \in \mathfrak{g}^{\ast} \mid \left<\alpha, u\right>=0,\quad \forall u \in \mathfrak{g}^{\Delta}\}.
\]
For $u \in \mathfrak{g}^{\Delta}$, we have
\[
\pounds_{u}\mu_{g}=(\operatorname{div}u)\mu_{g}=0.
\]
For some smooth function $f \in C^{\infty}(M)$, the divergence theorem yields:
\[
\begin{split}
\int_{M}f(\operatorname{div}u)\mu_{g}&=\int_{M}\operatorname{div}(fu)\mu_{g}-\int_{M}u(f)\mu_{g}\\
&=\int_{\partial M}f(u\cdot n) \mu_{g}^{\partial}-\int_{M}u(f)\mu_{g}\\
&=-\int_{M}u(f)\mu_{g}=-\left<\mathbf{d}f, u\right>,
\end{split}
\] 
where $\mu_{g}^{\partial}$ is the volume on the boundary $\partial M$ induced from $\mu_{g}$,  and $n$ is the outward unit normal. 
Since $u \cdot n = 0$ on $\partial M$ and $\operatorname{div} u = 0$, this reduces to:
\[
0 = \int_{M} u(f) \mu_{g} = \langle \mathbf{d}f, u \rangle.
\]
Thus, the exact one-forms $\mathbf{d}f$ belong to the annihilator:
\[
\mathbf{d}f \in (\mathfrak{g}^{\Delta})^{\circ}.
\]

\begin{remark}
When extending the present Lagrange--Dirac framework to infinite-dimensional systems such as ideal fluids, one must naturally consider the appropriate functional-analytic framework. Strictly speaking, the space of smooth diffeomorphisms does not form a Banach Lie group. To ensure the rigorous existence of the nonholonomic or geodesic flows, it is standard to work with Sobolev class $H^s$ ($s > \dim M / 2 + 1$) vector fields and diffeomorphisms, where the kinetic energy defines a smooth Riemannian metric and the Cauchy problem is well-posed (see, e.g., \cite{EbMa1970}). In our formulation of the incompressible Euler--Poincar\'e--Saslov equations, we view the incompressibility condition as a nonholonomic-type constraint imposed on the larger framework of compressible fluids. In this paper, therefore, we restrict our attention to the formal geometric aspect to highlight the variational and Dirac structures behind the dynamics of the ideal fluid.
\end{remark}

\begin{remark}
In the incompressible case, the fluid motion is described by a curve in the subgroup $\operatorname{Diff}_{\operatorname{vol}}(M)=\{\varphi \in \operatorname{Diff}(M) \mid J \varphi = 1\}$ of volume preserving diffeomorphisms of $M$, where $J \varphi$ denotes the Jacobian of the diffeomorphism $\varphi$. Its Lie algebra is given by the space $\mathfrak{X}_{\operatorname{vol}}(M)$ of divergence free smooth vector fields parallel to the boundary, i.e., $\mathfrak{g}^{\Delta}=\mathfrak{X}_{\operatorname{vol}}(M)$. Correspondingly, the dual of $\mathfrak{X}_{\operatorname{vol}}(M)$ is given by the space of one-forms modulo exact one-forms, i.e., $\Omega^{1}(M)/\mathbf{d}C^{\infty}(M)$. Here, however, we treat the system as a constrained system in the ambient space $\operatorname{Diff}(M)$.
\end{remark}

\paragraph{Euler--Poincar\'e--Suslov equations for incompressible inviscid fluids.}
Let us introduce an inertia operator $\mathbb{I}: \mathfrak{g} \to \mathfrak{g}^{\ast}$ such that
\[
\left<\mathbb{I}v, u \right>=\left<v,u \right>_{\mathfrak{g}}=\int_M \rho_0v^{\flat}(u)\mu_{g}=\int_M \rho_0g(v,u)\mu_{g},
\]
where we assume that the fluid is incompressible and homogeneous, namely, the mass density $\rho_{0}$ is constant.
The (reduced) Lagrangian of the fluid is $\ell: \mathfrak{X}(M) \rightarrow \mathbb{R}$, given by
\[
\ell(v)=\frac{1}{2}\left<\mathbb{I}v, v \right>=\frac{1}{2}\left<v,v\right>_{\mathfrak{g}}=  \frac{1}{2} \int_M \rho_{0}g(v,v) \mu_{g},
\]
where we consider the weakly nondegenerate $L^2$-pairing and hence the Lagrangian $\ell(v)$ is nondegenerate. 

The Fr\'echet derivative $\mathbf{D} \ell(v) \in \mathfrak{g}^*$ is defined by
\[
\mathbf{D} \ell(v) \cdot \delta v= \int_M \rho_{0}g(v, \delta v) \mu_{g}.
\]
Using the $L^2$-pairing, it can be represented as
\[
\mathbf{D} \ell(v) \cdot \delta v
=
\left< \frac{\delta \ell}{\delta v}, \delta v\right>,
\quad
\frac{\delta \ell}{\delta v}=\mathbb{I}v=\rho_{0}v^{\flat}.
\]

Now, we apply the reduced Lagrange--d'Alembert--Dirac principle as follows:
Find a critical curve $(u(t), v(t)),\;t \in [0,T]$ on $\mathfrak{g} \times \mathfrak{g}$ of the reduced action functional:
\[
\int_{0}^{T} \bigg\{ \ell(v(t))+ 
\mathbf{D} \ell(v) \cdot \big( u(t) - v(t) \big) \bigg\} dt,
\]
namely, it satisfies
\begin{equation*}
\delta \int_{0}^{T} \bigg\{ \ell(v(t))+ 
\mathbf{D} \ell(v) \cdot \big( u(t) - v(t) \big) \bigg\} dt = 0, \quad \delta{u}=\frac{\partial \zeta}{\partial t}+[u, \zeta], 
\end{equation*}
subject to $u=\dot{\varphi}\circ \varphi^{-1} \in \mathfrak{g}^\Delta$ and $\zeta=\delta{\varphi}\circ \varphi^{-1} \in \mathfrak{g}^\Delta$ vanishes at the endpoints,
where $v=V\circ \varphi^{-1}$.
\medskip

The critical curve $(v(t), u(t))$ on $\mathfrak{g} \times \mathfrak{g}$ is a solution curve of the {\it Euler--Poincar\'e--Suslov equations}:
\begin{equation*}
\frac{d}{dt}\mathbf{D} \ell(v) - \operatorname{ad}_{u}^{\ast} \mathbf{D} \ell(v) \in (\mathfrak{g}^{\Delta})^{\circ},  \quad u = v, \quad u \in \mathfrak{g}^{\Delta}.
\end{equation*} 
In the above, since the Lagrangian $\ell(\eta)$ is nondegenerate, we get $u=v$. Because of the right action of $\mathrm{Diff}(M)$ on itself, the coadjoint action is given by $\left<\mathrm{ad}^{\ast}_{u}\alpha, \eta \right> = \left<\alpha, -[u,\eta]\right>$ for $\alpha \in \Omega^{1}(M)$ and $\eta \in \mathfrak{g}$. Since $u \in \mathfrak{g}^\Delta$ (i.e., $\operatorname{div}u = 0$ and $u \cdot n = 0$ on $\partial M$), the relation $\mathrm{ad}^{\ast}_{u}\alpha = \pounds_{u}\alpha$ holds for any one-form $\alpha$. Specifically, with $\mathbf{D} \ell(v) = \rho_{0}v^{\flat}$, we have:
\[
\mathrm{ad}^{\ast}_{u}\mathbf{D} \ell(v) = \rho_{0}\pounds_{u}v^{\flat}.
\]
By substituting $u=v$ into the above Euler--Poincar\'e--Suslov equations, we obtain:
\[
\rho_{0} \left( \frac{\partial u^{\flat}}{\partial t} + \pounds_{u}u^{\flat} \right) \in (\mathfrak{g}^{\Delta})^{\circ}, \quad \operatorname{div}u = 0.
\]
By Cartan's formula, the Lie derivative of the velocity 1-form is given by
\[
\pounds_{u}u^{\flat} = \mathbf{i}_{u}\mathbf{d}u^{\flat} + \mathbf{d}(\mathbf{i}_{u}u^{\flat}) = \mathbf{i}_{u}\mathbf{d}u^{\flat} + \mathbf{d}|u|^{2}.
\]
To align this with the classical notion of covariant acceleration,
we utilize the Levi--Civita connection $\nabla$ associated with the
Riemannian metric $g$.
Using the classical identity valid for metric-compatible torsion-free
connections,
\[
\mathbf{i}_u \mathbf{d} u^\flat
=
(\nabla_u u)^\flat
-\frac12 \mathbf{d}|u|^2,
\]
Cartan's formula gives
\[
\pounds_u u^\flat
=
(\nabla_u u)^\flat
+\frac12 \mathbf{d}|u|^2.
\]

Then, the Euler--Poincar\'e--Suslov equation becomes:
\[
\rho_{0} \left( \frac{\partial u^{\flat}}{\partial t} + (\nabla_u u)^\flat + \frac{1}{2}\mathbf{d}|u|^{2} \right) \in (\mathfrak{g}^{\Delta})^{\circ}.
\]
As established, $(\mathfrak{g}^{\Delta})^{\circ}$ consists of exact one-forms $\mathbf{d}f$ for some $f \in C^{\infty}(M)$. 
By defining the total pressure $p$ such that it appears as a constraint force, $-\mathbf{d}p := -(\mathbf{d}f + \frac{1}{2}\rho_{0}\mathbf{d}|u|^{2}) \in (\mathfrak{g}^{\Delta})^{\circ}$, we recover the classical Euler equations:
\[
\rho_{0} \left( \frac{\partial u}{\partial t} + \nabla_u u \right) = -\nabla p, \quad \operatorname{div}u = 0,
\]
where $(\nabla p)^{\flat}=\mathbf{d}p$.
This example shows that ideal fluid dynamics can be naturally formulated in terms of a Lagrange--Dirac structure on the infinite-dimensional Lie group $\mathrm{Diff}(M)$, where the Euler equations are naturally understood as Euler--Poincar\'e--Suslov type equations in the context of infinite-dimensional Lagrange--Dirac dynamical systems with constraints.

\begin{remark}[Geometric interpretation of pressure]
In the standard geometric formulation of ideal fluids (e.g., \cite{Ar1966, MaRa1999}), the fluid motion is described as a geodesic flow on the group of volume-preserving diffeomorphisms $\mathrm{Diff}_{\mathrm{vol}}(M)$. In this framework, the pressure term is recovered through the projection onto the divergence-free subspace.
In contrast, the present Lagrange--Dirac framework treats the incompressibility condition $\operatorname{div}u=0$ as a constraint on the ambient space $\mathfrak{X}(M)$. Consequently, the pressure gradient $-\mathbf{d}p$ arises naturally as an element of the annihilator $(\mathfrak{g}^{\Delta})^\circ$, namely, as the constraint force associated with incompressibility, rather than being introduced through an external projection procedure.
This provides a direct geometric interpretation of pressure within the dynamics itself and demonstrates that the Lagrange--Dirac structure naturally incorporates both the incompressibility constraint and the associated pressure force into a unified variational and geometric framework.
\end{remark}

\paragraph{Summary of examples.}
The above examples illustrate that the Lagrange--Dirac formulation provides a unified geometric and variational framework for a broad class of mechanical systems. 

In the example of charged particle dynamics in electromagnetic fields, the velocity phase-space Lagrangian introduced by Littlejohn admits a natural interpretation within the Lagrange--Dirac framework. Although formulated as a degenerate Lagrangian on $TTQ$, it is shown to be intrinsically determined by the Lagrange--Dirac data $(E_L,\Theta_L)$ on $TQ$. This clarifies the geometric origin of phase-space variational formulations in plasma physics.

For nonholonomic systems, such as the Heisenberg system and the Suslov problem, constraint distributions and their annihilators naturally define Dirac structures, leading to the Lagrange--Dirac and Euler--Poincar\'e--Suslov equations. 
For degenerate Lagrangian systems, including electric circuits, the Lagrange--Dirac formulation naturally incorporates the degeneracy of the Lagrangian. In particular, we clarify the relationship between the Lagrange--Dirac variational principle on $TQ$ and the standard Hamilton principle on $Q$ in the case of Lagrangians linear in velocity.

Furthermore, we present ideal fluid dynamics as an infinite-dimensional example, in which the divergence-free condition induces a constraint distribution and the pressure appears naturally as a constraint force belonging to the annihilator, thereby recovering the Euler equations through the framework of the Lagrange--Dirac dynamical system.

These examples demonstrate that degeneracy, nonholonomic constraints, and symmetry reduction can all be treated in a unified manner within the Lagrange--Dirac framework, both in finite- and infinite-dimensional settings.

\section{Conclusions}
In this paper, we have introduced a notion of Lagrange--Dirac structures on the tangent bundle $TQ$ and developed a unified geometric and variational framework for a broad class of mechanical systems in terms of Lagrange--Dirac dynamical systems. 
\medskip

The main results are summarized as follows:

\begin{itemize}

\item We introduced a Dirac structure on $TQ$, referred to as a Lagrange--Dirac structure, determined by a presymplectic form $\Omega_{L}$ together with a constraint distribution on $Q$. This structure provides a geometric formulation of mechanical systems with nonholonomic constraints and degenerate Lagrangians.

\item For the regular case, we explore the constrained Lagrange--Dirac structure by splitting the motion on $TQ$ into the admissible subbundle and complementary directions under the normality condition of the constraint distribution, and extract the constrained dynamics through the constrained Lagrange--Dirac structure on the admissible subbundle. 

\item Along the flow of the vector field on the admissible subbundle, we showed that the Lagrange--Dirac structure is preserved up to a gauge transformation generated by an exact two-form, revealing a gauge covariance property of the associated dynamics.

\item Associated with Lagrange--Dirac dynamical systems, we established new variational principles on $TQ$: the Lagrange--Dirac principle for unconstrained systems and the Lagrange--d'Alembert--Dirac principle for nonholonomic systems, both applicable to degenerate Lagrangians. The resulting first-order Lagrange--d'Alembert--Dirac equations recover the classical Euler--Lagrange equations in the unconstrained regular case and the Lagrange--d'Alembert equations in the constrained regular case.

\item Through various examples, we demonstrate that phase-space formulations, nonholonomic dynamics, degenerate Lagrangian systems, symmetry reduction, and infinite-dimensional extensions can all be treated in a unified manner within the Lagrange--Dirac framework.

\end{itemize}

In future work, we plan to further develop and apply the proposed Lagrange--Dirac formulation in several directions.

\begin{itemize}

\item We will study reduction under a free and proper action of a Lie group $G$ on $Q$, where the associated principal bundle $Q \to Q/G$ allows the construction of a reduced Lagrange--Dirac structure
\[
D_L^{/G} \subset (TTQ \oplus T^{\ast}TQ)/G
\]
over the reduced bundle $TQ/G$; see also \cite{YoMa2009}.

\item We will extend these reduction theories to the broader setting of Dirac anchored vector bundles developed in \cite{CeMaRaYo2009}.

\item In connection with Dirac's theory of constraints, we will investigate the relation with the Gotay--Nester approach for determining the final constraint submanifold in the context of Lagrange--Dirac dynamical systems; see \cite{GoNe1979, GoNe1980, CeEtFe2014}.

\item We will also study discrete analogues of Lagrange--Dirac structures and the associated Lagrange--Dirac dynamical systems, building on the discrete frameworks developed in \cite{PeYo2025, GBYo2018b}.

\end{itemize}

\paragraph{Acknowledgements.} 
I am deeply grateful to Tudor Ratiu and Fran\c{c}ois Gay-Balmaz for many insightful discussions and valuable suggestions. In particular, this research originated from fruitful discussions and collaborations with Jerry Marsden and Hern\'an Cendra. I dedicate this paper to their memory and to the honor of Hern\'an Cendra.
\medskip

This work was partially supported by JST CREST (JPMJCR24Q5), JSPS Grant-in-Aid for Scientific Research (22K03443), and Waseda University Grants for Special Research Projects (2026C-093).

\end{document}